\documentclass[12pt]{article}%
\usepackage{amssymb}
\usepackage{amsfonts}
\usepackage{amsmath}
\usepackage{amsthm}
\usepackage{booktabs}
\usepackage{float}
\usepackage[nohead]{geometry}
\usepackage[singlespacing]{setspace}
\usepackage[bottom]{footmisc}
\usepackage{indentfirst}
\usepackage{endnotes}
\usepackage{graphicx}%
\usepackage{rotating}
\RequirePackage[OT1]{fontenc} \RequirePackage{amsthm,amsmath}
\RequirePackage[numbers]{natbib}
\RequirePackage[colorlinks,citecolor=blue,urlcolor=blue]{hyperref}
\RequirePackage{hypernat} \makeatletter
\newcommand{\field}[1]{\mathbb{#1}}

\newcommand{\N}{\field{N}}

\newcommand{\E}{\field{E}}

\theoremstyle{example} \theoremstyle{remark} \theoremstyle{lemma}
\theoremstyle{definition} \theoremstyle{corol}
\theoremstyle{proposition} \theoremstyle{condition}
\theoremstyle{assumption}
\newtheorem{assumption}{\n{Assumption}}[section]
\newtheorem{theorem}{\n{Theorem}}[section]

\newtheorem{remark}{\n{Remark}}[section]
\newtheorem{lemma}{\n{Lemma}}[section]

\newtheorem{proposition}{\n{Proposition}}[section]

\font\n=cmcsc10

\makeatletter
\def\@biblabel#1{\hspace*{-\labelsep}}
\makeatother 
\newcommand{\rmnum}[1]{\romannumeral #1}
\newcommand{\Rmnum}[1]{\expandafter\romannumeral #1}

\begin{document}
\title{\Large On the Coverage Bound Problem of Empirical Likelihood Methods For Time Series}
\author{Xianyang Zhang\thanks{Address: Department of Statistics, University of Missouri-Columbia, Columbia, MO 65211, USA. E-mail: zhangxiany@missouri.edu. Tel:
+1 (573) 882-4455. Fax: +1 (573) 884-5524.} and Xiaofeng Shao
\medskip\\{\normalsize University of Missouri-Columbia and University of Illinois at Urbana-Champaign}} \maketitle
\sloppy%
\strut \onehalfspacing \small \textbf{Abstract} \strut  The upper
bounds on the coverage probabilities of the confidence regions based on blockwise empirical
likelihood [Kitamura (1997)] and nonstandard expansive empirical likelihood [Nordman et al. (2013)] methods for time series data
are investigated via studying the probability for the violation of
the convex hull constraint. The large sample bounds  are derived on the basis of the pivotal
limit of the blockwise empirical log-likelihood ratio obtained under the
fixed-$b$ asymptotics, which  has been
recently shown to provide a more accurate approximation to the
finite sample distribution than the conventional $\chi^2$
approximation. Our theoretical and numerical findings suggest that both the finite sample and large sample upper bounds for coverage
probabilities are strictly less than one and the blockwise empirical likelihood confidence region can exhibit serious undercoverage when (i) the dimension of moment conditions is moderate or large; (ii) the time series dependence is positively strong; or (iii) the block size is large relative to sample size.  A similar finite sample coverage problem occurs for the nonstandard expansive empirical likelihood.
To alleviate the coverage bound problem, we  propose to penalize both
empirical likelihood methods by relaxing the convex hull constraint.  Numerical simulations and data illustration
 demonstrate the effectiveness of our proposed remedies in terms of
delivering confidence sets with more accurate coverage.



\section{Introduction}
Empirical likelihood [EL, Owen (1988; 1990)] is a nonparametric
methodology for deriving estimates and confidence sets for unknown
parameters, which shares some of desirable properties of parametric
likelihood [see DiCiccio et al. (1991); Chen and Cui (2006)]. Due to
its  effectiveness and flexibility, it has been advanced to many
branches in statistics, such as regression models, time series,  and
censored data, among others; see Owen (2001) for a nice book-length
treatment of the subject.

The EL-based confidence sets inherit some nice features from their
parametric likelihood counterparts, but there is a finite sample
upper bound for the coverage of the EL ratio
confidence region [see Owen (2001, page 209); Tsao (2004)] due to
convex hull constraint, which may limit its applicability and make
it less appealing. For example, the EL confidence
region for the mean of a random sample are nested within the convex
hull of the data and their coverage level is necessarily smaller
than that of the convex hull itself. The upper bound can be much
smaller than nominal coverage level $1-\alpha$ in the small sample
and multidimensional situations. Following the terminology in Tsao
and Wu (2013a), the finite sample coverage bound problem is due to
the mismatch between the domain of the EL and the parameter space,
so it is also called a {\it mismatch problem}. There has been a few
recent proposals to alleviate or resolve the mismatch problem, see
e.g., the adjusted EL [Chen et al. (2008); Emerson
and Owen (2009); Liu and Chen (2010); Chen and Huang (2012)], the
penalized EL [Bartolucci (2007); Lahiri and
Mukhopadhyay (2012)] and the domain expansion approach [Tsao and Wu
(2013a; 2013b)]. However all these works deal with independent
estimation equations, and their direct applicability to the
important time series case is not clear.

In this article, our interest concerns the coverage bound problems
for EL methods tailed to the stationary and weakly dependent time
series. Although many variants have been proposed to extend the EL
to the time series setting [see Nordman and Lahiri (2013) for a
recent review], it seems that no investigation has been conducted
regarding the coverage bound problem, which is expected to exist but
its impact in the time series setting is unknown. We focus on two EL
methods: blockwise EL (BEL, hereafter) proposed by Kitamura (1997)
and nonstandard expansive BEL (EBEL, hereafter) recently proposed by
Nordman et al. (2013). The BEL applies the EL to the blockwise
averaged moment conditions to accommodate the dependence in time
series nonparametrically and it possesses a number of useful
properties of EL, such as Wilks' theorem. In Kitamura (1997), the
limiting $\chi^2$ distribution for the empirical log-likelihood
ratio (up to a multiplicative constant) was shown under the
traditional small-$b$ asymptotics, in which $b$,  the fraction of
block size relative to sample size, goes to zero as sample size
$n\rightarrow+\infty$.
Adopting the fixed-$b$ asymptotics [Kiefer and Vogelsang (2005)], in
which $b\in (0,1)$ is held fixed as $n\rightarrow+\infty$, Zhang and
Shao (2014) derived the pivotal limit of the empirical
log-likelihood ratio at the true parameter value, and used that as
the basis for confidence region construction. The pivotal limit
depends on $b$ and the simulations show that the fixed-$b$ based
confidence set has more accurate coverage than the small-$b$
counterpart, indicating that the approximation by the fixed-$b$
pivotal limit is more accurate than the small-$b$ counterpart (i.e.,
$\chi^2$). 

Since this paper is related to our previous work in Zhang and Shao (2014), it pays to highlight the difference. 
The focus of this paper is rather different from that of
Zhang and Shao (2014), and we investigate the coverage upper bound
problem of the block-based EL methods for time series. The technique
we used to derive the large sample bound, which depends on $b$, is
completely different from the one involved in the derivation of the
fixed-$b$ limit of the EL ratio statistic in Zhang and Shao (2014).
The main contribution of the current paper is (\rmnum{1}) to
identify the coverage bound problem for block-based EL methods in
time series setting and study the factors (e.g., sample size, block
size, joint distribution of time series, form of moment conditions)
that determine its magnitude. The large sample bound we derive under
the fixed-$b$ asymptotics provides an approximation to its finite
sample counterpart and the approximation is accurate for large $n$;
(\rmnum{2}) to propose the penalized BEL and EBEL methods as
remedies of the coverage bound problem, and show their effectiveness
through theory and simulations.

Let $1-\beta_n$ denote  the probability that convex hull of the moment conditions at the true parameter value contains the origin as an interior point and it is a natural upper bound on the coverage probability of the BEL ratio confidence region (with any finite critical values) regardless of its confidence level. In Tsao (2004), a finite sample upper bound was derived for independent estimation equations and the EL method. Tsao's technique is tailored to the independent case, and seems not applicable to time series data.
The calculation of the finite sample bound in the dependent and BEL case is challenging since it depends on the sample size, block size, dimension and form of moment conditions as well as  the joint distribution of time series. To shed some light on the coverage bound $1-\beta_n$, we approximate $1-\beta_n$ by its large sample counterpart $1-\beta$, where $\beta$ is shown to be the probability that the pivotal limit (under fixed-$b$ asymptotics) equals to infinity.
We further provide an
analytical formula for $\beta$ as a function of $b$ in the case $k=1$, and derive an upper bound for $1-\beta$ in the case $k>1$, where $k$ denotes the dimension of moment conditions.
Interestingly, we discover that
$\beta=\beta(b)>0$ for any $b>0$ and $\beta$ can be quite large for fixed $b\in (0,1)$ if the dimension of moment conditions $k$ is moderately large. Compared to Tsao (2004) and Kitamura (1997),
the large sample bound problem (i.e., $\beta>0$) is a unique feature that is associated with BEL under the fixed-$b$ asymptotics and it does not occur under the traditional small-$b$ asymptotic approximation or for independent estimation equations. It is also worth pointing out that the large sample bound is always one for any $b>0$ under the small-$b$ asymptotics, and it provides an inaccurate approximation of the finite sample bound and could lead to an overoptimistic but misleading inference.
 In corroboration with our theoretical results, our simulations show that the finite sample coverage bound can deviate substantially from one when (1) the block size is large relative to sample size (i.e., $b$ is large); (2) the dimension of moment conditions is moderate or high; (3) the time series dependence is positively strong. In any one of these cases, constructing a confidence set of a conventional nominal level (say,  $95\%$ or $99\%$) is likely to lead to undercoverage.  Thus our finding represents a cautionary note on the recent (theoretical) extension of BEL in the high dimensional setting [see Chang et al.
(2013)], where the dimension of moment conditions can also grow to infinity as sample size $n$ grows to infinity.

The EBEL uses a sequence of nested blocks with growing sizes so no choice of block size is involved, and the empirical log-likelihood ratio at the true parameter value converges to a pivotal but nonstandard limit. Unlike BEL,  there is no large sample bound problem for EBEL as the probability that the pivotal limit of EBEL equals to infinity is zero.
However, the finite sample bound can be far below the nominal level as shown in our simulations and results in a severe undercoverage.
 To alleviate the finite sample undercoverage problem caused by the convex hull constraint,  we propose to  penalize BEL and EBEL by dropping the convex hull constraint. The penalized EL was first introduced by Bartolucci (2007) for the inference of the mean of independent and identically distributed (i.i.d) data, and our generalization to the time series context requires a nontrivial modification. In particular, we introduce a new normalization matrix that takes the dependence into account and derive the limit of log EL ratio at the true value under the fixed-$b$ asymptotics. Our numerical results in the supplementary material suggest that the fixed-$b$ asymptotics not only provides better approximation for the original BEL [see Zhang and Shao (2014)] but it also tends to provide better finite sample approximation for its penalized counterpart.
Our new penalized EL ratio test statistic can be viewed as an intermediate between the empirical log-likelihood ratio test statistic and the self-normalized score test statistic [see (\ref{eq:self-test})] with the tuning parameter in the penalization term determining the amount of relaxation of the convex null constraint. Our numerical results show the effectiveness of the two penalization based EL methods in terms of  delivering more accurate confidence sets for a range of tuning parameters.

A word on notation. Let $D[0,1]$ be the space of functions on
$[0,1]$ which are right-continuous and have left limits, endowed
with the Skorokhod topology [Billingsley (1999)]. Weak convergence
in $D[0,1]$ or more generally in the $\mathbb{R}^q$-valued function
space $D^q[0, 1]$ is denoted by $``\Rightarrow"$, where
$q\in\mathbb{N}$. Convergence in probability and convergence in
distribution are denoted by $``\rightarrow^p"$ and
$``\rightarrow^d"$ respectively. Let $\lfloor a\rfloor$ be the
integer part of $a\in\mathbb{R}$. The notation $N(v,\Sigma)$ is used
to denote the multivariate normal distribution with mean $v$ and
covariance $\Sigma$.

\section{BEL and EBEL}
Suppose we are interested in the inference of a $p$-dimensional parameter vector $\theta$,
which is identified by a set of moment conditions. Denote by
$\theta_0$ the true parameter of $\theta$ which is an interior point
of a compact parameter space $\Theta\subseteq \mathbb{R}^{p}$. Let
$\{z_t\}^{n}_{t=1}$ be a sequence of $\mathbb{R}^l$-valued
stationary time series and assume that the moment conditions
\begin{equation}\label{moment}
\E[f(z_t,\theta_0)]=0,\quad t=1,2,\dots,n,
\end{equation}
hold, where $f(z_t,\theta): \mathbb{R}^{l}\times \Theta\rightarrow
\mathbb{R}^k$ is a map which is differentiable with respect to
$\theta$ and $\text{rank}(E[\partial
f(z_t,\theta_0)/\partial\theta'])=p$ with $k\geq p$. To deal with
time series data, we consider the fully overlapping smoothed moment
condition [Kitamura (1997)] which is given by
$f_{tn}(\theta)=\frac{1}{m}\sum^{t+m-1}_{j=t}f(z_j,\theta)$ with
$t=1,2,\dots,n-m+1$ and $m=\lfloor nb \rfloor$ for $b\in(0,1)$. The
overlapping data blocking scheme aims to preserve the underlying
dependence among neighboring time observations. Consider the profile
empirical log-likelihood function based on the fully overlapping
smoothed moment conditions,
\begin{equation}
\mathcal{L}_n(\theta)=\sup\left\{\sum^{N}_{t=1}\log(\pi_t):
\pi_t\geq 0, \sum^{N}_{t=1}\pi_t=1,
\sum^{N}_{t=1}\pi_tf_{tn}(\theta)=0\right\},\quad N:=n-m+1.
\end{equation}
Standard Lagrange multiplier arguments imply that the maximum is attained when
$$\pi_t=\frac{1}{N\{1+\lambda'f_{tn}(\theta)\}},\quad \text{with} \quad \sum^{N}_{t=1}\frac{f_{tn}(\theta)}{1+\lambda'f_{tn}(\theta)}=0,$$
where $\lambda$ is the Lagrange multiplier. By duality, the empirical log-likelihood ratio function (up to a multiplicative constant) is given by
\begin{equation}
elr(\theta)=\frac{2}{nb}\max_{\lambda\in
\mathbb{R}^k}\sum^{N}_{t=1}\log(1+\lambda'f_{tn}(\theta)),\quad
\theta\in\Theta.
\end{equation}
Under the traditional small-$b$ asymptotics, i.e., $nb^2+1/(nb)\rightarrow0$ as $n\rightarrow \infty$,
and suitable weak dependence assumptions [Kitamura (1997); also see Theorem 1 of Nordman and Lahiri (2013)], it can be shown that
\begin{equation}
elr(\theta_0)\rightarrow^d \chi^2_k.
\end{equation}
As pointed out by Nordman et al. (2013), the coverage accuracy of BEL can depend crucially
on the block length $m=\lfloor nb \rfloor$ and appropriate choices can vary with respect to the joint distribution of the series.
To capture the choice of block length in the asymptotics, Zhang and Shao (2014) adopted the fixed-$b$ approach
proposed by Kiefer and Vogelsang (2005) in the context of heteroscedasticity-autocorrelation
robust (HAR) testing and derived the nonstandard limit of $elr(\theta_0)$ under the fixed-$b$ asymptotics.
To proceed, we make the following assumption which can be verified under
suitable moment and weak dependence assumptions on $f(z_j,\theta_0)$ [see e.g., Phillips (1987)].
\begin{assumption}\label{ass1}
Assume that $\sum^{\lfloor nr
\rfloor}_{j=1}f(z_j,\theta_0)/\sqrt{n}\Rightarrow \Lambda W_k(r)$
for $r\in [0,1]$, where
$\Lambda\Lambda'=\Omega=\sum^{+\infty}_{j=-\infty}\Gamma_j$ with
$\Gamma_j=\E f(z_{t+j},\theta_0)f(z_t,\theta_0)'$ and $W_k(r)$ is a
$k$-dimensional vector of independent standard Brownian motions.
\end{assumption}
Under Assumption \ref{ass1}, Zhang and Shao (2014) showed that when $n\rightarrow +\infty$ and $b$ is
held fixed,
\begin{equation}\label{fixed-b-limit}
elr(\theta_0)\rightarrow^d U_{el,k}(b):=\frac{2}{b}\max_{\lambda \in
\mathbb{R}^k}\int^{1-b}_{0}\log(1+\lambda'\{W_k(r+b)-W_k(r)\})dr,
\end{equation}
where we define $\log(x)=-\infty$ for $x\leq 0.$ The asymptotic distribution $U_{el,k}(b)$ is nonstandard yet pivotal for a given $b$,
and its critical values can be obtained via simulation or bootstrap. Given $b\in(0,1)$,
a $100(1-\alpha)\%$ confidence region for the parameter $\theta_0$
is then given by
\begin{equation}
CI(1-\alpha;b)=\left\{\theta\in\Theta:
\frac{elr(\theta)}{1-b}\leq u_{el,k}(b; 1-\alpha)\right\},
\end{equation}
where $u_{el,k}(b; 1-\alpha)$ denotes the $100(1-\alpha)\%$ quantile
of the distribution $P(U_{el,k}(b)/(1-b)\leq x)$. It was demonstrated in Zhang and Shao (2014)
that the confidence region based on the fixed-$b$ approximation has more accurate coverage than the traditional counterpart.
Our analysis in the next section reveals an interesting coverage upper bound problem associated with the fixed-$b$ approach in the BEL framework.
This result provides some insight on the use of fixed-$b$ based critical values as suggested in Zhang and Shao (2014).
It also sheds some light on the finite sample coverage bound problem that can occur as long as the BEL ratio statistic is used to
construct confidence region. Moreover, we propose a penalized version of the fixed-$b$ based BEL, which improves the finite sample performance of the method
in Zhang and Shao (2014).

To avoid the choice of block length and also improve the finite
sample coverage, Nordman et al. (2013) proposed a new version of BEL
which uses a nonstandard data-blocking rule. To describe their
approach, we let
$\tilde{f}_{tn}(\theta)=\frac{\omega(t/n)}{n}\sum^{t}_{j=1}f(z_j,\theta)$
for $t=1,2,\dots,n$, where $\omega(\cdot): [0,1]\rightarrow
[0,+\infty)$ denotes a nonnegative weight function. The block
collection, which constitutes a type of forward scan in the block
subsampling language of McElroy and Politis (2007), contains a data
block of every possible length for a given sample size $n$. It is worth noting that this nonstandard data-blocking rule bears
some resemblance to recursive estimation in the self-normalization approach of
Shao (2010). Following Nordman et al. (2013), we define the
EBEL ratio function as
\begin{equation}
\widetilde{elr}(\theta)=\frac{1}{n}\max_{\lambda\in
\mathbb{R}^k}\sum^{n}_{t=1}\log(1+ \lambda'\tilde{f}_{tn}(\theta)).
\end{equation}
For the smooth function model, Nordman et al. (2013) showed that
\begin{equation}\label{nord}
\widetilde{elr}(\theta_0)\rightarrow^d
U_{ebel,k}(\omega)=\max_{\lambda \in
\mathbb{R}^k}\int^{1}_{0}\log(1+\lambda'\omega(r)W_k(r))dr.
\end{equation}
The numerical studies in Nordman et al. (2013) indicate that the
EBEL generally exhibits comparable (or in some cases even better) coverage accuracy  than the BEL with
$\chi^2$ approximation and suitable block size. Though the fixed-$b$ based BEL and EBEL
provide improvement over the traditional $\chi^2$ based BEL, our
study in the next section reveals that both the fixed-$b$ based BEL
and EBEL can suffer seriously from the coverage upper bound problem in finite sample. To
the best of our knowledge, this is the first time that the coverage
upper bound problem is revealed  for EL methods in time series.


\section{Bounds on the coverage probabilities}\label{sec:bound}
\subsection{Large sample bounds}
In the framework of BEL, asymptotic theory is typically established under the small-$b$ asymptotics, where large sample bound problem does not occur as
the empirical log-likelihood ratio statistic converges to a $\chi^2$ limit. However, in finite sample, the coverage upper bound $1-\beta_n$ can deviate significantly from the unity.
To shed some light on the finite sample coverage bound, we derive a limiting upper bound on the coverage
probabilities of the BEL ratio confidence region based on the
fixed-$b$ limiting distribution given in (\ref{fixed-b-limit}).
The fixed-$b$ method adopted here reflects the coverage upper bound problem in the asymptotics, while
the original BEL under the small-$b$ asymptotics is somewhat ``over-optimistic'' as the
corresponding upper bound in the limit is always one regardless of what the finite sample
bound is. Define $D_k(r;b)=W_k(r+b)-W_k(r)$
and $\mathcal{A}=\mathcal{A}_b=\left\{\lambda\in \mathbb{R}^k:
\min_{r\in(0,1-b)}(1+\lambda'D_k(r;b))\geq 0 \right\}.$ Let
$t_k(r;b)=\frac{D_k(r;b)}{|D_k(r;b)|}\mathbf{I}\{|D_k(r;b)|>0\}$ be
the direction of $D_k(r;b)$ on the $k-1$ dimensional sphere
$\mathcal{S}^{k-1}$, where $|\cdot|$ denotes the Euclidian norm and
$\mathbf{I}\{\cdot\}$ denotes the indicator function. We first
present the following lemma regarding the unboundedness of
$\mathcal{A}$.
\begin{lemma}\label{lemma1}
Define the convex hull
$\mathcal{H}(D_k)=\{\sum^{s}_{j=1}\alpha_jD_k(r_j;b):
s\in\mathbb{N}, \alpha_j\geq 0, \sum^{s}_{j=1}\alpha_j=1,
r_j\in(0,1-b)\}$. Then the set $\mathcal{A}$ is unbounded if and
only if the origin is not an interior point of $\mathcal{H}(D_k)$.
\end{lemma}
From the proof of Lemma \ref{lemma1} (given in the appendix) and the fact that the components of $D_k(r;b)$ are linearly independent (with probability one),
we know $\{\mathcal{A}~\text{is unbounded}\}$ implies that
$\{U_{el,k}(b)=+\infty\}$. On the other hand, when
$U_{el,k}(b)=+\infty$, it is easy to see that $\mathcal{A}$ cannot
be bounded. Therefore we have $P(\mathcal{A}~\text{is
unbounded})=P(U_{el,k}(b)=+\infty)$. Let
$\mathcal{H}_n(\theta_0;b)=\{\sum^{N}_{t=1}\alpha_tf_{tn}(\theta_0):
\alpha_t\geq0, \sum^{N}_{t=1}\alpha_t=1\}$ and denote
by $\mathcal{H}_n^o(\theta_0;b)$ the interior of $\mathcal{H}_n(\theta_0;b)$. By Lemma
\ref{lemma1} and strong approximation, we have for large $n$,
$P(\text{the origin is not contained in}~\mathcal{H}_n^o(\theta_0;b))\approx
P(\mathcal{A}~\text{is unbounded})=P(U_{el,k}(b)=+\infty).$

It was conjectured in Zhang and Shao (2014) that
$P(\mathcal{A}~\text{is unbounded})>0$, which implies that
$P(U_{el,k}(b)=+\infty)>0$. In what follows,
we give an affirmative answer to this conjecture and provide an
explicit formula for the probability $P(\mathcal{A}~\text{is
unbounded})$ when $k=1.$
Notice that for $k=1$, we must have
$\{\mathcal{A}~\text{is unbounded}\}=\{D_1(r;b)\geq 0, \forall~r\in
(0,1-b]\}\cup\{D_1(r;b)\leq 0, \forall~r\in (0,1-b]\}.$ By the symmetry
of Wiener process, we have $P(\mathcal{A}~\text{is
unbounded})=2P(D_1(r;b)\geq 0, \forall~r\in (0,1-b]).$  For $\beta>0$,
we let $\phi_{\beta}(\cdot)=\phi(\cdot/\sqrt{\beta})/\sqrt{\beta}$
with $\phi(x)=\frac{1}{\sqrt{2\pi}}\exp(-x^2/2)$ being the standard
normal density. For two vectors $x=(x_1,x_2,\dots,x_{L})'$ and $y=(y_1,y_2,\dots,y_{L})'$
of real numbers with $L\in \N$, define the matrix $\mathcal{Q}_{\beta,L}(x,y)=(\phi_{\beta}(x_i-y_j))^{L}_{i,j=1}$.
Let $q_{\beta,L}(x,y)$ be the determinant of $\mathcal{Q}_{\beta,L}(x,y)$. For a vector $x=(x_1,x_2,\dots,x_{L})'$,
denote by $x_{s_1:s_2}=(x_{s_1},x_{s_1+1},\dots,x_{s_2})'$ the subvector of $x$ for $1\leq s_1\leq s_2\leq L.$
Using similar arguments as in Shepp (1971) [also see Karlin and Mcgregor (1959)],
we prove the following result.
\begin{theorem}\label{bound-k=1}
If $L=1/b$ is a positive integer, we have
\begin{eqnarray}
\label{eq:prob1} P(D_1(r;b)\geq 0,
\forall~r\in(0,1-b])=\int_{0=x_1<x_2<x_3<\cdots<x_{L+1}}q_{1,L}(x_{1:L},x_{2:(L+1)})dx_2dx_3\cdots
dx_{L+1},
\end{eqnarray}
where $x=(x_1,\dots,x_{L+1})'$. If $bL+\tau=1$ with $L$ being a positive integer and $0<\tau<b$, we
have
\begin{eqnarray}
\label{eq:prob2}
\begin{split}
 &P(D_1(r;b)\geq 0, \forall~r\in(0,1-b])
\\=&\int\cdots\int_Sq_{\xi,L+1}(x,y)q_{1-\xi,L}(x_{2:(L+1)},y_{1:L})dy_1dx_2dy_2\cdots
dx_{L+1}dy_{L+1},\quad 0<\xi=\tau/b<1,
\end{split}
\end{eqnarray}
where $x=(x_1,\dots,x_{L+1})'$ with $x_1=0$, $y=(y_1,\dots,y_{L+1})'$, and the integral is over the set
$S:=\{(y_1,x_2,y_2,\cdots,x_{L+1},y_{L+1})\in\mathbb{R}^{2L+1}:
0<x_2<\cdots <x_{L+1}, y_1<y_2<\cdots<y_{L+1}\}$.
\end{theorem}
Theorem \ref{bound-k=1} provides an exact formula for the
probability $P(\mathcal{A}~\text{is unbounded})$ when $k=1$.
The probability can be manually calculated when $L$ is small.
In particular, if $b=1/2$ (i.e., $L=2$),  we have
\begin{align*}
P(\mathcal{A}~\text{is
unbounded})=&2\int_{0<x_2<x_3}\{\phi(-x_2)\phi(x_2-x_3)-\phi(-x_3)\phi(0)\}dx_2dx_3
\\=&2\left\{\Phi^2(0)-\phi(0)\int^{0}_{-\infty}\Phi(x)dx\right\}=0.18169,
\end{align*}
where $\Phi(\cdot)$ denotes the distribution function of the
standard normal random variable. When $b=1/3$ (i.e., $L=3$), direct calculation
yields that,
\begin{align*}
&P(\mathcal{A}~\text{is
unbounded})=2\left\{\Phi^3(0)+\frac{\phi^2(0)}{4}+\int_{0<x_2<x_3}\phi(-x_3)\phi(x_3-x_2)\Phi(x_2-x_3)dx_2dx_3\right\}
\\
-&2\left\{\int_{0<x_2<x_3}\phi^2(x_3-x_2)\Phi(-x_3)dx_2dx_3+\phi^2(0)\Phi(0)+\int_{0<x_2<x_3}\phi(-x_2)\phi(0)\Phi(x_2-x_3)dx_2dx_3\right\}\\
=&2\bigg\{\frac{1}{8}+\frac{\phi^2(0)}{4}+\int_{-\infty}^{0}\phi(u)\Phi^2(u)du+\frac{\phi^2(0)}{2\sqrt{2}}-\frac{1}{\sqrt{4\pi}}\int_{0}^{+\infty}\Phi(-x_3)\Phi(\sqrt{2}x_3)dx_3
-\phi^2(0)\bigg\}=0.03635.
\end{align*}
The calculation for larger $L$ is still possible but is more involved. An alternative way is to approximate
the probabilities in (\ref{eq:prob1}) and (\ref{eq:prob2}) using Monte Carlo simulation; see Table~\ref{tab:bound-1} and Figure \ref{fig:bound}.
Utilizing the result in Theorem \ref{bound-k=1}, we can derive a (conservative) upper bound on
 $P(\mathcal{A}~\mbox{is bounded})$ (i.e., $1-\beta$) in the multidimensional case.
For $k>1$, we let $D_k^{(j)}(r;b)$ be the $j$th element of $D_k(r;b)$ and
$\mathcal{V}_j=\{D_k^{(j)}(r;b)\geq 0, \forall~r\in
(0,1-b]\}\cup\{D_k^{(j)}(r;b)\leq 0, \forall~r\in (0,1-b]\}$ with $1\leq
j\leq k.$ By the independence of the components of $D_k(r;b)$, it is easy to derive that
\begin{align*}
P(\mathcal{A}~\text{is unbounded})\geq &
P\left(\cup_{j=1}^{k}\mathcal{V}_j\right)=1-P\left(\cap_{j=1}^{k}\mathcal{V}_j^c\right)=1-P^k\left(\mathcal{V}_1^c\right)
\\=&1-\left\{1-2P\left(D_k^{(j)}(r;b)\geq 0, \forall~r\in (0,1-b]\right)\right\}^k.
\end{align*}
Therefore, we obtain the following result.

\begin{proposition}\label{bound-k>1}
When $L=1/b$ is a positive integer, we have
\begin{eqnarray}
\label{eq:prob3}
\begin{split}
P(U_{el,k}(b)<+\infty)\leq &
\left(1-2\int_{0=x_1<x_2<x_3<\cdots<x_{L+1}}q_{1,L}(x_{1:L},x_{2:(L+1)})dx_2dx_3\cdots
dx_{L+1}\right)^k,
\end{split}
\end{eqnarray}
where $x=(x_1,\dots,x_{L+1})'$. When $bL+\tau=1$ with $L$ being a positive integer and $0<\tau<b$,
we have
\begin{eqnarray}
\label{eq:prob4}
\begin{split}
& P(U_{el,k}(b)<+\infty)\\ \leq & \left(1-2\int\cdots\int_S
q_{\xi,L+1}(x,y)q_{1-\xi,L}(x_{2:(L+1)},y_{1:L})dy_1dx_2dy_2\cdots
dx_{L+1}dy_{L+1}\right)^k,\quad 0<\xi=\tau/b<1,
\end{split}
\end{eqnarray}
where $x=(x_1,\dots,x_{L+1})'$ with $x_1=0$, $y=(y_1,\dots,y_{L+1})'$, and the integral is over the set
$S:=\{(y_1,x_2,y_2,\cdots,x_{L+1},y_{L+1})\in\mathbb{R}^{2L+1}:
0<x_2<\cdots <x_{L+1}, y_1<y_2<\cdots<y_{L+1}\}$. When $k=1$, the inequality becomes equality in
(\ref{eq:prob3}) and (\ref{eq:prob4}).
\end{proposition}
If the (asymptotic) critical value based on the fixed-$b$ pivotal limit $U_{el,k}(b)$ is used to construct a
$100(1-\alpha)\%$ confidence region, then
the following several cases can occur:

(1) $P(U_{el,k}(b)<\infty)=1-\beta\le 1-\alpha$, then the fixed-$b$
based critical value is $\infty$. In this case, it is impossible to
construct a meaningful confidence region as $\{\theta\in\Theta |
elr(\theta)\le \infty\}=\Theta$. Note that in the case $k=1$, the
value of $\beta$ is known but in the case $k=2$ or higher, only an
upper bound for $1-\beta$ is provided in the above proposition. Thus
if the upper bound is no greater than $1-\alpha$, then we are not
able to construct a sensible confidence region based on fixed-$b$
critical values.

(2) $P(U_{el,k}(b)<\infty)=1-\beta>1-\alpha$, then the fixed-$b$ based critical value is finite. The
$100(1-\alpha)\%$ quantile of the distribution of $U_{el,k}(b)/(1-b)$ (i.e., $u_{el,k}(b;1-\alpha)$)
is $100\gamma\%$ quantile of the conditional distribution $P(U_{el,k}(b)/(1-b)\leq x|U_{el,k}(b)<+\infty)$, where $\gamma=\frac{1-\alpha}{1-\beta}$. In the simulation experiment of Zhang and Shao (2014), the $100(1-\alpha)\%$ quantile of the conditional distribution is used as the critical value.
 Note that the largest $b$ considered in the latter paper  is $0.2$, which corresponds to $\beta\approx 1-0.9985=0.0015$ when $k=1$ and $\beta\approx 1-0.9872=0.0128$ when $k=2$, as seen from Table~\ref{tab:bound-1}. This suggests that the critical values used in Zhang and Shao (2014) are wrong, but not by a lot.

(3) In the event that $u_{el,k}(b;1-\alpha)$ is finite, which occurs in case (2) above or
when the $\chi^2$-based critical values are used,
$$P(\theta_0\in CI(1-\alpha;b))\le P(\mbox{the origin is contained in}~\mathcal{H}_n^o(\theta_0;b))=1-\beta_n,$$ which is a finite sample bound. The quantity $\beta_n$ depends on joint distribution of time series, the form of $f$, block size and sample size, so is in general difficult to calculate. We present some numerical results on $\beta_n$ in Section~\ref{subsec:bounds} below. If $1-\beta_n\le 1-\alpha$, then the confidence region is bound to undercover and the amount of undercoverage gets severe when $\beta_n$ is farther from zero.


Proposition \ref{bound-k>1} shows that for any fixed $b\in(0,1)$,
the bound decays exponentially to zero as the dimension $k$ grows.
This result suggests that caution needs to be taken in the recent
extension of the BEL to the high dimensional setting [see Chang et
al. (2013)], where the dimension of moment condition $k$ can grow
with respect to sample size $n$. In the latter paper, the small-$b$
asymptotics is adopted, and no discussion on such coverage bound
issue (either finite sample or large sample) seems provided. It
would be interesting to extend the  fixed-$b$ asymptotic approach to
the BEL in the high dimensional setting and we leave it for future
investigation.

The large sample bound on the coverage probabilities depends
crucially on how the smoothed moment conditions are constructed. By
Lemma 1 of Nordman et al. (2013), we know that for EBEL, the set
$\mathcal{A}_{\omega}=\{\lambda\in \mathbb{R}^k: \min_{r\in(0,1)}(1
+ \lambda'\omega(r)W_k(r))\geq 0 \}$ is bounded with probability
one, which implies that $P(U_{ebel,k}(\omega)<\infty)=1$. Thus for
large sample, no upper bound problem occurs for the EBEL. However,
the numerical results in Table \ref{tab:bound-EBEL} show that the
finite sample bounds on the coverage probabilities of the EBEL ratio
confidence regions can be significantly lower than one, which
indicates that the convergence of the EBEL ratio statistic
$\widetilde{elr}(\theta_0)$ to its limit $U_{ebel,k}(\omega)$ is in
fact slow and there can be substantial undercoverage associated with
EBEL-based confidence region in any one of the following three
cases: (1) the dependence is positively strong; (2)  sample size $n$
is small; (3) $k$ is moderate, say $k\ge 3$.
\begin{remark}
{\rm The convex hull constraint is related to the underlying
distance measure between $\pi=(\pi_1,\dots,\pi_N)$ and
$(1/N,\dots,1/N)$ in EL. If one considers alternative nonparametric
likelihood such as the Euclidean likelihood or more generally,
members of the Cressie-Read power divergence family of
discrepancies, then the origin is allowed to get outside of the
convex hull of the smoothed moment conditions  as long as the weights are allowed to be negative.
No coverage upper bound problem occurs for these alternative nonparametric likelihoods, but since EL has certain optimality property [Kitamura (2006); Kitamura et al. (2013)], it is still
a worthwhile effort to seek remedies of the coverage bound problem based on EL.    }
\end{remark}

\subsection{Finite sample results on coverage bounds}
\label{subsec:bounds} To evaluate the upper bounds on the coverage
probabilities for BEL and EBEL, we simulate time series from the
AR(1) models with the AR(1) coefficient $\rho=-0.5, 0, 0.2, 0.5,
0.8$, and i.i.d standard
normal errors. The sample size $n$ is equal to
$50,100,500,1000,5000,+\infty$. We approximate the probability
$P(\mathcal{A}~\text{is bounded})$ by simulating independent Wiener
processes, where the Wiener process is approximated by the
normalized partial sum of 50,000 i.i.d standard normal random
variables and the number of Monte Carlo replications is 100,000.
When $k>1$, we simulate VAR(1) processes with the coefficient matrix $A_1=\rho I_k$ for
$\rho=-0.5, 0, 0.2, 0.5, 0.8$, and standard multivariate normal errors. Table \ref{tab:bound-1} summarizes
the upper bounds on the coverage probabilities for BEL with $b=1/L$
for $L=2,3,\dots,10,15$ and 20, and Table \ref{tab:bound-EBEL}
provides the finite sample upper bounds on the coverage
probabilities for EBEL. For BEL, it is seen from the table that the
upper bound on the coverage probability decreases as the block
size increases and the positive dependence strengthens. The bound in
the multidimensional case is lower than its counterpart in the
univariate case, which is consistent with our theoretical finding.
It is interesting to note that negative dependence (corresponding to
$\rho=-0.5$) tends to bring the upper bound higher. In practice, if
the dependence is expected to be positively strong, a large block
size is preferable. However, our result indicates that the
corresponding upper bound on the coverage probabilities will be
lower for larger block size. It is also worth noting that the upper
bounds on the coverage probabilities generally increase as the
sample size grows and the result in Proposition \ref{bound-k>1}
provides conservative  bounds on $1-\beta$ when $k=2$. For EBEL,
though its large sample bound is one, its finite sample bound can be
significantly lower than one as seen from Table
\ref{tab:bound-EBEL}. To further assess the impact of the dimensionality $k$,
we present the coverage upper bounds for $k=5,10,15,20,50$, and $L=2,3,\dots,20,30,40,50$ in Figure \ref{fig:bound}, where
data are generated from multivariate standard normal distribution with sample size $n=5000$. We observe that (i) as $k$ grows, a smaller $b$ (or larger $L$)
is required to deliver meaningful finite sample upper bounds (say, larger than nominal level); (ii)  the coverage upper bound for EBEL can be close to zero for
$k=15$ or larger. We expect that the bound can get worse when we increase the positive dependence in the observations.
Based on the numerical results for this
specific setting, we suggest special attention be paid to the
potential coverage bound problem for the following cases: (1) the
nominal level is close to one (such as 99\%); (2) the dimension of
moment conditions $k$ is moderate or high; (3) the (positive)
dependence is strong; (4) $b$ is large.

\section{Penalized BEL and EBEL}
\label{sec:remedy} The convex hull constraint
violation underlying the mismatch is well known in the EL literature
[see Owen (1990); Owen (2001)]. Various methods have been proposed
to bypass this constraint, such as the penalized EL [Bartolucci (2007); Lahiri and Mukhopadhyay (2012)], the
adjusted EL [Chen et al. (2008); Emerson and Owen
(2009); Liu and Chen (2010); Chen and Huang (2012)] and the extended EL
[Tsao and Wu (2013a; 2013b)]. Motivated by the theoretical findings
as well as the finite sample results in Section \ref{subsec:bounds},
we propose a remedy based on penalization to circumvent the coverage bound problem which leads to improved coverage accuracy under the fixed-$b$ asymptotics.


\subsection{Penalized BEL}\label{pel}
To overcome the convex hull constraint violation problem, Bartolucci
(2007) dropped the convex hull constraint in the formulation of EL for the mean of a random sample
and defined the likelihood by penalizing the unconstrained EL using
the Mahalanobis distance. Recently, Lahiri and Mukhopadhyay (2012)
introduced a modified version of Bartolucci's penalized EL (PEL) in
the mean case. Under the assumption that the observations are i.i.d
and the components of each observation are dependent, Lahiri and
Mukhopadhyay (2012) derived the asymptotic distributions of the PEL
ratio statistic in the high dimensional setting. Other variants of the PEL where
a penalty function is added to the standard EL are considered by Otsu (2007) for efficient
estimation in semiparametric models and Tang and Leng (2010) for consistent parameter estimation and variable selection in linear models.
In these two papers, they either penalize high dimensional parameters or
roughness of unknown nonparametric function,  and their PELs still
suffer from the same convex hull constraint violation problem as the standard EL. In what follows, we
shall consider a penalized version of the BEL ratio test statistic
in the moment condition models, which allows weak dependence within
the moment conditions and may be computed even when the origin does
not belong to the convex hull of the smoothed moment conditions.
Compared to existing penalization methods in the literature, our
method is different in three aspects. First, our method is designed
for dependent data where existing methods are only applicable to
independent moment conditions. Second, our theoretical result is
established under the fixed-$b$ asymptotics which is expected to
provide better approximation to the finite sample distribution. And we suggest the use of the fixed-$b$ based critical values that
 capture the choice of tuning parameters (also see the simulations in the supplementary material).
 Third, our formulation produces a new class of statistic between the
empirical log-likelihood ratio statistic and the self-normalized score statistic
which is of interest in their own right. To illustrate the idea, we
first consider the case $k=p$, i.e.,  the moment condition is
exactly identified (see Remark \ref{overidentified} for the general
overidentified case). Define the simplex
$\mathfrak{F}_N=\{\pi=(\pi_1,\dots,\pi_N): \pi_t\geq 0,
\sum^{N}_{t=1}\pi_t=1\}$ and the quadratic distance measure
$\delta_n(\mu):=\delta_{\Psi_n}(\mu)=\mu' \Psi^{-1}_n\mu$ for
$\mu\in \mathbb{R}^k$, where $\Psi_n\in\mathbb{R}^{k\times k}$ is an
invertible normalization matrix. Let
$\mu_{\pi}(\theta)=\sum^{N}_{t=1}\pi_t f_{tn}(\theta)$ with
$\pi=(\pi_1,\dots,\pi_N)\in\mathfrak{F}_N$. We consider the
penalized BEL (PBEL) as follows,
\begin{equation}
\mathcal{L}_{pbel,n}(\theta)=\max_{\pi\in\mathfrak{F}_N}\prod^{N}_{t=1}\pi_t\exp\left\{-\frac{n\tau}{2}\delta_n(\mu_{\pi}\left(\theta)\right)\right\}.
\end{equation}
The PBEL ratio test statistic is then defined as
\begin{align*}
elr_{pbel}(\theta)=&-\frac{2}{nb}\log\left\{N^{N}\mathcal{L}_{pbel,n}(\theta)\right\}
=\min_{\pi\in\mathfrak{F}_N}\left\{-\frac{2}{nb}\sum^{N}_{t=1}\log
(N\pi_t)+\frac{\tau}{b}\delta_n(\mu_{\pi}\left(\theta)\right)\right\}.
\end{align*}
Under the constraint that $\mu=\sum^{N}_{t=1}\pi_tf_{tn}(\theta)$,
it is not hard to derive that
\begin{equation*}
\pi_t=\frac{1}{N\left\{1+\lambda'(f_{tn}(\theta)-\mu)\right\}},\quad
\text{with}~~\sum^{N}_{t=1}\frac{f_{tn}(\theta)-\mu}{1+\lambda'(f_{tn}(\theta)-\mu)}=0,
\end{equation*}
by using standard Lagrange multiplier argument. Denote by
$\mathcal{H}_n(\theta;b)=\{\sum^{N}_{t=1}\pi_tf_{tn}(\theta):\pi\in\mathfrak{F}_N\}.$
Thus we deduce that
\begin{equation}\label{pen-equ0}
\begin{split}
elr_{pbel}(\theta)=\min_{\mu\in\mathcal{H}_n(\theta;b)}\left\{\frac{2}{nb}\max_{\lambda\in\mathbb{R}^k}\sum^{N}_{t=1}\log
\left\{1+\lambda'(f_{tn}(\theta)-\mu)\right\}+\frac{\tau}{b}\delta_n\left(\mu\right)\right\},
\end{split}
\end{equation}
where $\mu$ is minimized to balance the empirical log-likelihood ratio and the
penalty term.
\begin{proposition}\label{prop-1}
If the space spanned by $\{f_{tn}(\theta)\}^{N}_{t=1}$ is of $k$ dimension, we have
\begin{equation}\label{pen-equ01}
\begin{split}
elr_{pbel}(\theta)=&\min_{\mu\in\mathbb{R}^k}\left\{\frac{2}{nb}\max_{\lambda\in\mathbb{R}^k}\sum^{N}_{t=1}\log
\left\{1+\lambda'(f_{tn}(\theta)-\mu)\right\}+\frac{\tau}{b}\delta_n\left(\mu\right)\right\}.
\end{split}
\end{equation}
\end{proposition}
The condition that the space spanned by $\{f_{tn}(\theta)\}^{N}_{t=1}$ is of $k$ dimension is fairly mild because $k$ is fixed and $N$ grows with $n$.
Notice that the minimizer $\mu^*$ of (\ref{pen-equ01}) is necessarily contained in $\mathcal{H}_n(\theta;b)$, which implies that
the origin of $\mathbb{R}^k$ is contained in the convex hull of
$\{f_{tn}(\theta)-\mu^*\}^{N}_{t=1}$. In addition, since the empirical log-likelihood
ratio and the penalty term in (\ref{pen-equ01}) are both convex
functions of $\mu$, it is not hard to obtain $\mu^*$ in practice.
Let $\tau=c^*n$ with $c^*$ being a nonnegative constant which
controls the magnitude of the penalty term, and suppose that
$\Psi_n^{-1}\rightarrow ^d (\Lambda\Phi_k\Lambda')^{-1}$ as
$n\rightarrow +\infty,$ where $\Phi_k\in\mathbb{R}^{k\times k}$ is a
pivotal limit. For example, if we let $Q(\cdot,\cdot):
[0,1]^2\rightarrow \mathbb{R}$ be a positive semi-definite kernel,
then one possible choice of the normalization matrix $\Psi_n$ is
given by
\begin{equation}
\Psi_n(\hat{\theta}_n)=\frac{1}{n}\sum^{n}_{t=1}\sum^{n}_{j=1}Q(t/n,j/n)f(z_t,\hat{\theta}_n)f'(z_j,\hat{\theta}_n),
\end{equation}
where $\hat{\theta}_n$ is a preliminary estimator obtained by
solving the equation $\sum^{n}_{j=1}f(z_j,\theta)=0$. In practice,
one can choose $Q(r,s)=\kappa(r-s)$ with $\kappa(\cdot)$ being the
kernels used in the heteroskedasticity and autocorrelation
consistent (HAC) estimation, such as the Bartlett kernel or
the quadratic spectral kernel. Under appropriate conditions [see
e.g., Kiefer and Vogelsang (2005); Sun (2013)], it can be shown that
\begin{equation}
\Psi_n(\hat{\theta}_n)\rightarrow^d \Lambda
\int^{1}_{0}\int^{1}_{0}Q(r,s)dB_k(r)dB_k'(s)\Lambda':=\Lambda \Phi_k
\Lambda',
\end{equation}
where $\Phi_k=\int^{1}_{0}\int^{1}_{0}Q(r,s)dB_k(r)dB_k'(s)$ with $B_k(r)=W_k(r)-rW_k(1)$.
Therefore, under Assumption \ref{ass1}, we have
\begin{equation}
\begin{split}
elr_{pbel}(\theta_0)\rightarrow^d U_{pbel,k}(b)=&
\min_{\mu\in
\mathcal{H}(b)}\left\{\frac{2}{b}\max_{\lambda\in\mathbb{R}^k}\int^{1-b}_{0}\log
\left\{1+\lambda'(\Lambda
D_k(r;b)/b-\mu)\right\}dr+\frac{c^*}{b}\mu'(\Lambda\Phi_k\Lambda')^{-1}\mu\right\},
\end{split}
\end{equation}
where $\mathcal{H}(b)$ denotes the convex hull of
$\{\Lambda D_k(r;b)/b: r\in (0,1-b)\}$. Note that when $\mu$ is outside the convex hull of
$\{\Lambda D_k(r;b)/b: r\in (0,1-b)\}$, the separating hyperplane theorem [see e.g. Section 11 of Rockafellar (1970)] implies that $\max_{\lambda\in\mathbb{R}^k}\int^{1-b}_{0}\log
\left\{1+\lambda'(\Lambda D_k(r;b)/b-\mu)\right\}dr=+\infty$. Thus we have the simplified expression,
\begin{equation}
\begin{split}
U_{pbel,k}(b)=&
\min_{\mu\in\mathbb{R}^k}\left\{\frac{2}{b}\max_{\lambda\in\mathbb{R}^k}\int^{1-b}_{0}\log
\left\{1+\lambda'(\Lambda
D_k(r;b)/b-\mu)\right\}dr+\frac{c^*}{b}\mu'(\Lambda\Phi_k\Lambda')^{-1}\mu\right\}
\\=&
\min_{\tilde{\mu}\in\mathbb{R}^k}\left\{\frac{2}{b}\max_{\tilde{\lambda}\in\mathbb{R}^k}\int^{1-b}_{0}\log
\left\{1+\tilde{\lambda}'(
D_k(r;b)/b-\tilde{\mu})\right\}dr+\frac{c^*}{b}\tilde{\mu}'\Phi^{-1}_k\tilde{\mu}\right\},
\end{split}
\end{equation}
where $\tilde{\lambda}=\Lambda'\lambda$ and $\tilde{\mu}=\Lambda^{-1}\mu.$ Notice that the limiting
distribution $U_{pbel,k}(b)$ is pivotal and its critical values can
be simulated by approximating the Brownian motion with the
standardized/normalized partial sum of i.i.d standard normal random variables.
As to the pivotal limit $U_{pbel,k}(b)$, we have the following
result.
\begin{proposition}\label{boundness}
For $b\in (0,1)$ and $c^*>0$, $P(U_{pbel,k}(b)<\infty)=1$.
\end{proposition}
Thus compared to the BEL, the PBEL is
well defined and does not suffer from the convex hull violation
problem  in both large sample and  finite sample cases, though it involves the
choice of additional tuning parameters such as $c^*$ and $\Psi_n$.

Note that when $c^*=\infty$, we have $\mu^*=0$
and the PBEL ratio statistic reduces to the BEL ratio statistic. 
On the other hand,
\begin{equation}
\begin{split}
elr_{pbel}(\theta)=c^*\min_{\mu\in\mathbb{R}^k}\left\{\frac{2}{c^*nb}\max_{\lambda\in\mathbb{R}^k}\sum^{N}_{t=1}\log
\left\{1+\lambda'(f_{tn}(\theta)-\mu)\right\}+\frac{n}{b}\delta_n\left(\mu\right)\right\},
\end{split}
\end{equation}
and
$$\max_{\lambda\in\mathbb{R}^k}\sum^{N}_{t=1}\log
\left\{1+\lambda'\left(f_{tn}(\theta)-\frac{1}{N}\sum^{N}_{t=1}f_{tn}(\theta)\right)\right\}=0.$$
Thus for small $c^*$, the minimizer $\mu^*$ should be close to
$\sum^{N}_{t=1}f_{tn}(\theta)/N$. In this case, the penalty term
dominates and the PBEL ratio statistic evaluated at the true parameter value behaves like the self-normalized score
statistic which is defined as
\begin{equation}\label{eq:self-test}
\mathcal{S}_n(\theta_0)=n\delta_{n}\left(\sum^{N}_{t=1}f_{tn}(\theta_0)/N\right)=n\left(\sum^{N}_{t=1}f_{tn}(\theta_0)/N\right)'\Psi_n^{-1}(\hat{\theta}_n)\left(\sum^{N}_{t=1}f_{tn}(\theta_0)/N\right).
\end{equation}
We call $\mathcal{S}_n(\theta_0)$ the self-normalized score statistics as $f_{tn}(\theta)$ plays the role of the score in likelihood-based inference and the
self-normalizer $\Psi_n(\hat{\theta}_n)$ is an inconsistent estimator of the asymptotic variance matrix $\Omega$ in the spirit of the self-normalized approach of Shao (2010).
Therefore, based on the quadratic distance measure, the penalized
BEL ratio statistic can be viewed as a combination of the BEL ratio
statistic and the self-normalized score statistic.

\begin{remark}\label{overidentified}
{\rm When the moment condition is overidentified (i.e. $k>p$), we
shall consider the normalization matrix $\Psi_n=\Psi_n(\hat{\theta}_n)$ with $\hat{\theta}_n$ being a
preliminary estimator such as the one-step GMM estimator with the
weighting matrix $W_n\rightarrow^p W_0$, where $W_0$ is a $k\times k$
positive definite matrix. To illustrate the idea, define $G_t(\theta)=\frac{1}{n}\sum^{t}_{j=1}\partial
f(z_j,\theta)/\partial \theta'$ and $G_0=\E [G_n(\theta_0)]$. Let
$\hat{u}_j=(G_n'(\hat{\theta}_n)W_nG_n(\hat{\theta}_n))^{-1}G_n'(\hat{\theta}_n)W_nf(z_j,\hat{\theta}_n).$
Consider the normalization matrix
$\Psi_n(\hat{\theta}_n)=\frac{1}{n}\sum^{n}_{t=1}\sum^{n}_{j=1}Q(t/n,j/n)\hat{u}_t\hat{u}_j'$.
Under suitable conditions [see Kiefer and Vogelsang (2005)], it can
be deduced that $\Psi_n(\hat{\theta}_n)\rightarrow^d \Delta
\int^{1}_{0}\int^{1}_{0}Q(r,s)dB_p(r)dB_p(s)\Delta'$, where
$\Delta\in \mathbb{R}^{p\times p}$ is an
invertible matrix such that
$\Delta\Delta'=(G_0'W_0G_0)^{-1}G_0'W_0\Omega
W_0G_0(G_0'W_0G_0)^{-1}.$ In this case, the PBEL ratio test
statistic can be defined as,
\begin{align*}
elr_{pbel}(\theta)=&\min_{\mu\in\mathbb{R}^p}\left\{\frac{2}{nb}\max_{\lambda\in\mathbb{R}^p}\sum^{N}_{t=1}\log
\left\{1+\lambda'\left(g_{tn}(\theta)-\mu\right)\right\}+\frac{\tau}{b}\mu'\Psi_n^{-1}(\hat{\theta}_n)\mu\right\},
\end{align*}
where $g_{tn}(\theta)=(G_n'(\hat{\theta}_n)W_nG_n(\hat{\theta}_n))^{-1}G_n'(\hat{\theta}_n)W_nf_{tn}(\theta)$
is the transformed smooth moment condition. Following the arguments above, it can be shown that $elr_{pbel}(\theta_0)$ admits the same pivotal limit,
\begin{equation}
\begin{split}
elr_{pbel}(\theta_0)\rightarrow^d {U}_{pbel,p}(b)=&
\min_{\tilde{\mu}\in\mathbb{R}^p}\left\{\frac{2}{b}\max_{\tilde{\lambda}\in\mathbb{R}^p}\int^{1-b}_{0}\log
\left\{1+\tilde{\lambda}'(
D_p(r;b)/b-\tilde{\mu})\right\}dr+\frac{c^*}{b}\tilde{\mu}'\Phi^{-1}_p\tilde{\mu}\right\}.
\end{split}
\end{equation}}
\end{remark}


\subsection{Penalized EBEL}\label{pebel}
As demonstrated in Section \ref{subsec:bounds}, the EBEL suffers
seriously from the convex hull violation problem in finite sample.
To deal with the convex hull condition, we introduce the penalized
version of the EBEL (PEBEL) which is shown to provide significant
finite sample improvement in Section \ref{sec:numerical}. We
describe the idea for exactly identified moment condition models.
The results below can be extended to more general cases following the discussion in
Remark \ref{overidentified}. Recall that
$\tilde{f}_{tn}(\theta)=\frac{\omega(t/n)}{n}\sum^{t}_{j=1}f(z_j,\theta)$
for $t=1,2,\dots,n$. We consider the PEBEL ratio test statistic
which is defined as
\begin{align*}
elr_{pebel}(\theta)=&-\frac{1}{n}\log\left\{n^{n}\mathcal{L}_{pebel,n}(\theta)\right\}
=\min_{\pi\in\mathfrak{F}_n}\left\{-\frac{1}{n}\sum^{n}_{t=1}\log
(n\pi_t)+\tau\delta_n(\tilde{\mu}_{\pi}\left(\theta)\right)\right\},\quad \tau=c^*n,
\end{align*}
where
\begin{equation}
\mathcal{L}_{pebel,n}(\theta)=\max_{\pi\in\mathfrak{F}_n}\prod^{n}_{t=1}\pi_t\exp\left\{-n\tau\delta_n(\tilde{\mu}_{\pi}\left(\theta)\right)\right\},
\end{equation}
and
$\tilde{\mu}_{\pi}(\theta)=\sum^{n}_{t=1}\pi_t\tilde{f}_{tn}(\theta)$
with $\pi=(\pi_1,\dots,\pi_n)\in\mathfrak{F}_n$. Following similar
derivations in the proof of Proposition \ref{prop-1}, we deduce that
\begin{equation}\label{pen-ebel1}
\begin{split}
elr_{pebel}(\theta)=&\min_{\mu\in\tilde{\mathcal{H}}_n(\theta)}\left\{\frac{1}{n}\max_{\lambda\in\mathbb{R}^k}\sum^{n}_{t=1}\log
\left\{1+\lambda'(\tilde{f}_{tn}(\theta)-\mu)\right\}+\tau\delta_n\left(\mu\right)\right\}\\
=&\min_{\mu\in\mathbb{R}^k}\left\{\frac{1}{n}\max_{\lambda\in\mathbb{R}^k}\sum^{n}_{t=1}\log
\left\{1+\lambda'(\tilde{f}_{tn}(\theta)-\mu)\right\}+\tau\delta_n\left(\mu\right)\right\},
\end{split}
\end{equation}
where $\tilde{\mathcal{H}}_n(\theta)$ denotes the convex hull of $\{\tilde{f}_{tn}(\theta)\}^{n}_{t=1}$. Under suitable assumptions [see Nordman et al. (2013)], it can be
shown that
\begin{equation}\label{pen-ebel2}
\begin{split}
elr_{pebel}(\theta_0)\rightarrow^d
\min_{\tilde{\mu}\in\mathbb{R}^k}\left\{\max_{\tilde{\lambda}\in\mathbb{R}^k}\int^{1}_{0}\log
\left\{1+\tilde{\lambda}'(
\omega(r)W_k(r)-\tilde{\mu})\right\}dr+c^*\tilde{\mu}'\Phi^{-1}_k\tilde{\mu}\right\}.
\end{split}
\end{equation}
Notice that the PEBEL is free of $b$, but again it requires the
choice of a tuning parameters $c^*.$ For large $c^*$, we have
$\mu^*\approx 0$ and $\delta_n(\mu^*)\approx 0$ with $\mu^*$ being
the minimizer in (\ref{pen-ebel1}). Thus the PEBEL behaves like the
EBEL when $c^*$ is large. Following the discussion in Section
\ref{pel}, as $c^*$ becomes close to zero, $\mu^*$ gets near
$\sum^{n}_{t=1}\tilde{f}_{tn}(\theta)/n$ which satisfies that
$\max_{\lambda\in\mathbb{R}^k}\sum^{n}_{t=1}\log
\left\{1+\lambda'(\tilde{f}_{tn}(\theta)-\sum^{n}_{t=1}\tilde{f}_{tn}(\theta)/n)\right\}=0.$
Thus for small $c^*$, the behavior of the PEBEL ratio statistic evaluated at the true parameter value is
closely related to the self-normalized score statistic given by
\begin{equation}\label{eq:pebel-c=0}
\tilde{\mathcal{S}}_n(\theta_0)=n\delta_n\left(\sum^{n}_{t=1}\tilde{f}_{tn}(\theta_0)/n\right)
=n\left(\sum^{n}_{t=1}\tilde{f}_{tn}(\theta_0)/n\right)'\Psi_n^{-1}(\hat{\theta}_n)\left(\sum^{n}_{t=1}\tilde{f}_{tn}(\theta_0)/n\right).
\end{equation}
\begin{remark}
{\rm To resolve the coverage upper bound problem, one may consider adjusted versions of BEL and EBEL,
which retain the formulation of BEL and EBEL but add one or two pseudo-observations to the sample [see Chen et al. (2007);
Emerson and Owen (2009)]. However, a direct extension to the current setting may not work due to temporal dependence in moment conditions.
A possible strategy is to add a small fraction of artificial data points instead of one or two pseudo-observations, and derive the limiting distributions under the
fixed-$b$ asymptotics. This approach also requires the choice of additional tuning parameters such as the fraction of points being added, and we leave it for further investigation.
}
\end{remark}


\section{Numerical results}\label{sec:numerical}
In this section, we conduct simulation studies to evaluate the
finite sample performance of the penalization methods
proposed in Section \ref{sec:remedy}. We shall focus on the
confidence region for the mean of univariate/multivariate time
series. In the univariate case, we consider the AR(1) process
$z_t=\rho z_{t-1}+\varepsilon_t$ with $\rho=-0.5, 0.2, 0.5, 0.8$,
and the MA(1) process $z_t=\theta\epsilon_{t-1}+\epsilon_t$ with
$\theta=-0.5, 0.2, 0.5, 0.95$, where $\{\varepsilon_t\}$ and
$\{\epsilon_t\}$ are two sequences of i.i.d standard normal errors.
In the multidimensional case (i.e. $k>1$), we generate multivariate
time series with each component being independent AR(1) or MA(1)
process. The sample sizes considered are $n=100$ and 400. In the supplementary material, we present additional simulation results
for time series regression models, where the results are qualitatively similar to those for the mean.

\subsection{PBEL}
To implement the PBEL, we consider the self-normalization matrix $\Psi_n$ [Shao (2010)]
which is defined as,
\begin{equation}\label{eq:self}
\Psi_n(\hat{\theta}_n)=\frac{1}{n}\sum^{n}_{i=1}\sum^{n}_{j=1}\left(1-\left|\frac{i-j}{n}\right|\right)(z_i-\bar{z}_n)(z_j-\bar{z}_n)',
\end{equation}
where $\hat{\theta}_n=\bar{z}_n=\sum^{n}_{j=1}z_j/n.$ The tuning
parameter $c^*$ is chosen between 0.01 and 2. As pointed out in
Section \ref{pel}, the limiting distribution of the PBEL under the
fixed-$b$ asymptotics is pivotal and it can be approximated
numerically. Table \ref{tab:crit-PBEL} in the supplementary material summarizes the simulated
critical values for the limiting distributions of BEL and PBEL.
Selected simulation results are presented in Figures
\ref{fig:pel-ar-k1}-\ref{fig:pel-ar}. In the univariate case, the performance of PBEL
with $c^*=2$ and BEL are generally comparable in terms of the coverage probability and interval width. The PBEL with
$c^*=0.01$ delivers more accurate coverage as compared to the two alternatives especially when the
positive dependence is strong, although the corresponding interval width is slightly wider for relatively small $b$.
This finding is presumably due to the fact that the finite sample bounds for BEL do not deviate much from
the unity for $k=1$ and not quite large $b$ (see Table
\ref{tab:bound-1}). The simulation results for the MA models are
quantitatively similar and thus not presented here to conserve
space. In the case of $k=2$, the PBEL tends to provide better
coverage uniformly over $b$ as compared to the BEL (when the
dependence is positive). The improvement becomes more significant as
the block size grows. Also the PBEL with $c^*=0.01$ delivers the
most accurate coverage in most cases. Unreported numerical results
show that for $k=1,2$ and $c^*$ between $0.01$ and $2$, the
performance of PBEL is generally between the two cases reported
here. To assess the impact of dimensionality, we present the
coverage probabilities for the PBEL with $b=0.05,0.1$, and various
$c^*$ when $k=5$ (see the right column of Figure \ref{fig:pel-ar}).
Along with Table \ref{tab:pbel-k=5}, which also shows the coverage bound for the case $k=10$,
 we see that PBEL with
suitable $c^*$ offers improvement over the unpenalized counterpart.
The coverage upper bound problem clearly shows up for BEL especially
when the dependence is strong and dimension $k$ is large.
We also note that the choice of $c^*$ (that delivers the most accurate coverage)
is delicate in this case as it depends on $b$, the sample size $n$
and the underlying dependence structure. Overall, the finite sample
performance of the PBEL is satisfactory in terms of delivering
better coverage (especially when the bound is substantially below one)
as compared to the BEL under the fixed-$b$
asymptotics.


\subsection{PEBEL}\label{sub:sim-pebel}
We implement the PEBEL with $\Psi_n$ being the self-normalization
matrix in (\ref{eq:self}) and various choices of $c^*$. The simulated critical values
for the PEBEL are summarized in Table \ref{tab:crit-EBEL} in the supplementary material. We
present the coverage probabilities and interval widths for the unweighted EBEL and PEBEL
(i.e. $w(t)=1$) in Figures \ref{fig:pebel-ar-k1}-\ref{fig:pebel-ar}. Compared to the EBEL,
the PEBEL significantly improves coverage probabilities in all cases
considered here. The right column of Figure \ref{fig:pebel-ar-k1} suggests that the
PEBEL is also able to deliver smaller interval widths for the range of $c^*$ being considered.
In the univariate case, the choice of small $c^*$ seems to provide both better coverage and shorter interval
width. In the case of $k=2$, a relatively large $c^*$ tends to
provide good coverage as well, and the performance of PEBEL is not
affected much by the choice of $c^*$. It is worth noting that when
$k=5$, the performance of PEBEL deteriorates for $c^*=2$, which,
along with the above findings in the cases $k=1,2$,
suggests that the optimal  $c^*$ that delivers the most accurate coverage and the sensitivity
of the coverage with respect to $c^*$ can very much depend on the underlying dimensionality $k$.
To sum up, the numerical results
demonstrate the usefulness of the PEBEL as it provides significant
improvement over the EBEL provided that $c^*$ is suitably chosen.

In view of the right column of Figure \ref{fig:pebel-ar}, there
seems to be an optimal $c^*$ in terms of delivering the most
accurate coverage when $\rho=0.5$ and 0.8. Below we present a simple
 block bootstrap based method for choosing the tuning parameter
$c^*$. Suppose $n=b_nl_n$ where $b_n,l_n\in\mathbb{Z}$. Conditional
on the sample $\{z_l\}_{l=1}^{n}$, we let $M_1,\dots,M_{l_n}$ be
i.i.d uniform random variables on $\{0,\dots,l_n-1\}$ and define
$z_{(j-1)b_n+i}^*=z_{M_jb_n+i}$ with $1\leq j\leq l_n$ and $1\leq
i\leq b_n.$ In other words, $\{z_t^*\}_{t=1}^{n}$ is a
non-overlapping block bootstrap sample with block size $b_n$. For
each $c^*$, we can compute the times that the sample mean
$\bar{z}_n$ is contained in the confidence region constructed based
on the bootstrap sample $\{z_l^*\}_{l=1}^{n}$ and then compute the
empirical coverage probabilities based on $B$ bootstrap samples.
This is based on the notion that $\bar{z}_n$ is the true mean for
the bootstrap sample conditional on the data and the $c^*$ that
delivers the most accurate coverage for bootstrap sample is an
estimate of the optimal $c^*$ for the original series. Specifically,
we consider the case $n=100$ and $b_n=5$, and set $B=100$ and the
number of Monte Carlo replication to be 100 to see if this scheme
works well. For VAR(1) model with the coefficient matrix being $0.5
I_5$, the coverage probability based on the above tuning parameter
selection procedure is 98\% and the most frequently chosen $c^*$ is
0.4 (33\%). When the coefficient matrix is $0.8 I_5$, the
corresponding coverage probability is 90\% and the most frequently
chosen $c^*$ is 0.1 (51\%), which is identical to the empirically
optimal $c^*$ as seen from the third plot in the right column of
Figure \ref{fig:pebel-ar}. Hence the method of choosing the tuning
parameter $c^*$ seems to perform quite well. We shall leave a more
detailed examination of this bootstrap based tuning parameter
selection method in a separate work.

\section{Records of hemispheric temperatures}
To further illustrate the finite sample performance, we
apply the penalization methods (PBEL and PEBEL) and their unpenalized counterparts
to the so-called hemisphere
temperature anomaly time series (HadCRUT3v) available from the
Climate Research Unit (U.K.). The data, consisting of adjusted
monthly temperature averages from 1850 to 2010, combines the land
and marine gridded temperature anomalies, after correcting for
nonclimatic (e.g., instrumental) errors and adjusting the variance
[see e.g. Rayner et al. (2006), Jones et al. (2011) and references
therein for more details about the data set]. Following Kim et al.
(2013), we consider the annual average anomalies for months
December-January-February (DJF) and June-July-August (JJA) over the
years 1850-2009 in both northern and southern hemispheres; the DJF
values are means of average temperature anomaly of December of the
current year and January and February of the next year. We consider
fitting a simple linear regression model
$$Y_t=X_t\theta+\epsilon_t,\quad t=1,\dots,160,$$
for predicting the DJF temperature anomalies $\{Y_t\}$ from the JJA
ones $\{X_t\}$. Define the estimating equation
$f(Z_t,\theta)=X_t(Y_t-X_t\theta)$ with $Z_t=(X_t,Y_t)'$. If the model
is correctly specified and $\E X_t\epsilon_t=0,$ then $\E
f(Z_t,\theta_0)=0$ with $\theta_0$ being the true parameter. We
apply the penalization methods and their unpenalized versions to
compute 95\% confidence intervals for $\theta_0$ (see Table
\ref{tab:temp}). Since $\theta_0$ is unknown to us, it makes a fair comparison of various
EL methods difficult as we do not really know if the constructed confidence interval covers
$\theta_0$ or not. To this end, we propose to apply the EL methods to non-overlapping bootstrap sample
which mimics the dependence structure of the original time series,  and make a fair comparison.
In particular, let $n=b_nl_n$ where $n=160$ and $b_n,l_n\in\mathbb{Z}$. Let $M_1,\dots,
M_{l_n}$ be i.i.d uniform random variables on $\{0,\dots,l_n-1\}$
and let
$(X_{(j-1)b_n+i}^*,Y_{(j-1)b_n+i}^*)=(X_{M_jb_n+i},Y_{M_jb_n+i})$
with $1\leq j\leq l_n$ and $1\leq i\leq b_n.$ It is not hard to
verify that $\E^*\sum^{n}_{t=1}X_t^*(Y_t^*-X_t^*\hat{\theta})=0,$
where $\hat{\theta}=\sum^{160}_{t=1} X_tY_t/\sum^{160}_{t=1} X_t^2$
is the ordinary least-square (OLS) estimator and $\E^*$ denotes the
expectation conditional on the sample $\{X_t,Y_t\}_{t=1}^{160}.$ Thus for the bootstrap sample, the true $\theta$
is $\hat{\theta}$ conditional on the data and we can compute the empirical coverage probabilities for $\hat{\theta}$ based on 1000
bootstrap samples, where the block size $b_n$ is chosen to be 4 or
8. It is seen from Table \ref{tab:temp} that for the northern
hemisphere, undercoverage occurs for BEL, while PBEL with suitable
choice of $c^*$ can deliver better coverage. In such cases, the
corresponding interval widths delivered by PBEL based on the original
data are wider. For the southern hemisphere, BEL provides quite
accurate coverage and PBEL with $c^*=1,2$ are comparable with BEL in
terms of the coverage accuracy based on the bootstrap samples and
the confidence intervals based on the original data. In view of
Table \ref{tab:temp}, PEBEL provides better coverage compared to the
unpenalized version for all cases considered here. For the northern
hemispheric temperature anomalies, the PEBEL based confidence
intervals are wider while for the southern ones, PEBEL delivers
shorter interval widths. Based on 1000 bootstrap samples, we can
compute the percentages
of convex hull violation for EBEL. For the northern hemisphere, the upper bounds are 90.1\% and 88.2\% for $b_n=4$ and $8$ respectively; for
the southern hemisphere, the  upper bounds are 93.2\%
and 96.5\%, showing a serious deficit of EBEL method.
It is worth pointing out that the penalized methods
generally deliver wider interval widths for the northern hemispheric
data (in particular, PEBEL seems quite conservative in this case).
This finding might be due to the following facts. First, the JJA
temperature anomalies appear to be worse predictors for the DJF
anomalies in the northern hemisphere (with adjusted $R$-squared
0.6234) than in the southern hemisphere (with adjusted $R$-squared
0.8768). Second, the plot of $f(Z_t,\hat{\theta})$ in the northern
hemisphere (Figure \ref{fig:temp}) tends to exhibit certain
nonstationarity features (e.g. in the second order property), which
may pose difficulty in constructing a confidence interval for
$\theta$.

\section{Conclusion}
In this paper, we study the upper bounds on the coverage
probabilities of the BEL and EBEL based confidence regions via
theory and simulations. Our theoretical results, which are derived
for the pivotal limit of the BEL ratio obtained under the fixed-$b$
asymptotics,  suggest that the large sample coverage upper bound
for BEL is strictly less than one for any $b\in (0,1)$. This result
is in sharp contrast to those corresponding to the EL for
independent moment conditions, where the large sample bound is
always equal to one due to the $\chi^2$ limit. By numerical
simulations, we discover that the finite sample coverage bounds for
both BEL and EBEL  can be far below the nominal level in the cases
when (\rmnum{1}) the dimension of moment condition $k$ is moderate or
high; (\rmnum{2}) the dependence of moment conditions is positively strong;
or (\rmnum{3}) $b$ is large for BEL. The deterioration of  the coverage
for the EBEL based confidence region with respect to $k$ is
especially noticeable. These phenomena appear to be discovered
 for the first time for these two important EL methods in the time series context, which will hopefully
lead to a new research direction on EL methods for dependent data.

To overcome the convex hull constraint and related undercoverage
problem, we introduce the penalization based  BEL and EBEL methods,
which drop the convex hull constraint and penalize the original
EL using the quadratic distance measure, and
derive their limiting distributions under the fixed-$b$ asymptotics.
Interestingly, the penalization generates a new class of statistics
which lies in between the empirical log-likelihood ratio statistic
and the self-normalized score statistic through the choice of a
tuning parameter $c^*$. Our simulation studies show that the
penalization based methods can outperform their unpenalized
counterparts in terms of coverage accuracy especially when the
coverage bound is below the nominal level. In addition, we propose a
method of choosing the tuning parameter and demonstrate its
effectiveness through a simulation example. It is worth mentioning
that our techniques (i.e., fixed-$b$ asymptotics and penalization)
are expected to be extendable to other variants of
EL methods for time series or spatial data, such as tapered
blockwise EL [Nordman (2009)] and spatial EL
[Nordman and Caragea (2008)]. We shall leave these for future
investigation.
\\
\\
\centerline{\Large\textbf{Acknowledgments}} The authors would like to thank the Editor, the Associate Editor
and two reviewers for their constructive comments, which substantially improve the paper. Shao's research is supported in part by National Science Foundation grant DMS-11-04545.

\section{Technical appendix}
\begin{proof}[Proof of Lemma \ref{lemma1}]
Suppose $\mathcal{A}$ is unbounded. We note that
$\mathcal{A}=\cap_{r\in(0,1-b)}\{\lambda\in\mathbb{R}^k:
\lambda'D_k(r;b)\geq -1\}$ which is the intersection of a set of
closed half-spaces. The recession cone of $\mathcal{A}$ is then
given by $0^+\mathcal{A}=\cap_{r\in(0,1-b)}\{\lambda\in\mathbb{R}^k:
\lambda'D_k(r;b)\geq 0\}$ [see Section 8 of Rockafellar (1970)]. By
Theorem 8.4 of Rockafellar (1970), there exists a nonzero vector
$\tilde{\lambda}\in 0^+\mathcal{A}$, that is
$\tilde{\lambda}'D_k(r;b)\geq 0$ for all $r\in(0,1-b)$. Thus we know
$\{t_k(r;b): |D_k(r;b)|>0, r\in(0,1-b)\}$ lie on the same hemisphere
of the unit sphere $\mathcal{S}^{k-1}$ [see e.g. Wendel (1962)], and
the origin is not an interior point of the convex hull
$\mathcal{H}(D_k)$. On the other hand, if the origin is not an
interior point of the convex hull $\mathcal{H}(D_k)$. Then
$\{t_k(r;b): |D_k(r;b)|>0, r\in(0,1-b)\}$ lie on the same
hemisphere. By the supporting/separating hyperplane theorem [see e.g. Section 11 of Rockafellar (1970)], we can find a nonzero vector
$\tilde{\lambda}\in\mathbb{R}^k$ such that
$\tilde{\lambda}'D_k(r;b)\geq 0$ for all $r\in (0,1-b)$. It is easy to see that
$a\tilde{\lambda}\in\mathcal{A}$ for any $a>0,$ which implies that
$\mathcal{A}$ is unbounded.
\end{proof}

\begin{proof}[Proof of Theorem \ref{bound-k=1}]
Because $P(D_1(r;b)\geq 0,
\forall~r\in(0,1-b])=\lim_{a\downarrow 0}P(D_1(r;b)> -a,
\forall~r\in(0,1-b])$, we shall derive a formula for $P(D_1(r;b)> -a,
\forall~r\in(0,1-b])$ with $a>0.$ We first consider the case where $bL=1$ with $L$ being a positive integer. Note that
\begin{align*}
&\mathcal{B}:=\mathcal{B}(a)=\{D_1(r;b)>-a, \forall~r\in (0,1-b]\}=\{W_1(r)<W_1(r+b)+a,
\forall~r\in(0,1-b]\}
\\=&\{W_1(r)<W_1(r+b)+a<W_1(r+2b)+2a<\cdots<W_1(r+(L-1)b)+(L-1)a, \forall~r\in(0,b]\}
\\=&\{\mathcal{W}_1(r)<\mathcal{W}_2(r)<\cdots<\mathcal{W}_{L}(r), \forall~r\in(0,b]\},
\end{align*}
where $\mathcal{W}_i(r)=\mathcal{W}_i(r;a)=W_1(r+(i-1)b)+(i-1)a$ with $i=1,\dots,L$. Thus we deduce that
\begin{equation}\label{integral}
\begin{split}
&P(D_1(r;b)>-a, \forall~r\in (0,1-b])
\\=&\int\cdots\int_{\mathcal{C}(a)} P(\mathcal{B},
W_1((i-1)b)\in dx_i, i=1,2,\dots,L+1)
\\=&\int\cdots\int_{\mathcal{C}(a)} P(\mathcal{B}|
\mathcal{W}_i(0)=x_i+(i-1)a, \mathcal{W}_i(b)=x_{i+1}+(i-1)a, i=1,2,\dots,L)
\\&P(W_1((i-1)b)\in dx_i, i=1,2,\dots,L+1),
\end{split}
\end{equation}
where
$\mathcal{C}(a)=\{(x_2,\dots,x_{L+1})\in\mathbb{R}^L: -a<x_2<x_3+a<\cdots<x_{L+1}+(L-1)a\}$ and $x_1=0.$
For the first term under the integral (\ref{integral}), the
conditioned Wiener processes $\mathcal{W}_i$s are independent. Therefore, by
equation (2.12) of Shepp (1971) [also see Karlin and Mcgregor
(1959)], we obtain
\begin{align*}
&P(\mathcal{B}|
\mathcal{W}_i(0)=x_i+(i-1)a, \mathcal{W}_i(b)=x_{i+1}+(i-1)a, i=1,2,\dots,L)
\\=&q_{b,L}(v_{1}(x;a),v_{2}(x;a))/\prod^{L}_{i=1}\phi_b(x_{i+1}-x_i),
\end{align*}
where $v_1(x;a)=(x_1,x_2+a,\dots,x_{L}+(L-1)a)'$ and $v_2(x;a)=(x_2,x_3+a,\dots,x_{L+1}+(L-1)a)'$ with $x_1=0$, and
$\phi_b(\cdot)=\phi(\cdot/\sqrt{b})/\sqrt{b}$ with
$\phi(x)=\frac{1}{\sqrt{2\pi}}\exp(-x^2/2)$. By the property of
Wiener process, the second term under the integral (\ref{integral})
is simply given by
\begin{align*}
P(W_1((i-1)b)\in dx_i, i=1,2,\dots,L+1)=\prod_{i=1}^{L}\phi_b(x_{i+1}-x_i)dx_2dx_3\cdots
dx_{L+1}.
\end{align*}
Combing the above results, we have
\begin{equation}\label{eq:k=1}
\begin{split}
P(D_1(r;b)>-a, \forall~r\in
(0,1-b))=&\int_{\mathcal{C}(a)}q_{b,L}(v_{1}(x;a),v_{2}(x;a))dx_2dx_3\cdots.
dx_{L+1}.
\end{split}
\end{equation}
By letting $a\downarrow 0$ in (\ref{eq:k=1}), we deduce that
\begin{equation}
\begin{split}
P(D_1(r;b)\geq 0,\forall~r\in(0,1-b])=&\int_{0<x_2<x_3<\cdots<x_{L+1}}q_{b,L}(x_{1:L},x_{2:(L+1)})dx_2dx_3\cdots
dx_{L+1}
\\=&\int_{0<x_2<x_3<\cdots<x_{L+1}}q_{1,L}(x_{1:L},x_{2:(L+1)})dx_2dx_3\cdots
dx_{L+1},
\end{split}
\end{equation}
where $x=(x_1,\dots,x_{L+1})$.

Next, we consider the case where $Lb+\tau=1$ with $L$ being a
positive integer and $0<\tau<b$. With some abuse of notation, define
$\mathcal{W}_j(r)=\mathcal{W}_j(r;a)=W_1(r+(j-1)b)+(j-1)a$ with
$1\leq j\leq L+1$ and $r\in (0,\tau]$, and
$\mathcal{W}_l'(r')=\mathcal{W}_l'(r';a)=W_1(r'+(l-1)b+\tau)+(l-1)a$
with $1\leq l\leq L$ and $r'\in (0,b-\tau]$. Following Shepp (1971),
we have $\mathcal{B}=\mathcal{B}_1\cap \mathcal{B}_2,$ where
\begin{align*}
\mathcal{B}_1=&\{W_1(r)<W_1(r+b)+a<\cdots<W_1(r+Lb)+La, \forall~r\in (0,\tau]\}
\\=&\{\mathcal{W}_1(r)<\mathcal{W}_2(r)<\cdots<\mathcal{W}_{L+1}(r), \forall~r\in (0,\tau]\},
\end{align*}
and
\begin{align*}
\mathcal{B}_2=&\{W_1(r'+\tau)<W_1(r'+b+\tau)+a<\cdots<W_1(r'+(L-1)b+\tau)+(L-1)a,
\forall~r'\in(0,b-\tau]\}
\\=&\{\mathcal{W}_1'(r')<\mathcal{W}_2'(r')<\cdots<\mathcal{W}_{L}'(r'),
\forall~r'\in(0,b-\tau]\}.
\end{align*}
Following the arguments above, we shall consider the processes
$\mathcal{W}_j(r)$ and $\mathcal{W}_j'(r')$ conditional on their
boundary values, i.e., $\mathcal{W}_j(0)=x_{j}+(j-1)a$ and
$\mathcal{W}_j(\tau)=y_j+(j-1)a$ with $1\leq j\leq L+1$, and
$\mathcal{W}_l'(0)=y_l+(l-1)a$ and
$\mathcal{W}_{l}'(b-\tau)=x_{l+1}+(l-1)a$ with $1\leq l\leq L$. Here
we have $0=x_1<x_2+a<\cdots <x_{L+1}+La$ and
$y_1<y_2+a<\cdots<y_{L+1}+La$. Notice that conditioning on the
boundary values, the events $\mathcal{B}_1$ and $\mathcal{B}_2$ are
mutually independent, and the Wiener processes
$\{\mathcal{W}_j(r)\}_{j=1}^{L+1}$
($\{\mathcal{W}_j'(r')\}_{j=1}^{L}$) are independent. Define
$\mathcal{D}_1(a)=\{\mathcal{W}_j(0)=x_{j}+(j-1)a,
\mathcal{W}_j(\tau)=y_j+(j-1)a, 1\leq j\leq L+1\}$ and
$\mathcal{D}_2(a)=\{\mathcal{W}_l'(0)=y_l+(l-1)a,
\mathcal{W}_{l}'(b-\tau)=x_{l+1}+(l-1)a, 1\leq l\leq L\}$. We deduce
that
\begin{align*}
P(\mathcal{B})=&\int\cdots \int_{\mathcal{C}'(a)}
P(\mathcal{B}|\mathcal{D}_1\cap\mathcal{D}_2)P(W_{1}((j-1)b)\in dx_j,W_1(\tau+(j-1)b)\in dy_j,1\leq
j\leq L+1)
\\=&\int\cdots\int_{\mathcal{C}'(a)}
P(\mathcal{B}_1|\mathcal{D}_1)P(\mathcal{B}_2|\mathcal{D}_2)P(W_{1}((j-1)b)\in dx_j,W_1(\tau+(j-1)b)\in dy_j,1\leq
j\leq L+1),
\end{align*}
where
$\mathcal{C}'(a)=\{(y_1,x_2,y_2,\cdots,x_{L+1},y_{L+1})\in\mathbb{R}^{2L+1}:
-a<x_2<\cdots <x_{L+1}+(L-1)a, y_1<y_2+a<\cdots<y_{L+1}+La\}$. Let
$x=(x_1,\dots,x_{L+1})'$ with $x_1=0$, and $y=(y_1,\dots,y_{L+1})'$.
Define $u_1(x;a)=(x_1,x_2+a,\dots,x_{L+1}+La)'$,
$u_1'(y;a)=(y_1,y_2+a,\dots,y_{L+1}+La)'$,
$u_2(x;a)=(x_2,x_3+a,\dots,x_{L+1}+(L-1)a)'$ and
$u_2'(y;a)=(y_1,y_2+a,\dots,y_{L}+(L-1)a)'$. Using the
fact that [see (2.12) of Shepp (1971)],
\begin{align*}
&P(\mathcal{B}_1|\mathcal{D}_1)=q_{\tau,L+1}(u_1(x;a),u_1'(y;a))/\prod^{L+1}_{i=1}\phi_{\tau}(y_i-x_i),\\
&P(\mathcal{B}_2|\mathcal{D}_2)=q_{b-\tau,L}(u_2(x;a),u_2'(y;a))/\prod^{L}_{i=1}\phi_{b-\tau}(y_i-x_{i+1}),
\end{align*}
and
\begin{align*}
&P(W_{1}((j-1)b)=x_j,W_1(\tau+(j-1)b)=y_j,1\leq j\leq L+1)
\\=&\prod^{L}_{i=1}\phi_{b-\tau}(y_i-x_{i+1})\prod^{L+1}_{i=1}\phi_{\tau}(y_i-x_i)dy_1dx_2dy_2\dots
dx_{L+1}dy_{L+1}.
\end{align*}
It thus implies that
\begin{align*}
&P(D_1(r;b)>-a, \forall~r\in (0,1-b])
\\=&\int\cdots \int_{\mathcal{C}'(a)}
q_{\tau,L+1}(u_1(x;a),u_1'(y;a))q_{b-\tau,L}(u_2(x;a),u_2'(y;a))dy_1dx_2dy_2\dots
dx_{L+1}dy_{L+1}.
\end{align*}
By letting $a\downarrow 0$, we obtain
\begin{equation}
\begin{split}
&P(D_1(r;b)\geq 0, \forall~r\in (0,1-b])
\\=&\int\cdots\int_S
q_{\tau,L+1}(x,y)q_{b-\tau,L}(x_{2:(L+1)},y_{1:L})dy_1dx_2dy_2\cdots
dx_{L+1}dy_{L+1}
\\=&\int\cdots\int_S
q_{\xi,L+1}(x,y)q_{1-\xi,L}(x_{2:(L+1)},y_{1:L})dy_1dx_2dy_2\cdots
dx_{L+1}dy_{L+1},
\end{split}
\end{equation}
where $\xi=\tau/b$ and
$S=\{(y_1,x_2,y_2,\cdots,x_{L+1},y_{L+1})\in\mathbb{R}^{2L+1}:
0=x_1<x_2<\cdots <x_{L+1}, y_1<y_2<\cdots<y_{L+1}\}$. In fact, one can
derive the results presented above by applying the results in Shepp
(1971) and the scaling property of Wiener process, i.e.,
$W(rb)/\sqrt{b}$ is another Wiener process. We present the  details
for the sake of clarity.
\end{proof}

\begin{proof}[Proof of Proposition \ref{prop-1}]
When $\mu\notin \mathcal{H}_n(\theta;b)$ i.e., the origin is outside
the convex hull of $\{f_{tn}(\theta)-\mu\}^{N}_{t=1}$, and the space spanned by $\{f_{tn}(\theta)\}^{N}_{t=1}$ is of $k$ dimension,
the separating hyperplane theorem implies that there exits a $\lambda_0:=\lambda_0(\theta)\in\mathbb{R}^k$ such that $\lambda_0'(f_{tn}(\theta)-\mu)\geq 0$ for all $1\leq t\leq N$
and $\lambda_0'(f_{t_0n}(\theta)-\mu)>0$ for at least one $1\leq t_0\leq N$. Thus we have
$$\max_{\lambda\in\mathbb{R}^k}\sum^{N}_{t=1}\log
\left\{1+\lambda'(f_{tn}(\theta)-\mu)\right\}\geq \max_{a\geq 0}\sum^{N}_{t=1}\log
\left\{1+a\lambda'_0(f_{tn}(\theta)-\mu)\right\}=+\infty,$$
which implies that
\begin{equation}
\begin{split}
elr_{pbel}(\theta)=&\min_{\mu\in\mathcal{H}_n(\theta;b)}\left\{\frac{2}{nb}\max_{\lambda\in\mathbb{R}^k}\sum^{N}_{t=1}\log
\left\{1+\lambda'(f_{tn}(\theta)-\mu)\right\}+\frac{\tau}{b}\delta_n\left(\mu\right)\right\}\\
=&\min_{\mu\in\mathbb{R}^k}\left\{\frac{2}{nb}\max_{\lambda\in\mathbb{R}^k}\sum^{N}_{t=1}\log
\left\{1+\lambda'(f_{tn}(\theta)-\mu)\right\}+\frac{\tau}{b}\delta_n\left(\mu\right)\right\}.
\end{split}
\end{equation}
\end{proof}

\begin{proof}[Proof of Proposition \ref{boundness}].
From the definition
of $elr_{pbel}(\theta)$, we get
\begin{equation}\label{pen-equ1}
elr_{pbel}(\theta)\leq
\frac{2}{nb}\max_{\lambda\in\mathbb{R}^k}\sum^{N}_{t=1}\log
\left\{1+\lambda'(f_{tn}(\theta)-\bar{f}_{n}(\theta))\right\}+\frac{\tau}{b}\delta_n\left(\bar{f}_{n}(\theta)\right),
\end{equation}
where $\bar{f}_{n}(\theta)=\frac{1}{N}\sum^{N}_{t=1}f_{tn}(\theta)$.
When evaluated at $\theta=\theta_0$, the RHS of (\ref{pen-equ1})
converges in distribution to,
$$ \widetilde{U}_k(b)=\frac{2}{b}\max_{\lambda\in\mathbb{R}^k}\int^{1-b}_{0}\log
\left\{1+\lambda'(D_k(r;b)-\bar{D}_k(b))\right\}dr+\frac{c^*}{b^3}\bar{D}_k(b)'\Phi^{-1}_k\bar{D}_k(b),$$
where $\bar{D}_k(b)=\frac{1}{1-b}\int^{1-b}_{0}D_k(r;b)dr.$ Thus we see that
\begin{equation}\label{pen-equ2}
P\left(U_{pbel,k}(b)\leq \widetilde{U}_k(b)\right)=1.
\end{equation}
Because $\int^{1-b}_{0}\left\{D_k(r;b)-\bar{D}_k(b)\right\}dr=0$, we
have $P(\widetilde{U}_k(b)<\infty)=1$, which along with
(\ref{pen-equ2}) implies that $P\left(U_{pbel,k}(b)<\infty\right)=1.$
Note that if there exists a $\tilde{\lambda}\in\mathbb{R}^k$ such
that $\tilde{\lambda}'(D_k(r;b)-\bar{D}_k(b))\geq 0$ for $r\in
[0,1-b]$ and $\tilde{\lambda}'(D_k(r;b)-\bar{D}_k(b))>0$ for
$r\in\mathcal{M}\subset [0,1-b]$, where $\mathcal{M}$ has positive
Lebesgue measure, then
$\int^{1-b}_{0}\tilde{\lambda}'(D_k(r;b)-\bar{D}_k(b))dr>0$ which
contradicts with the fact that
$\int^{1-b}_{0}\left\{D_k(r;b)-\bar{D}_k(b)\right\}dr=0$. With probability one, the origin of $\mathbb{R}^k$ is
an interior point of the convex hull of $\{D_k(r;b)-\bar{D}_k(b)\}_{r\in [0,1-b]}$ and thus $P(\widetilde{U}_k(b)<\infty)=1$.
\end{proof}


\newpage

\begin{table}[H]
\caption{Bounds on the coverage probabilities for BEL in \%}\label{tab:bound-1}
\begin{center}\footnotesize
\begin{tabular}{ccc rrrrrrrrrrr}\toprule
& &\multicolumn{10}{c}{$L=1/b$}
\\ \cmidrule(r){4-14}
$n$ & $\rho$ & $k$ & $2$ & $3$ & $4$ & $5$ & $6$ & $7$ & $8$ & $9$ & $10$
& $15$ & $20$\\\midrule
$50$ & 0.0 & 1 & 73.61&93.61&98.07&99.30&99.83&99.93&99.98&99.99&99.99&100.00&100.00\\
     & 0.0 & 2 & 37.22&75.83&90.47&95.71&98.80&99.52&99.82&99.97&99.97&100.00&100.00\\
     & 0.2 & 1 & 71.20&92.26&97.36&98.92&99.70&99.87&99.94&99.99&99.99&100.00&100.00\\
     & 0.2 & 2 & 33.45&71.65&87.90&93.80&97.89&99.02&99.68&99.91&99.91&100.00&100.00\\
     & 0.5 & 1 & 66.47&89.26&95.61&97.83&99.21&99.57&99.80&99.93&99.93&100.00&100.00\\
     & 0.5 & 2 & 26.73&63.07&81.76&89.19&95.31&97.20&98.60&99.43&99.43&99.98&100.00\\
     & 0.8 & 1 & 56.83&80.69&89.45&92.89&95.82&96.93&97.91&98.63&98.63&99.62&99.84\\
     & 0.8 & 2 & 15.98&44.42&62.81&72.30&81.15&85.25&89.09&92.53&92.53&97.24&98.84\\
     & -0.5& 1 & 80.16&96.44&99.24&99.80&99.98&99.99&100.00&100.00&100.00&100.00&100.00\\
     & -0.5& 2 & 48.66&85.17&95.78&98.61&99.84&99.90&100.00&100.00&100.00&100.00&100.00\\
$100$ & 0.0& 1 & 76.36&94.13&98.34&99.57&99.90&99.97&99.99&99.99&99.99&100.00&100.00\\
      & 0.0& 2 & 40.72&76.22&91.04&96.94&99.13&99.76&99.90&99.97&99.99&100.00&100.00\\
      & 0.2& 1 & 74.59&93.17&97.86&99.40&99.85&99.95&99.98&99.99&99.99&100.00&100.00\\
      & 0.2& 2 & 37.51&73.02&89.00&95.91&98.68&99.47&99.80&99.90&99.97&100.00&100.00\\
      & 0.5& 1 & 71.16&91.01&96.77&98.92&99.65&99.84&99.95&99.98&99.99&100.00&100.00\\
      & 0.5& 2 & 32.39&67.18&84.64&93.24&97.33&98.54&99.47&99.69&99.85&100.00&100.00\\
      & 0.8& 1 & 63.71&85.58&93.20&96.71&98.53&99.08&99.49&99.64&99.76&99.98&99.99\\
      & 0.8& 2 & 22.53&53.78&72.08&83.59&91.28&94.11&96.37&97.40&98.21&99.78&99.90\\
      &-0.5& 1 & 80.78&96.27&99.21&99.82&99.99&99.99&100.00&100.00&100.00&100.00&100.00\\
      &-0.5& 2 & 48.77&83.68&95.10&98.70&99.67&99.99&100.00&100.00&100.00&100.00&100.00\\
$500$ & 0.0 & 1 &79.87&95.51&98.96&99.79&99.96&99.98&100.00&100.00&100.00&100.00&100.00\\
      & 0.0 & 2 &45.60&80.19&93.63&98.12&99.57&99.87&99.95&99.99&100.00&100.00&100.00\\
      & 0.2 & 1 &79.16&95.11&98.83&99.74&99.95&99.98&99.99&100.00&100.00&100.00&100.00\\
      & 0.2 & 2 &44.10&79.22&92.90&97.75&99.44&99.83&99.94&99.98&99.99&100.00&100.00\\
      & 0.5 & 1 &77.50&94.37&98.52&99.63&99.92&99.97&99.98&100.00&100.00&100.00&100.00\\
      & 0.5 & 2 &41.33&76.71&91.50&96.83&99.10&99.68&99.86&99.96&99.98&100.00&100.00\\
      & 0.8 & 1 &73.91&92.43&97.65&99.25&99.75&99.90&99.95&99.98&99.99&100.00&100.00\\
      & 0.8 & 2 &35.56&70.69&87.36&94.58&97.81&99.13&99.58&99.84&99.92&100.00&100.00\\
      & -0.5& 1 &81.92&96.36&99.30&99.86&99.97&100.00&100.00&100.00&100.00&100.00&100.00\\
      & -0.5& 2 &49.76&83.39&95.42&98.87&99.74&99.91&99.97&100.00&100.00&100.00&100.00\\
$1000$ & 0.0& 1 &80.06&95.67&99.05&99.78&99.96&99.99&100.00&100.00&100.00&100.00&100.00\\
       & 0.0& 2 &45.25&80.59&94.01&98.16&99.54&99.87&99.97&99.99&100.00&100.00&100.00\\
       & 0.2& 1 &79.49&95.43&98.97&99.77&99.95&99.99&100.00&100.00&100.00&100.00&100.00\\
       & 0.2& 2 &44.09&79.60&93.50&98.02&99.39&99.84&99.97&99.98&100.00&100.00&100.00\\
       & 0.5& 1 &78.40&94.94&98.79&99.72&99.94&99.99&100.00&100.00&100.00&100.00&100.00\\
       & 0.5& 2 &42.06&77.83&92.50&97.55&99.23&99.73&99.89&99.95&100.00&100.00&100.00\\
       & 0.8& 1 &76.01&93.68&98.27&99.48&99.86&99.95&100.00&100.00&100.00&100.00&100.00\\
       & 0.8& 2 &37.93&73.21&89.65&96.21&98.57&99.49&99.76&99.90&99.97&100.00&100.00\\
       &-0.5& 1 &81.64&96.30&99.24&99.85&99.98&100.00&100.00&100.00&100.00&100.00&100.00\\
       &-0.5& 2 &48.21&82.96&95.11&98.76&99.74&99.92&99.98&100.00&100.00&100.00&100.00\\
$+\infty$ & 0.0 & 1 &81.70&96.26&99.23&99.85&99.97&99.99&100.00&100.00&100.00&100.00&100.00\\
          & 0.0 & 2 &48.58&82.93 &95.04 &98.72 &99.70 &99.91 &99.99 &100.00 &100.00 &100.00 &100.00\\
\midrule
\end{tabular}
Note: the number of Monte Carlo replication is 50,000 for $k=1$
(10,000 for $k=2$). For the last row $n=+\infty,$ we approximate the
probability $P(\mathcal{A}~\text{is bounded})$ by simulating
independent Wiener processes, where the Wiener process is
approximated by a normalized partial sum of 50,000 for $k=1$ (10,000
for $k=2$) i.i.d standard normal random variables and the number of
replications is 100,000 for $k=1$ (50,000 for $k=2$).
\end{center}
\end{table}

\newpage

\begin{table}[H]
\caption{Bounds on the coverage probabilities for EBEL
in \%}\label{tab:bound-EBEL}
\begin{center}
\begin{tabular}{cc rrrrr}\toprule
& &\multicolumn{5}{c}{$n$}
\\ \cmidrule(r){3-7}
  $\rho$ & $k$ & $50$ & $100$ & $500$ & $1000$ & $5000$ \\\midrule
  0.0 & 1 &84.02 &89.51  & 94.64 &96.18 & 98.30 \\
  0.0 & 2 &52.45  &62.66  &78.48  &84.22 &92.09 \\
 0.2 & 1 &82.20  &87.93 &93.99 &95.74 &98.12\\
  0.2 & 2 &48.95  &59.47  &76.50 &82.55 &91.51\\
  0.5 & 1 &77.57  &84.31 &92.35 & 94.57 &97.72\\
 0.5 & 2 &41.24  &52.96  &72.31  &79.09  &89.84 \\
 0.8 & 1 &65.99  &75.78  &88.32 &91.91 &96.52 \\
 0.8 & 2 &26.56  &39.06  &62.10 &71.32  &85.62\\
 -0.5 & 1 &87.42  &91.75 &95.89  &96.83 &98.70\\
 -0.5 & 2 &60.24  &69.35 &82.52  &87.16 &93.87\\
\midrule
\end{tabular}\\
Note: the number of Monte Carlo replication is 10,000. The bounds on
the coverage probabilities for EBEL do not depend on the choice of
the weight function $\omega(\cdot)$.
\end{center}
\end{table}

\begin{table}[H]
\caption{Coverage probabilities in \% for the mean delivered by
BEL}\label{tab:pbel-k=5}
\begin{center}\footnotesize
\begin{tabular}{cc cccc}\toprule
& & \multicolumn{4}{c}{$\rho$}
\\ \cmidrule(r){3-6}
$n$ & $b$ & $0.2$  & $0.5$ & $0.8$ & $-0.5$
\\ \midrule
100 & $0.05$ & 88.5 (97.8) & 76.1 (86.7) & 34.3 (24.1) & 98.6 (99.9)\\
100 & $0.10$ & 84.7 (37.3) & 74.6 (17.6) & 42.4 (1.4) & 97.5 (80.9)\\
400 & $0.05$ & 93.8 (99.9) & 91.8 (99.6) & 80.3 (92.6) & 97.2 (100.0)\\
400 & $0.10$ & 93.0 (68.0) & 90.2 (55.8) & 78.7 (26.1) & 96.5 (87.2)\\
\midrule
\end{tabular}\\
Note: the data is generated from the VAR(1) process with the
coefficient matrix being $\rho I_k$ for $k=5$ or 10. The number in
the parentheses is the coverage upper bound for the case of $k=10$.
\end{center}
\end{table}

\newpage

\begin{table}[H]
\caption{Confidence intervals and coverage probabilities in \% for
the hemispheric temperatures records}\label{tab:temp}
\begin{center}\footnotesize
\begin{tabular}{ccc crrcrr}\toprule
& & & \multicolumn{3}{c}{northern hemisphere} &
\multicolumn{3}{c}{southern hemisphere}
\\ \cmidrule(r){4-6} \cmidrule(r){7-9}
 & $m$ & $c^*$ & CI & CP$_4$ & CP$_8$ & CI & CP$_4$ & CP$_8$ \\ \midrule
PBEL & 4 & 0.05 & $[1.012, 1.336]$ & 90.6 & 82.2 & $[0.862, 0.949]$ & 93.5 & 91.3          \\
     & 4 & 0.1  & $[0.961, 1.431]$ & 98.5 & 95.6 & $[0.833, 0.977]$ & 99.0 & 99.4         \\
     & 4 & 0.2  & $[0.983, 1.404]$ & 98.0 & 94.5 & $[0.837, 0.973]$ & 98.7 & 99.1        \\
     & 4 & 1    & $[1.028, 1.330]$ & 93.2 & 86.0 & $[0.854, 0.957]$ & 95.5 & 95.1        \\
     & 4 & 2    & $[1.031, 1.326]$ & 93.1 & 85.5 & $[0.854, 0.956]$ & 95.1 & 94.8          \\
BEL  & 4 & ---  & $[1.033, 1.323]$ & 92.7 & 85.0 & $[0.855, 0.955]$ & 95.0 & 94.8\\
PBEL & 8 & 0.05 & $[0.983, 1.417]$ & 91.5 & 84.8 & $[0.858, 0.948]$ & 93.7 & 92.5\\
     & 8 & 0.1  & $[1.022, 1.377]$ & 90.3 & 83.0 & $[0.860, 0.946]$ & 92.1 & 89.2\\
     & 8 & 0.2  & $[0.974, 1.481]$ & 98.1 & 95.1 & $[0.838, 0.976]$ & 98.6 & 99.1\\
     & 8 & 1    & $[1.029, 1.383]$ & 93.9 & 87.5 & $[0.853, 0.956]$ & 95.5 & 94.6\\
     & 8 & 2    & $[1.034, 1.374]$ & 93.0 & 86.2 & $[0.855, 0.955]$ & 95.2 & 94.3\\
BEL  & 8 & --- & $[1.039, 1.365]$  & 92.3 & 84.9 & $[0.856, 0.953]$
& 94.7 & 93.6\\ \hline
PEBEL & ---& 0.05 & $[0.738, 1.825]$ & 93.7 & 90.5 & $[0.868, 0.952]$ & 95.2 & 95.1          \\
      & ---& 0.1  & $[0.742, 1.830]$ & 93.6 & 90.6 & $[0.869, 0.955]$ & 95.2 & 95.3         \\
      & ---& 0.2  & $[0.746, 1.832]$ & 93.6 & 90.8 & $[0.871, 0.962]$ & 95.0 & 95.5        \\
      & ---& 1    & $[0.800, 1.786]$ & 92.8 & 90.5 & $[0.876, 0.988]$ & 94.6 & 94.8        \\
      & ---& 2    & $[0.831, 1.760]$ & 92.3 & 89.7 & $[0.877, 1.002]$ & 94.6 & 94.6          \\
EBEL  & ---& ---  & $[1.059, 1.538]$ & 87.0 & 84.4 & $[0.880, 1.133]$ & 89.6 & 93.3\\
\midrule
\end{tabular}\\
Note: the columns CP$_4$ and CP$_8$ correspond to the coverage
probabilities based on the bootstrap samples with block size $b_n=4$
and $b_n=8$ respectively.
\end{center}
\end{table}

\newpage

\begin{figure}[H]
\centering
\includegraphics[height=9cm,width=9cm]{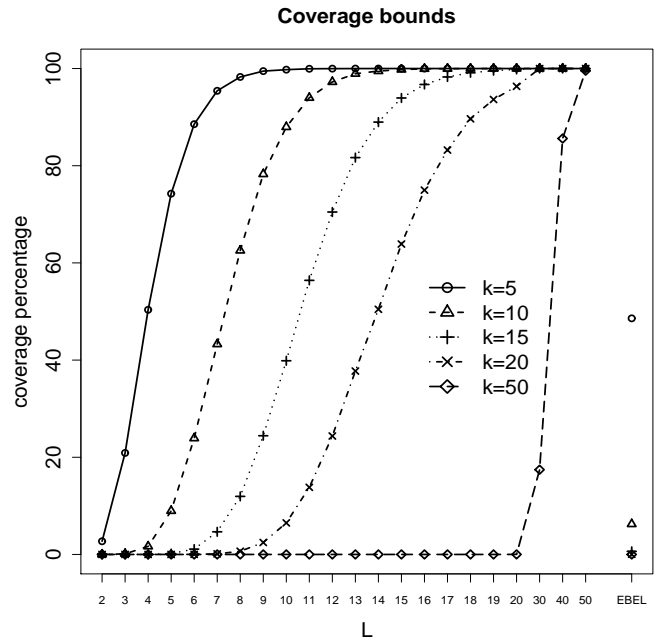}
\caption{Bounds on the coverage probabilities for BEL and EBEL in \%. Note: the
data are generated from multivariate standard normal distribution with $n=5,000$ and the number of Monte Carlo
replications is 10,000.}\label{fig:bound}
\end{figure}

\newpage

\begin{figure}[H]
\centering
\includegraphics[height=5.2cm,width=6cm]{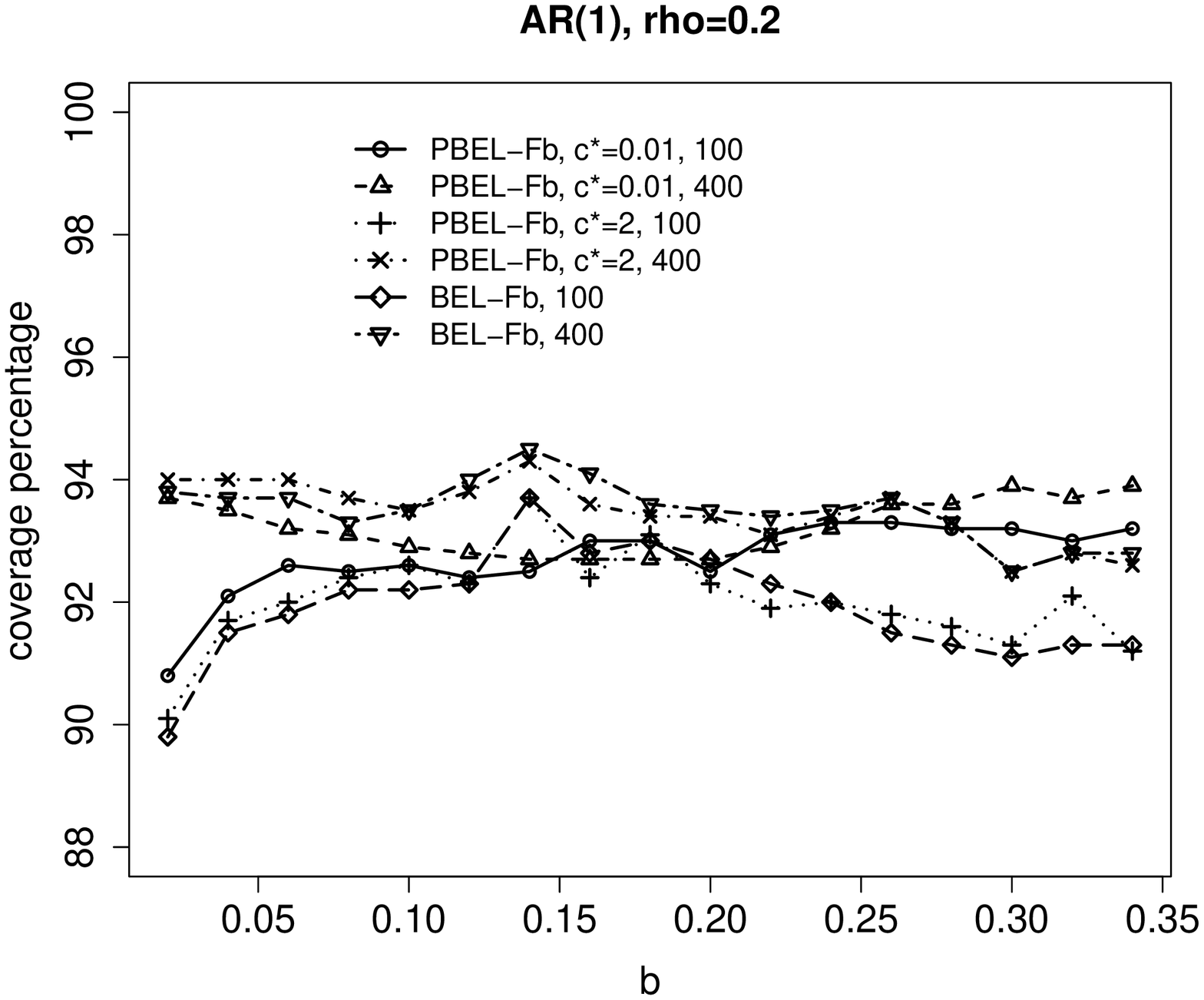}
\includegraphics[height=5.2cm,width=6cm]{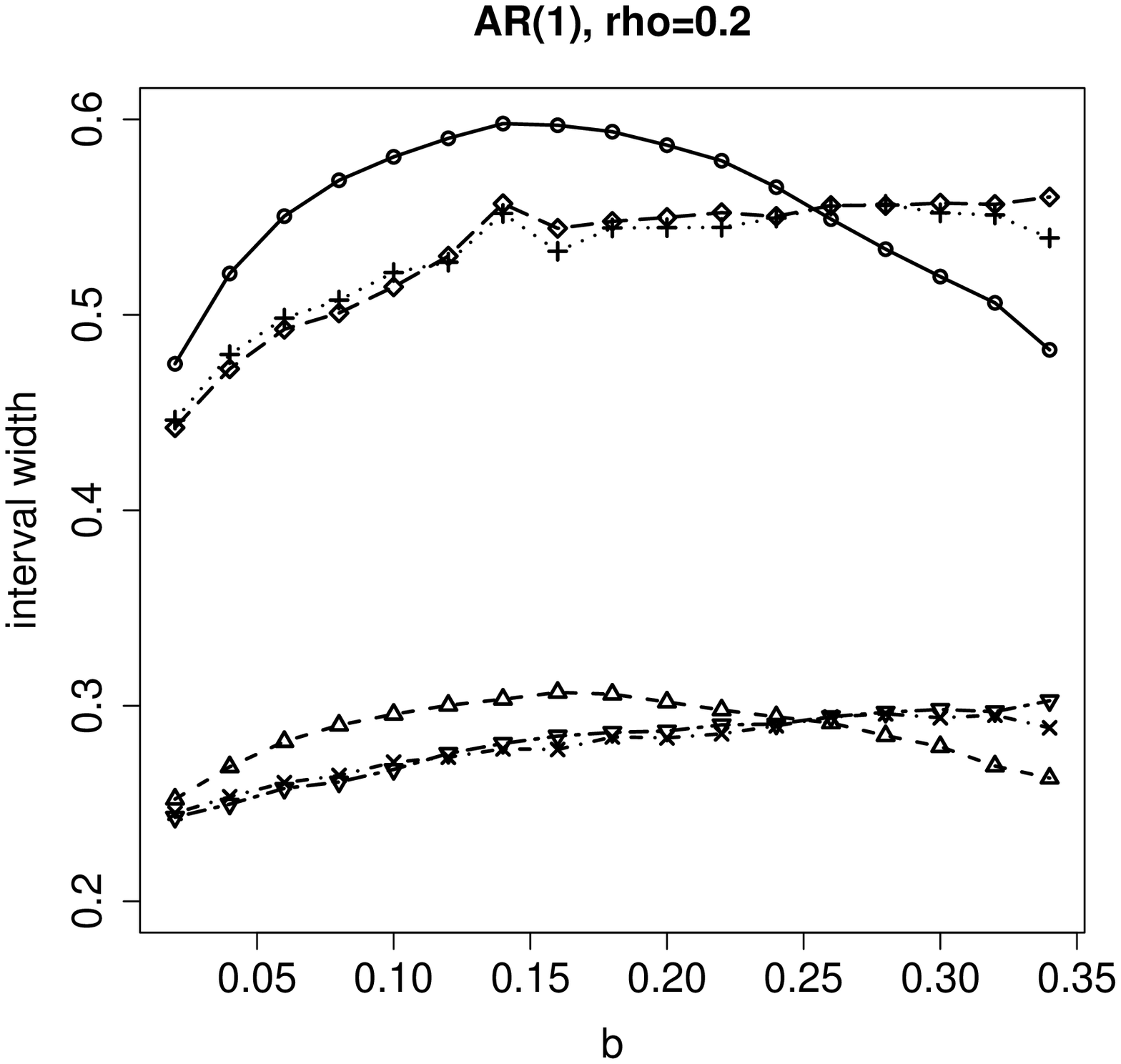}
\includegraphics[height=5.2cm,width=6cm]{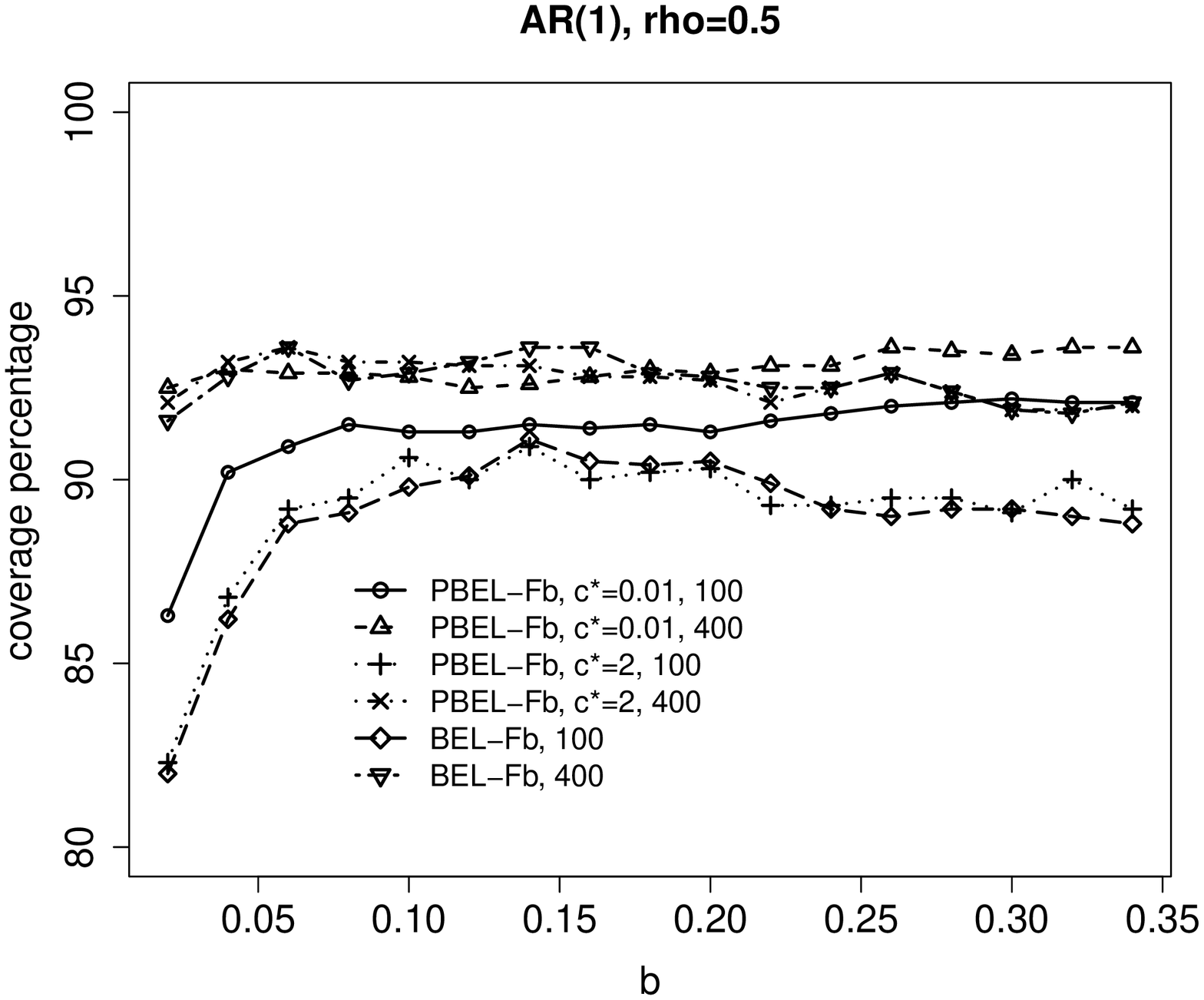}
\includegraphics[height=5.2cm,width=6cm]{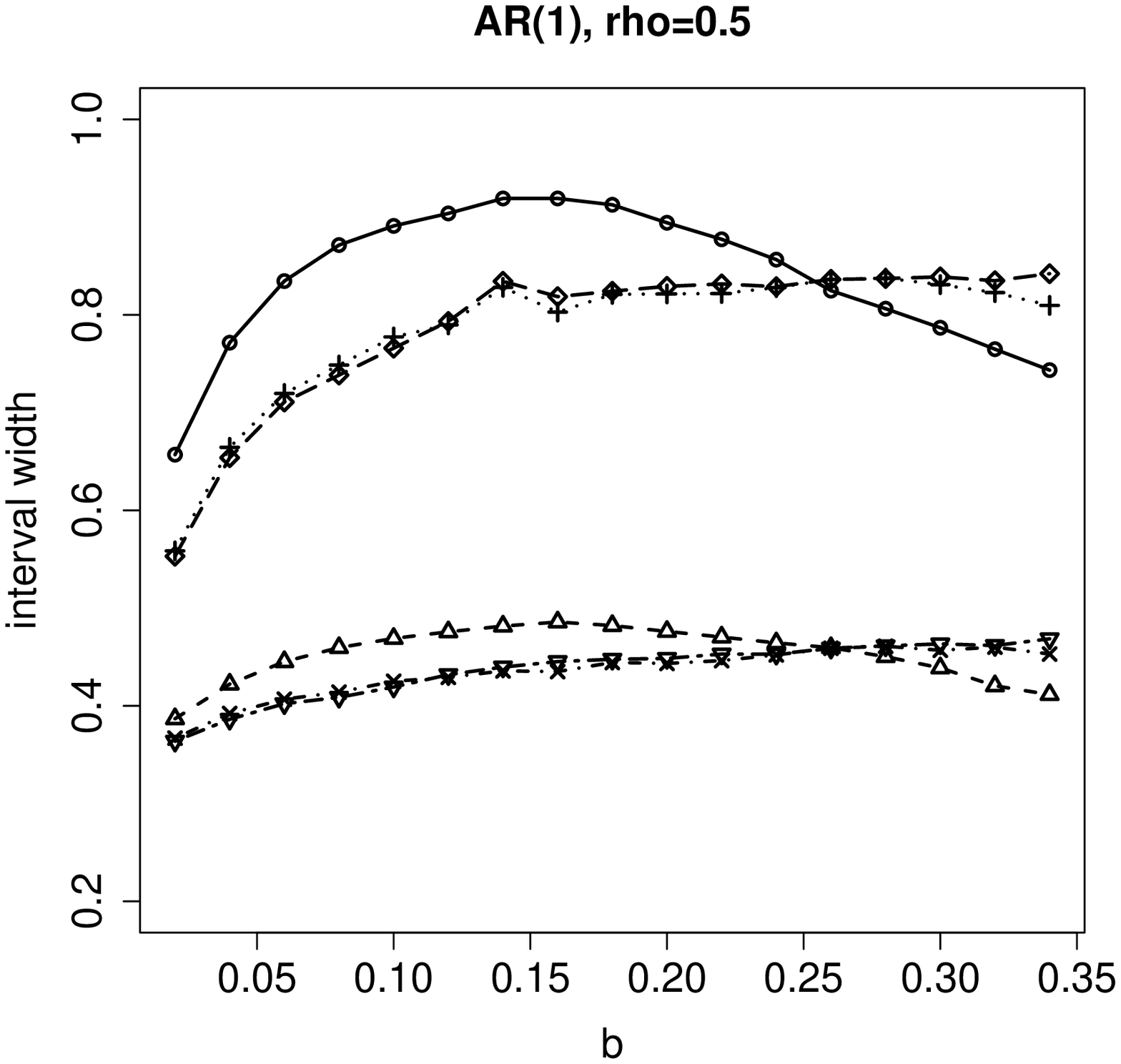}
\includegraphics[height=5.2cm,width=6cm]{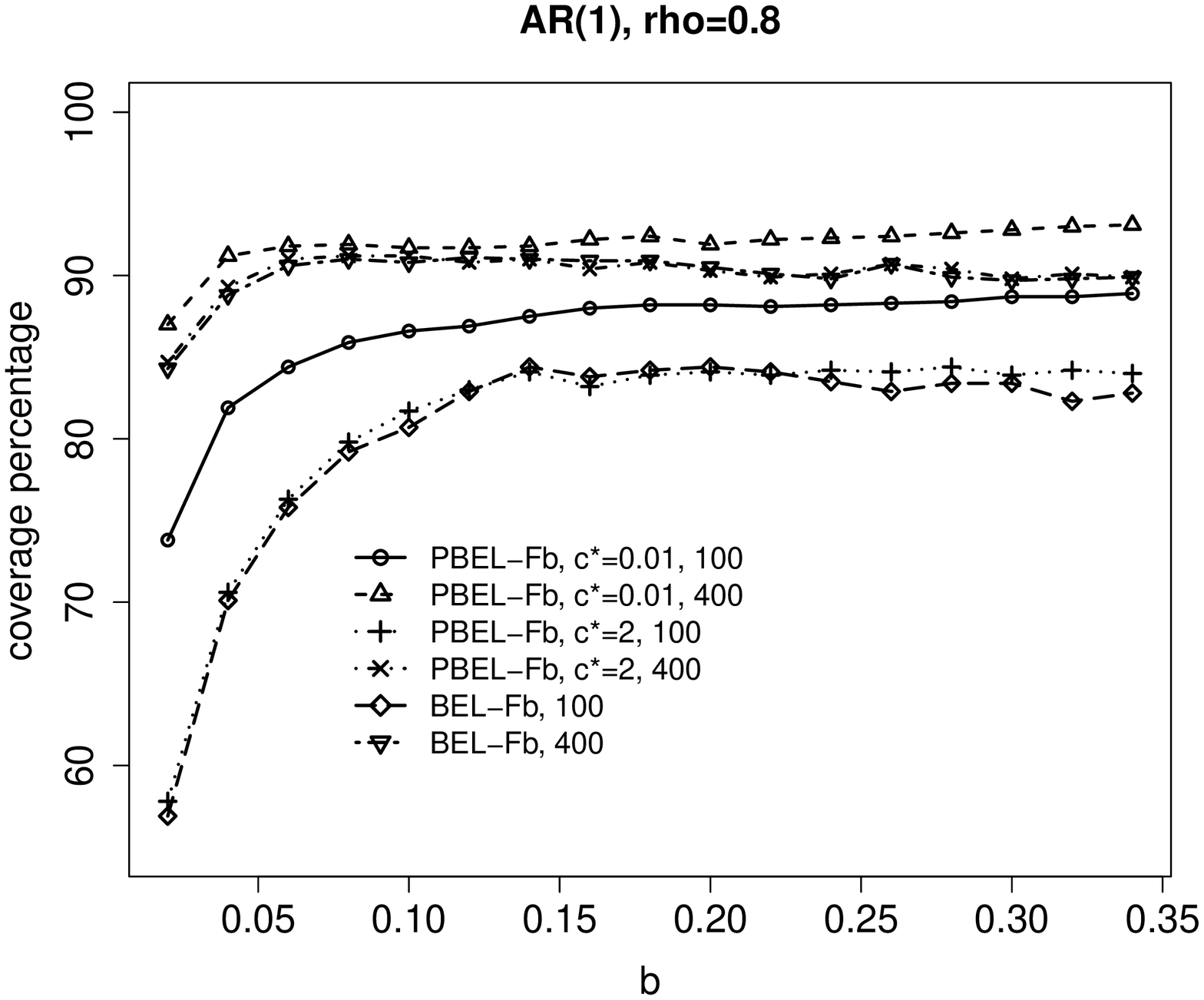}
\includegraphics[height=5.2cm,width=6cm]{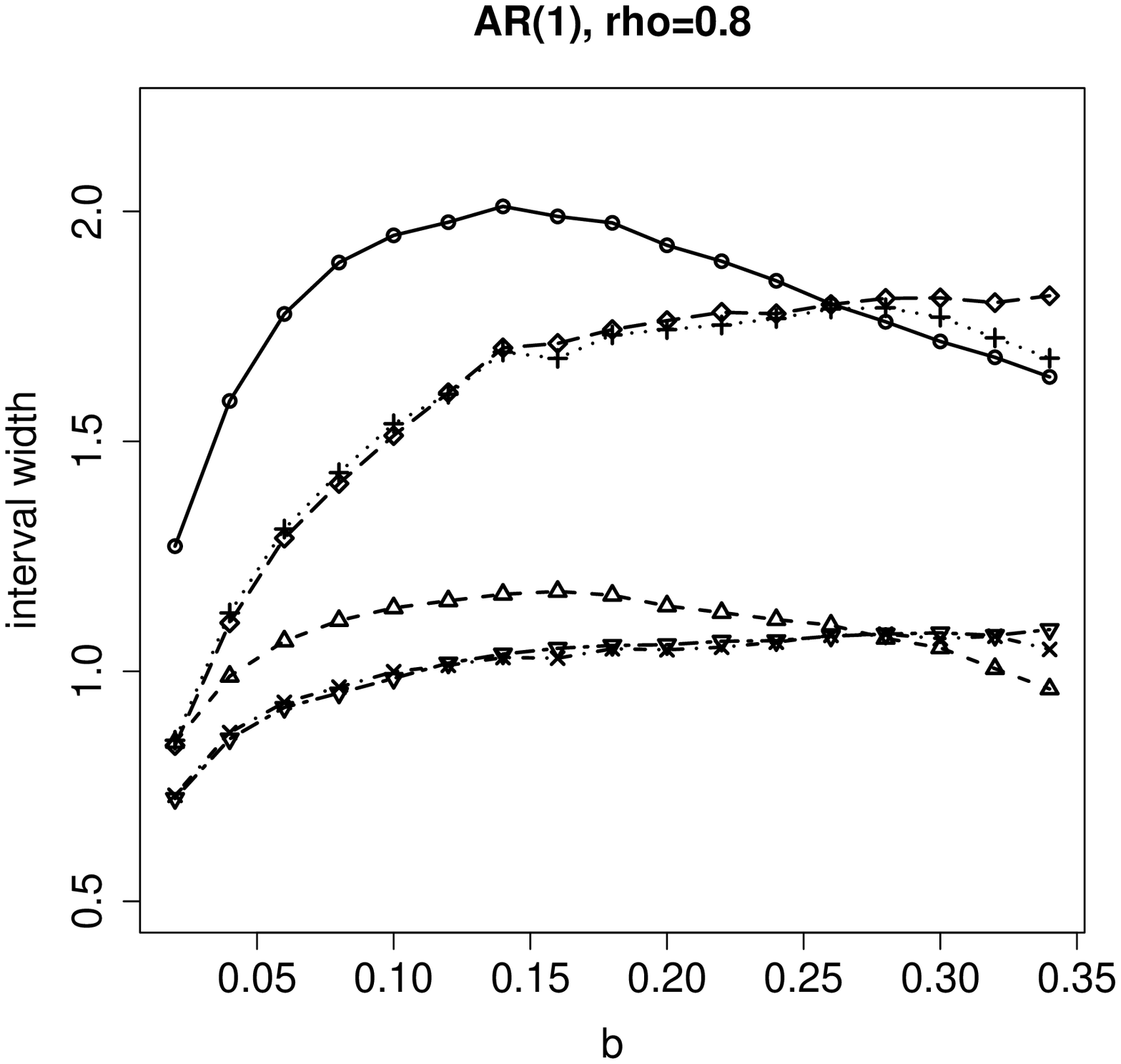}
\includegraphics[height=5.2cm,width=6cm]{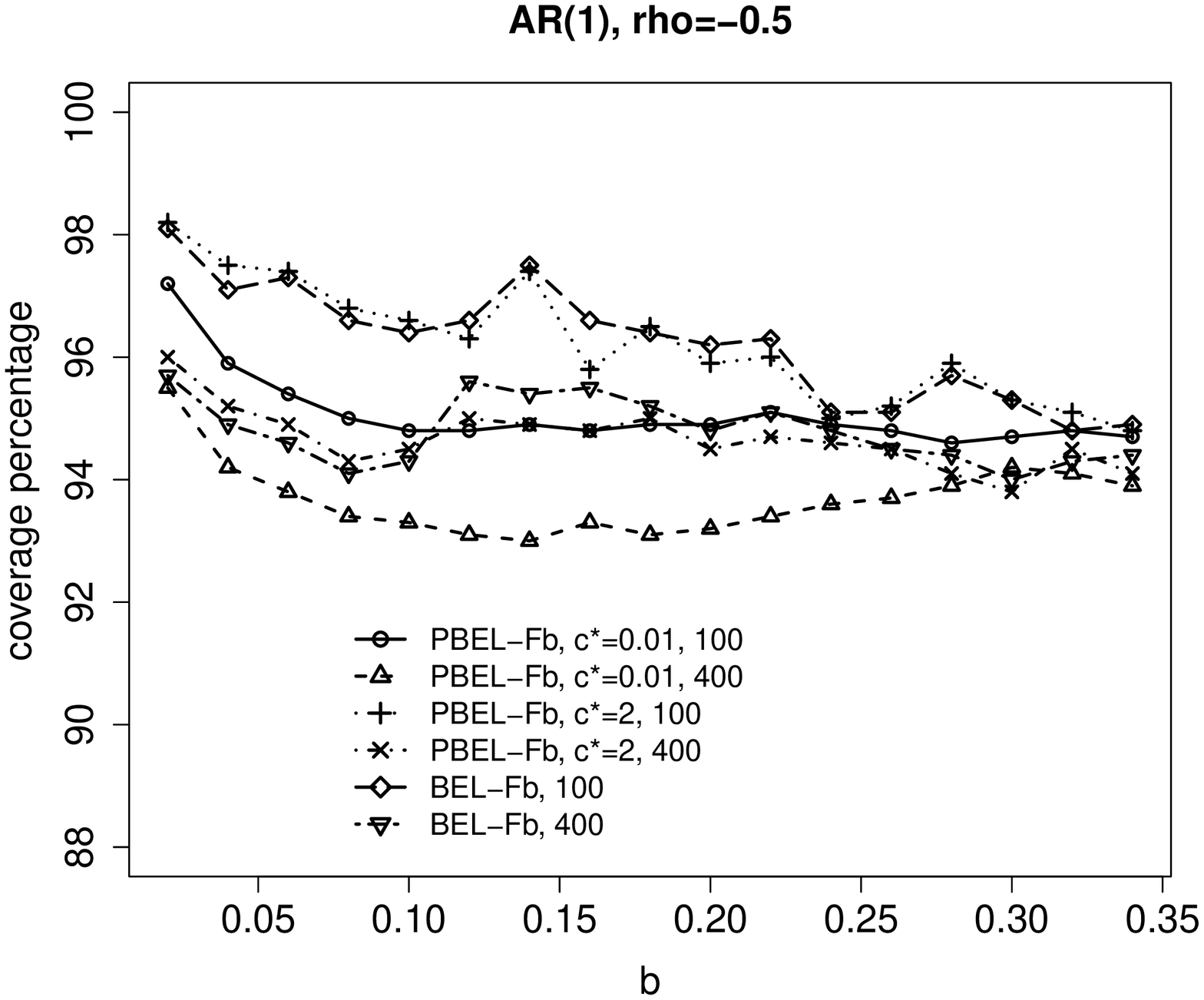}
\includegraphics[height=5.2cm,width=6cm]{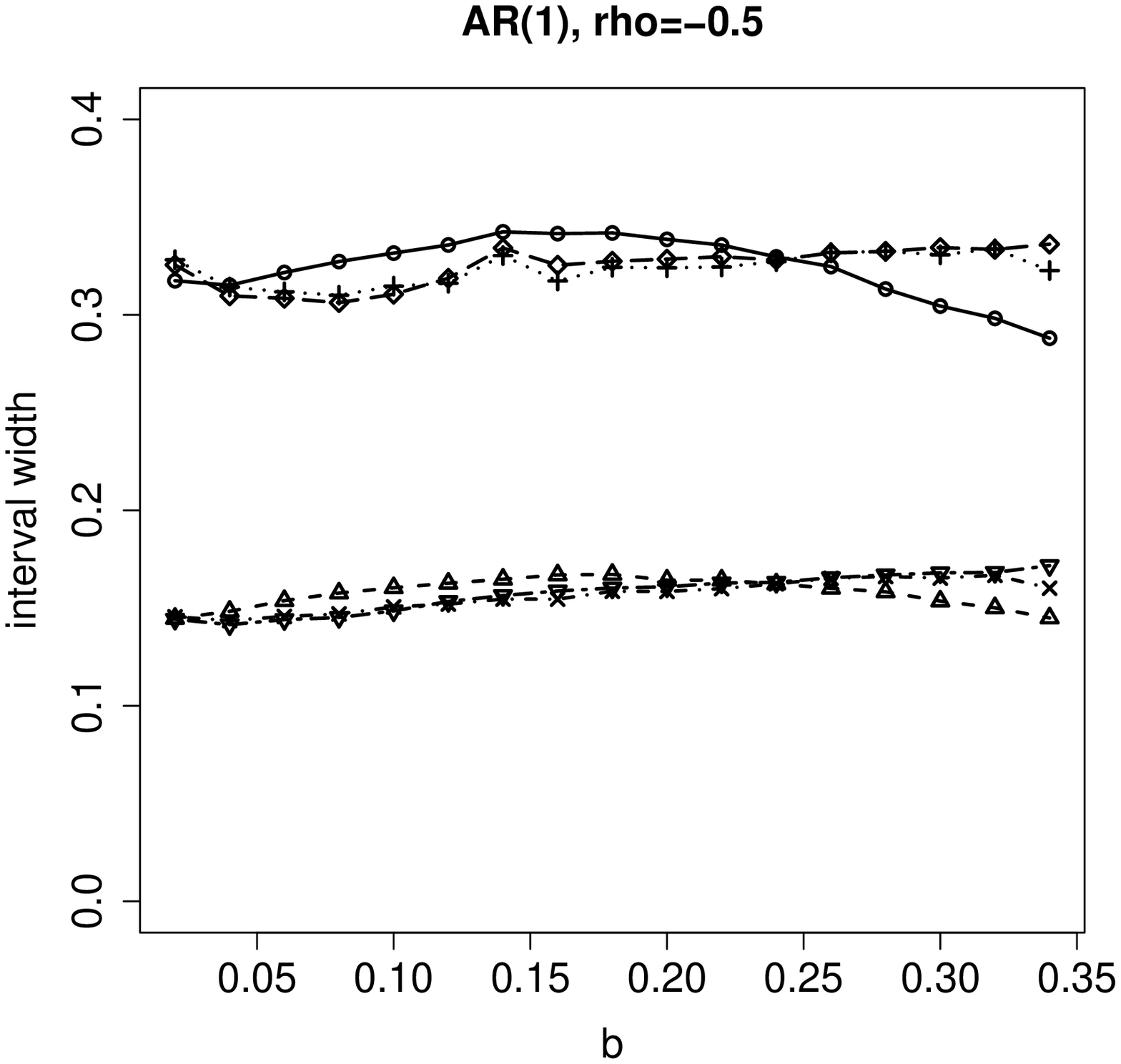}
\caption{Coverage probabilities (left panels) and interval widths
(right panels) for the mean delivered by the PBEL with
$Q(r,s)=(1-|r-s|)\mathbf{I}\{|r-s|\leq 1\}$, and BEL, where $k=1$.
The nominal level is 95\% and the number of Monte Carlo replications
is 1,000.}\label{fig:pel-ar-k1}
\end{figure}

\newpage

\begin{figure}[H]
\centering
\includegraphics[height=5.2cm,width=6cm]{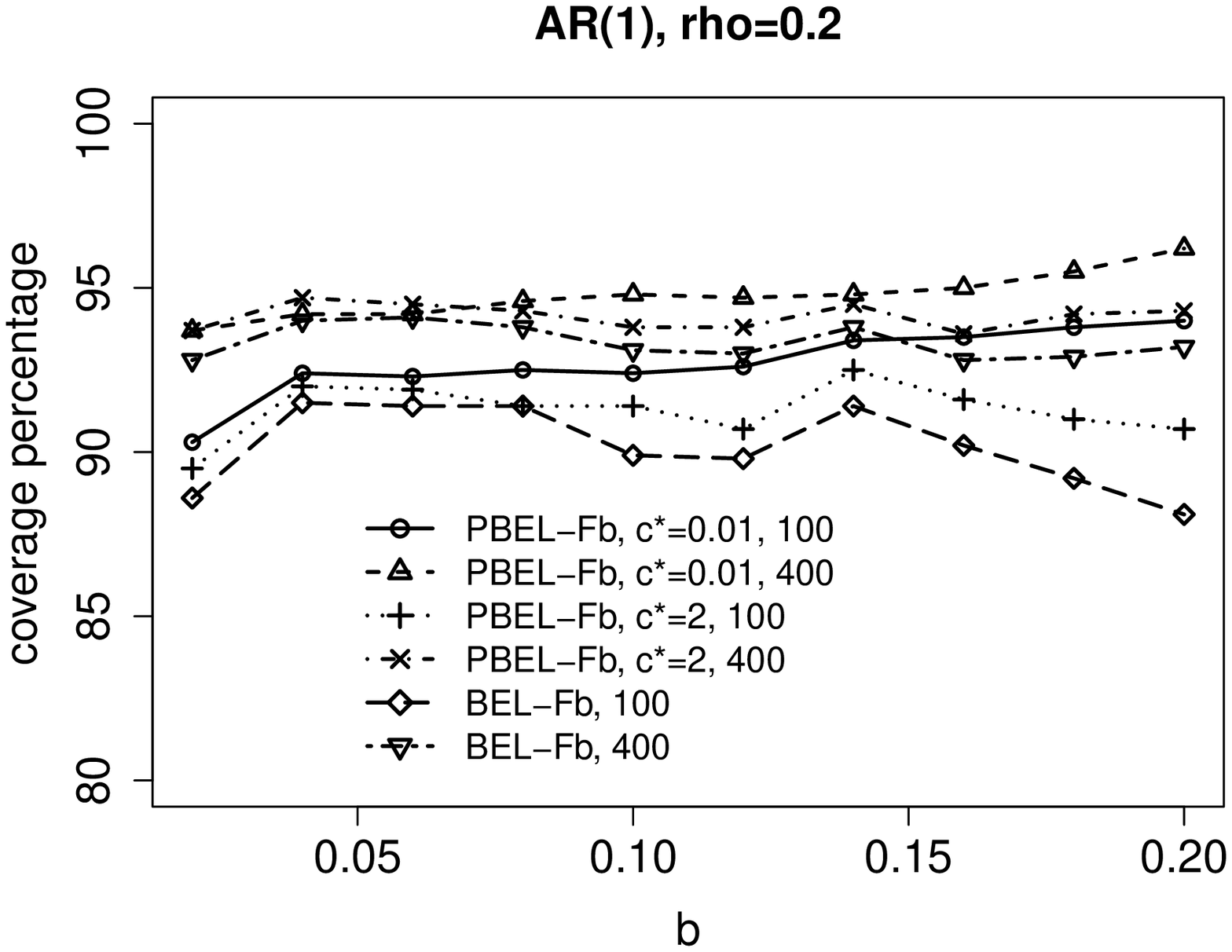}
\includegraphics[height=5.2cm,width=6cm]{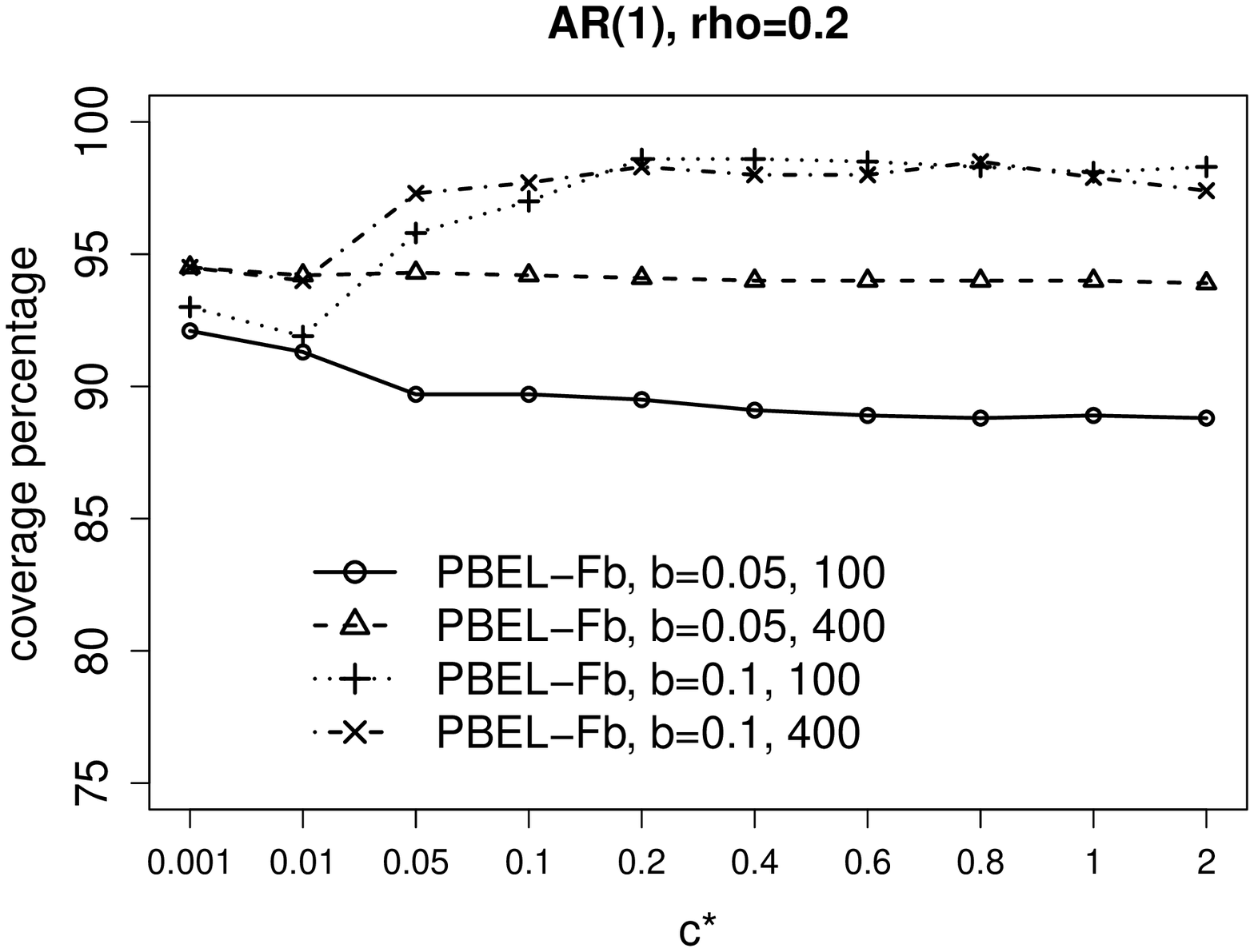}
\includegraphics[height=5.2cm,width=6cm]{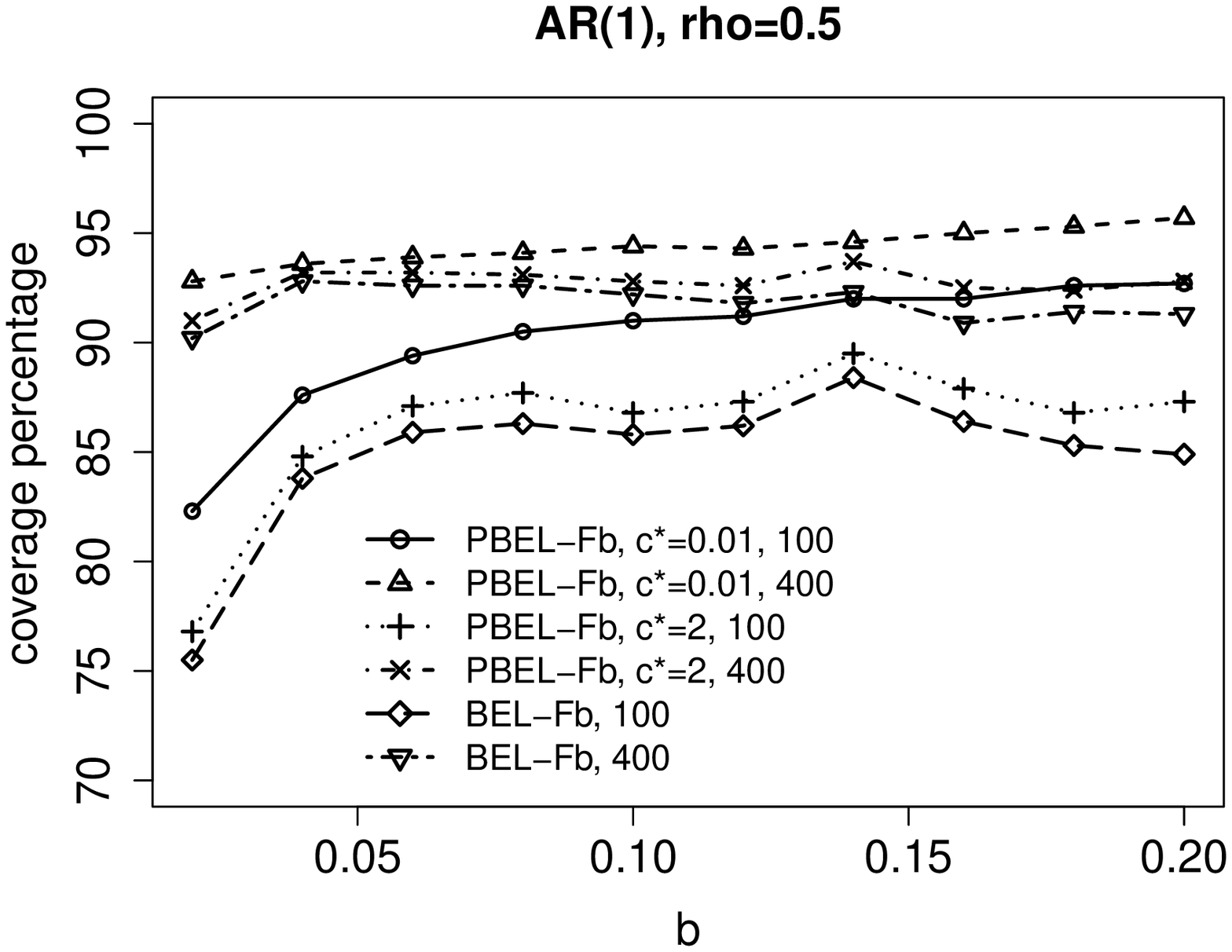}
\includegraphics[height=5.2cm,width=6cm]{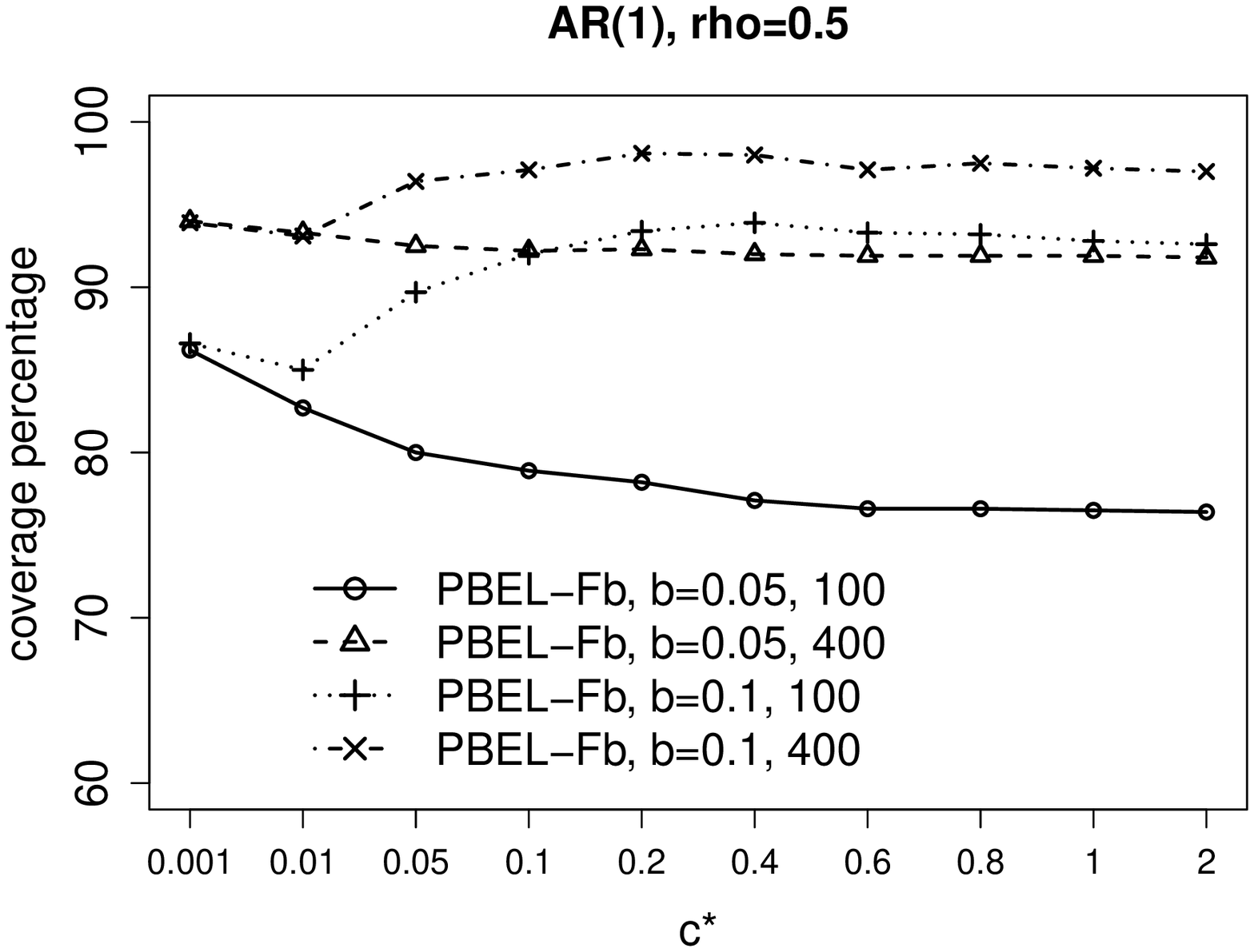}
\includegraphics[height=5.2cm,width=6cm]{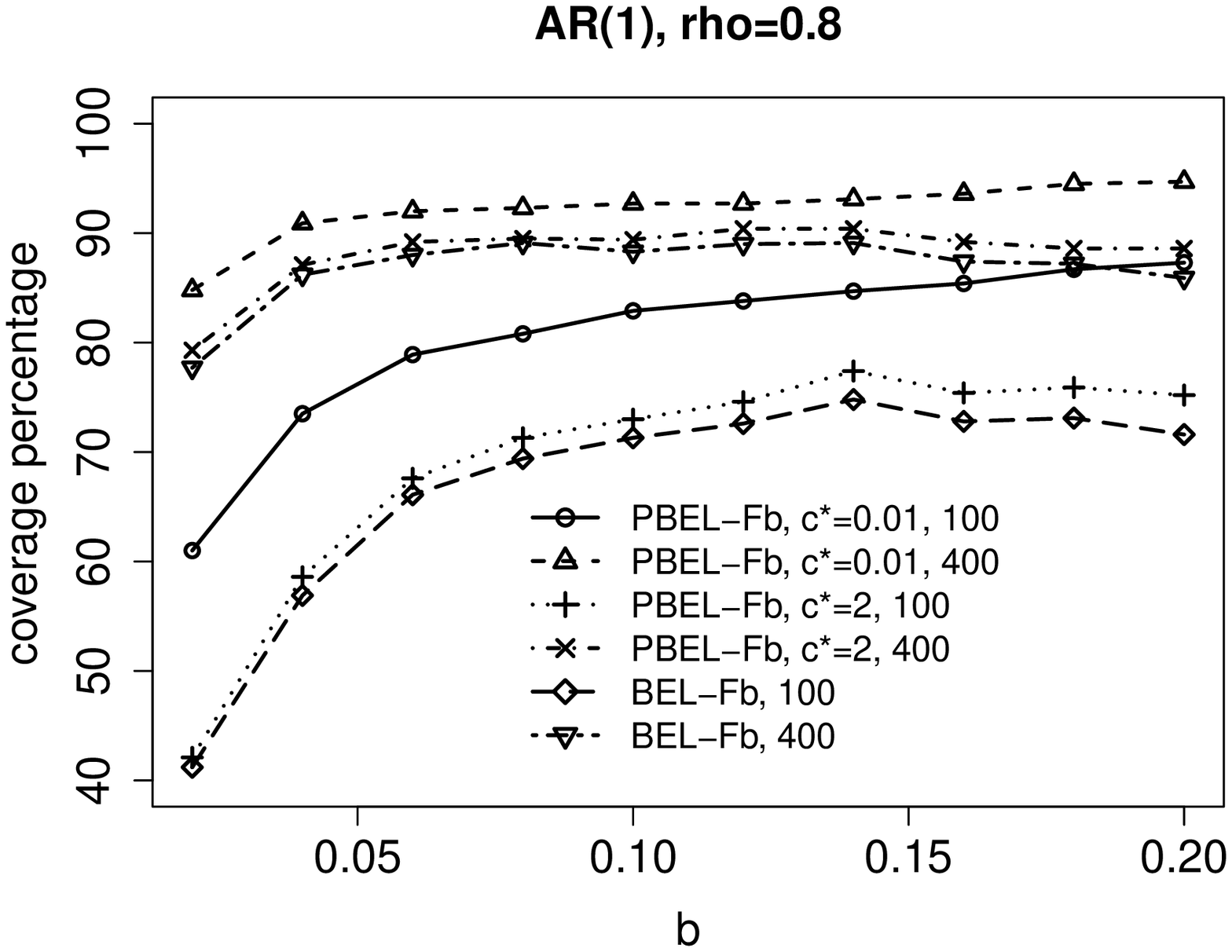}
\includegraphics[height=5.2cm,width=6cm]{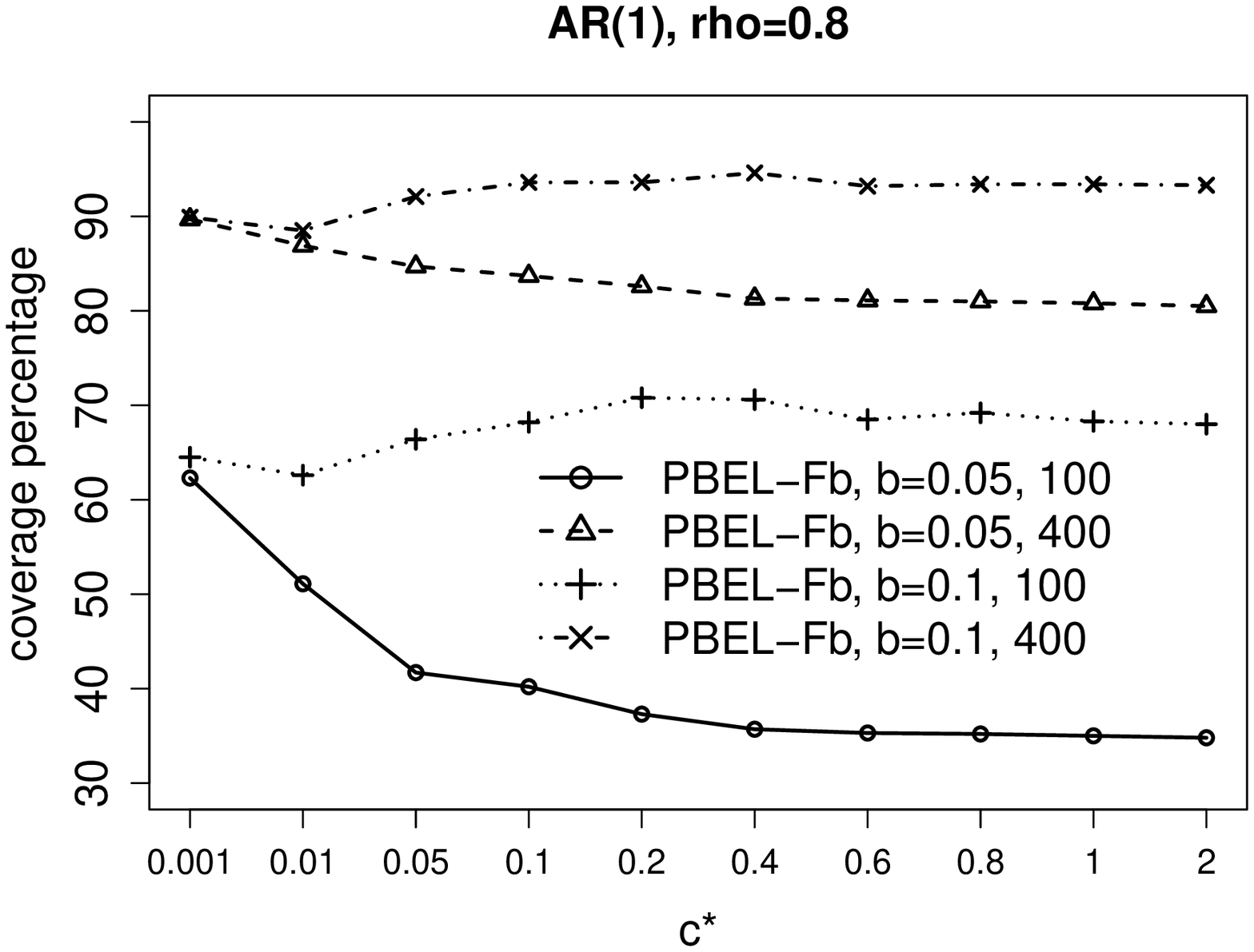}
\includegraphics[height=5.2cm,width=6cm]{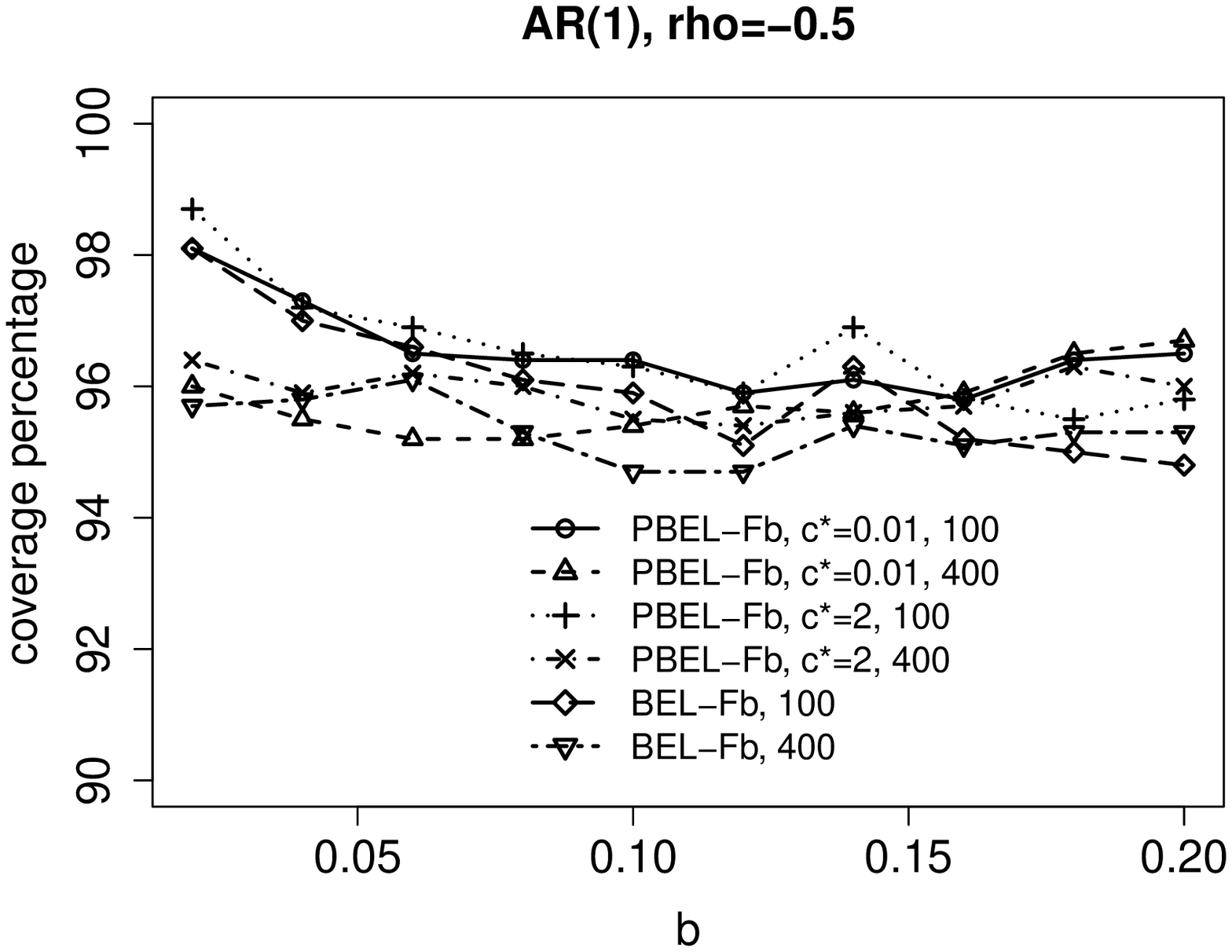}
\includegraphics[height=5.2cm,width=6cm]{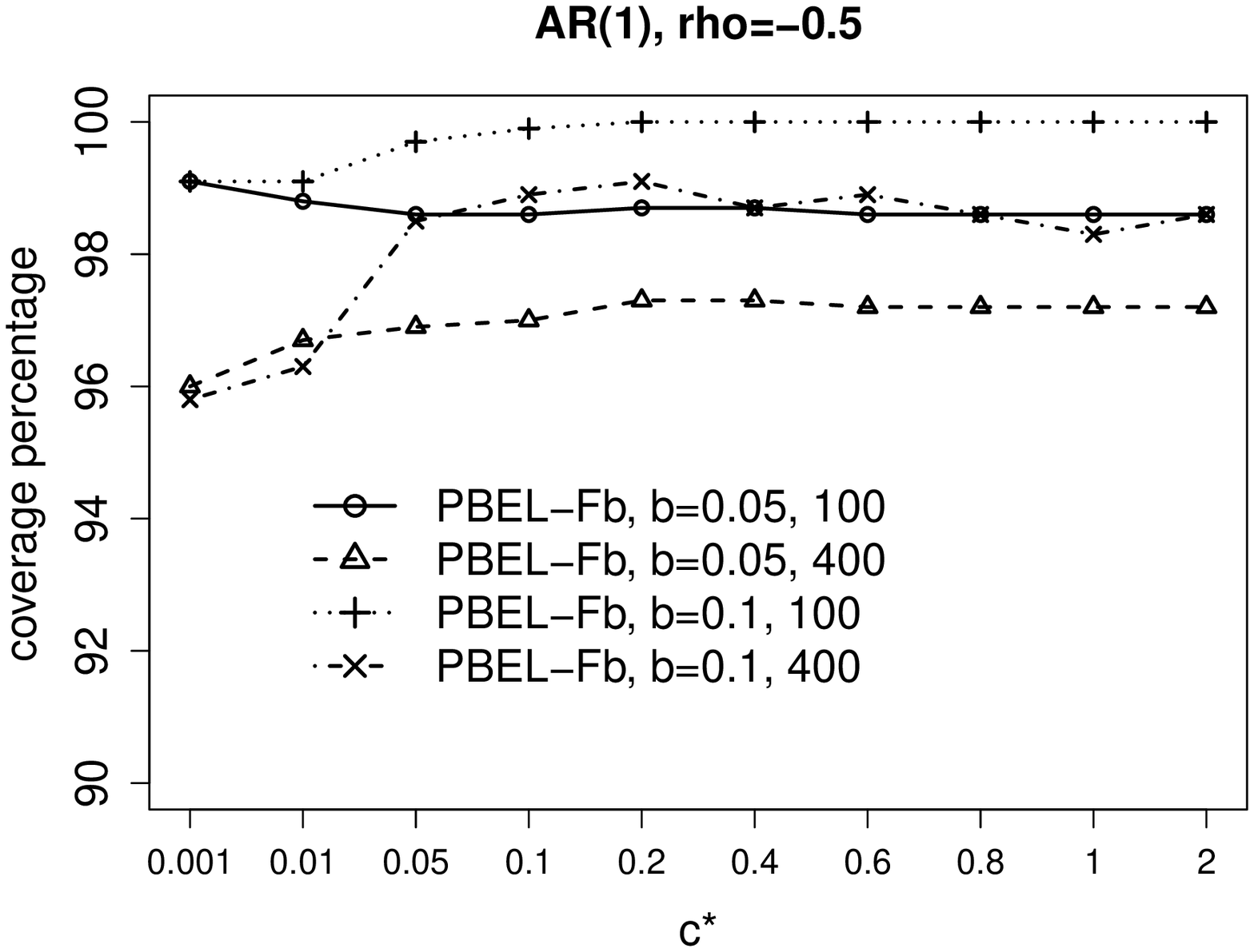}
\caption{Coverage probabilities for the mean delivered by the PBEL
with $Q(r,s)=(1-|r-s|)\mathbf{I}\{|r-s|\leq 1\}$, and BEL, where
$k=2$ for the left column and $k=5$ for the right column. The
nominal level is 95\% and the number of Monte Carlo replications is
1,000.}\label{fig:pel-ar}
\end{figure}


\newpage

\begin{figure}[H]
\centering
\includegraphics[height=5.2cm,width=6cm]{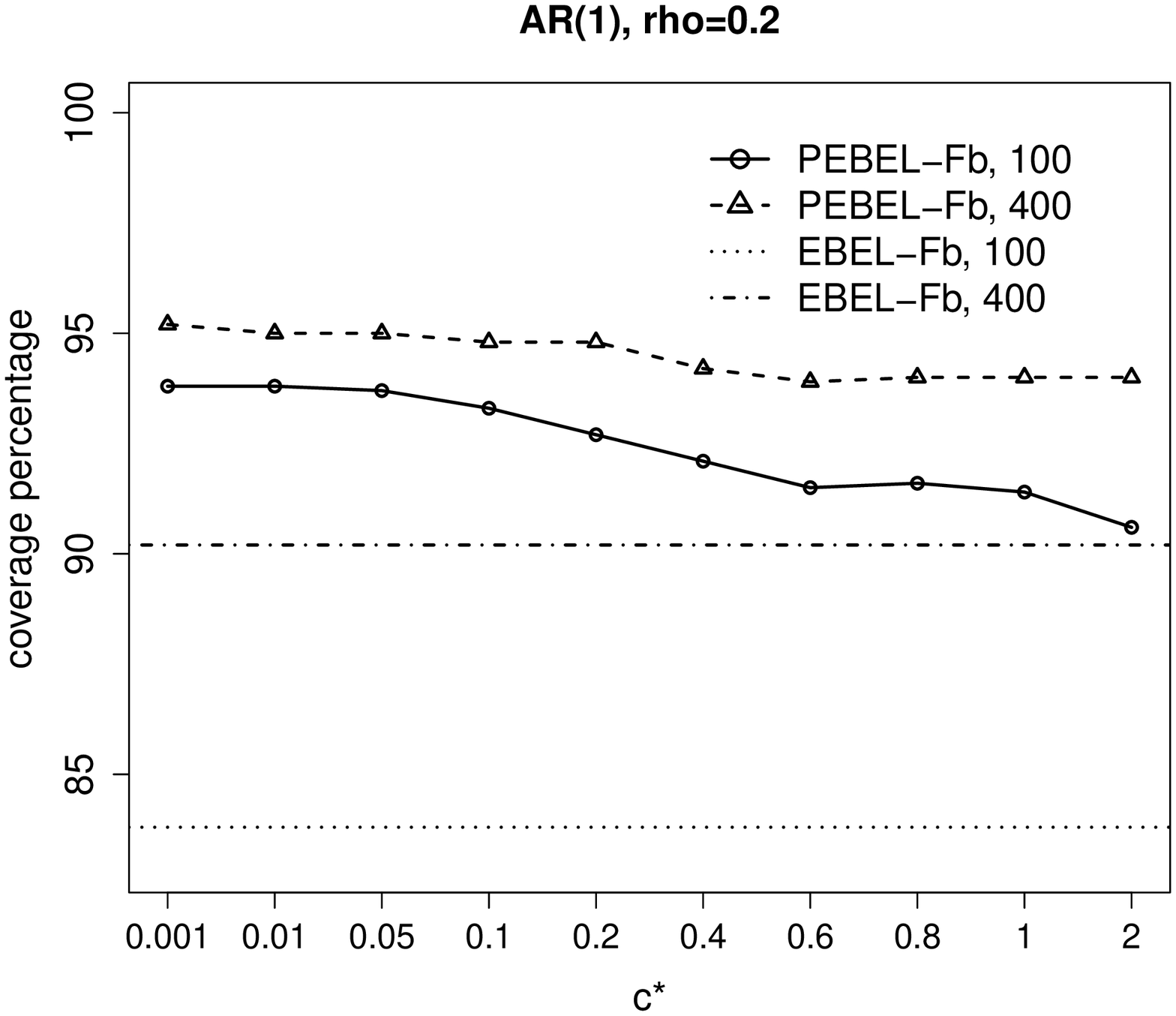}
\includegraphics[height=5.2cm,width=6cm]{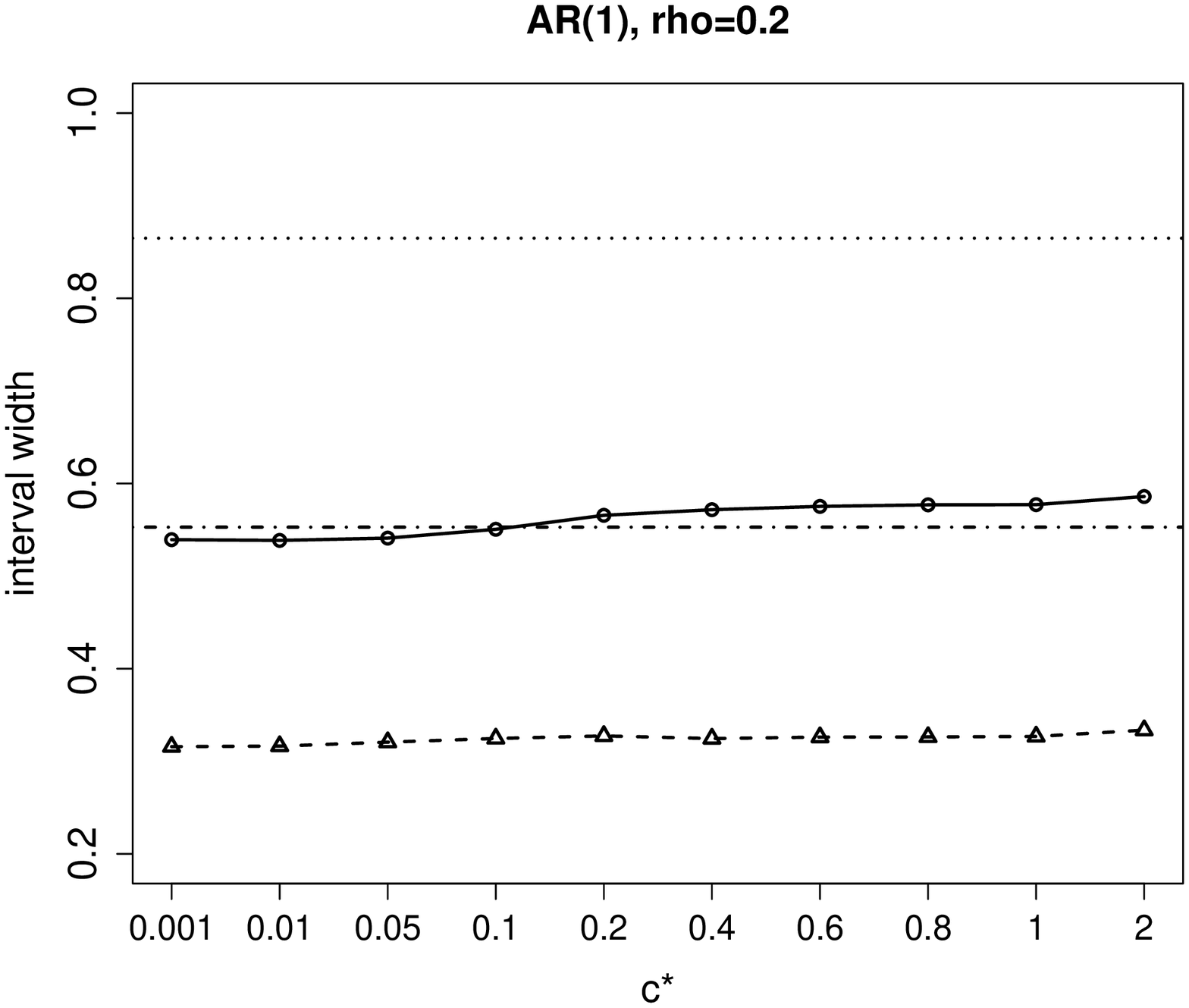}
\includegraphics[height=5.2cm,width=6cm]{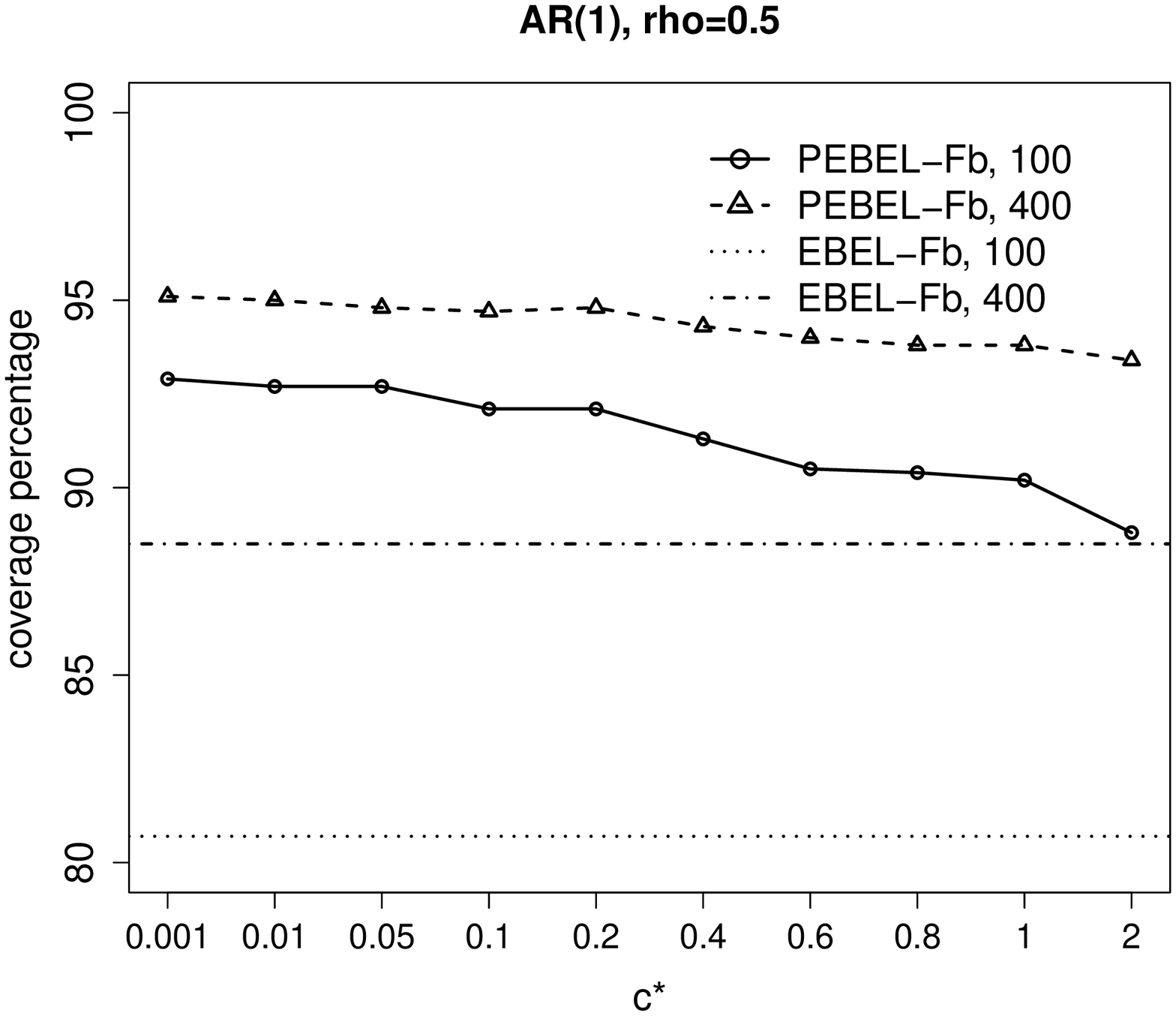}
\includegraphics[height=5.2cm,width=6cm]{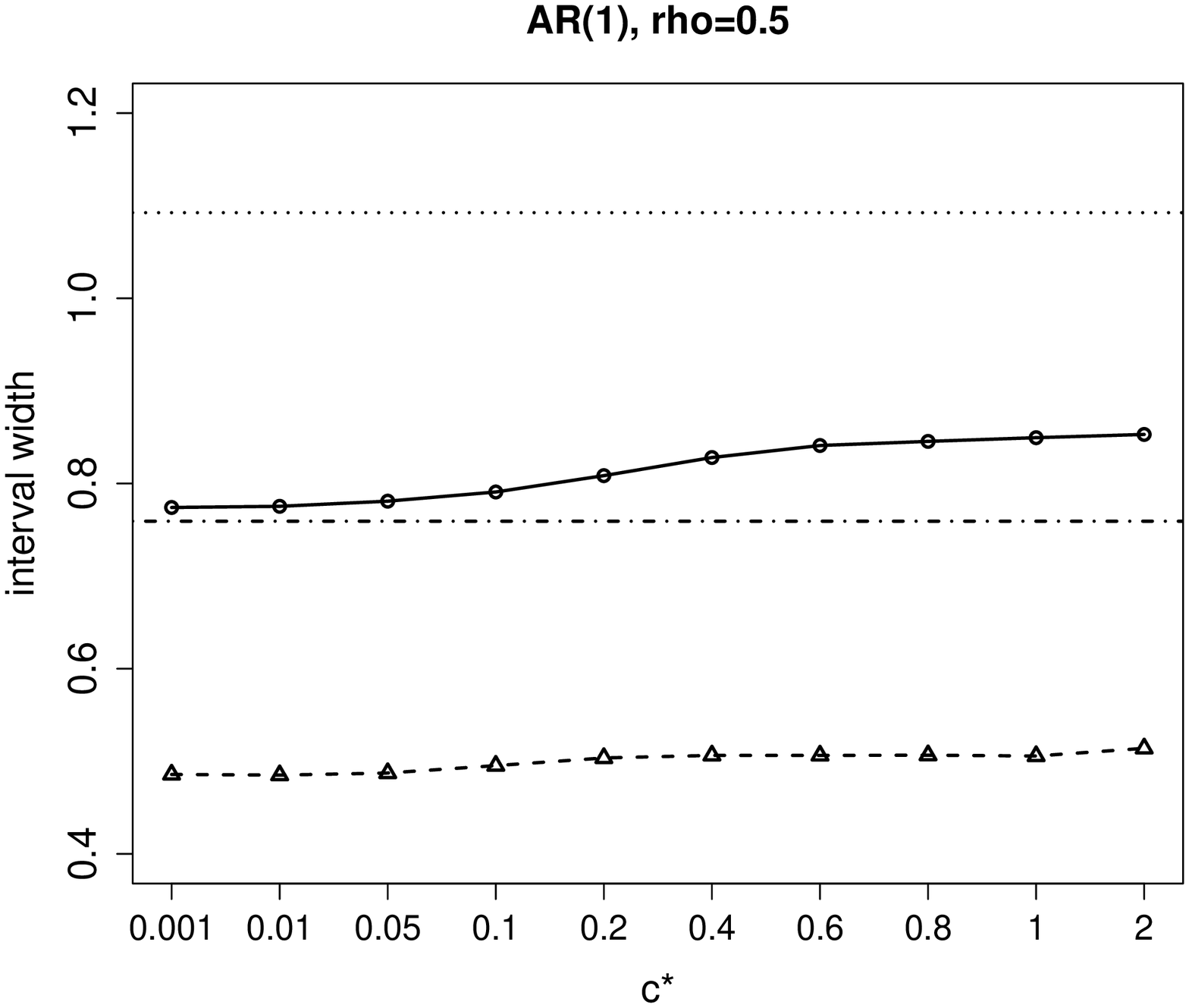}
\includegraphics[height=5.2cm,width=6cm]{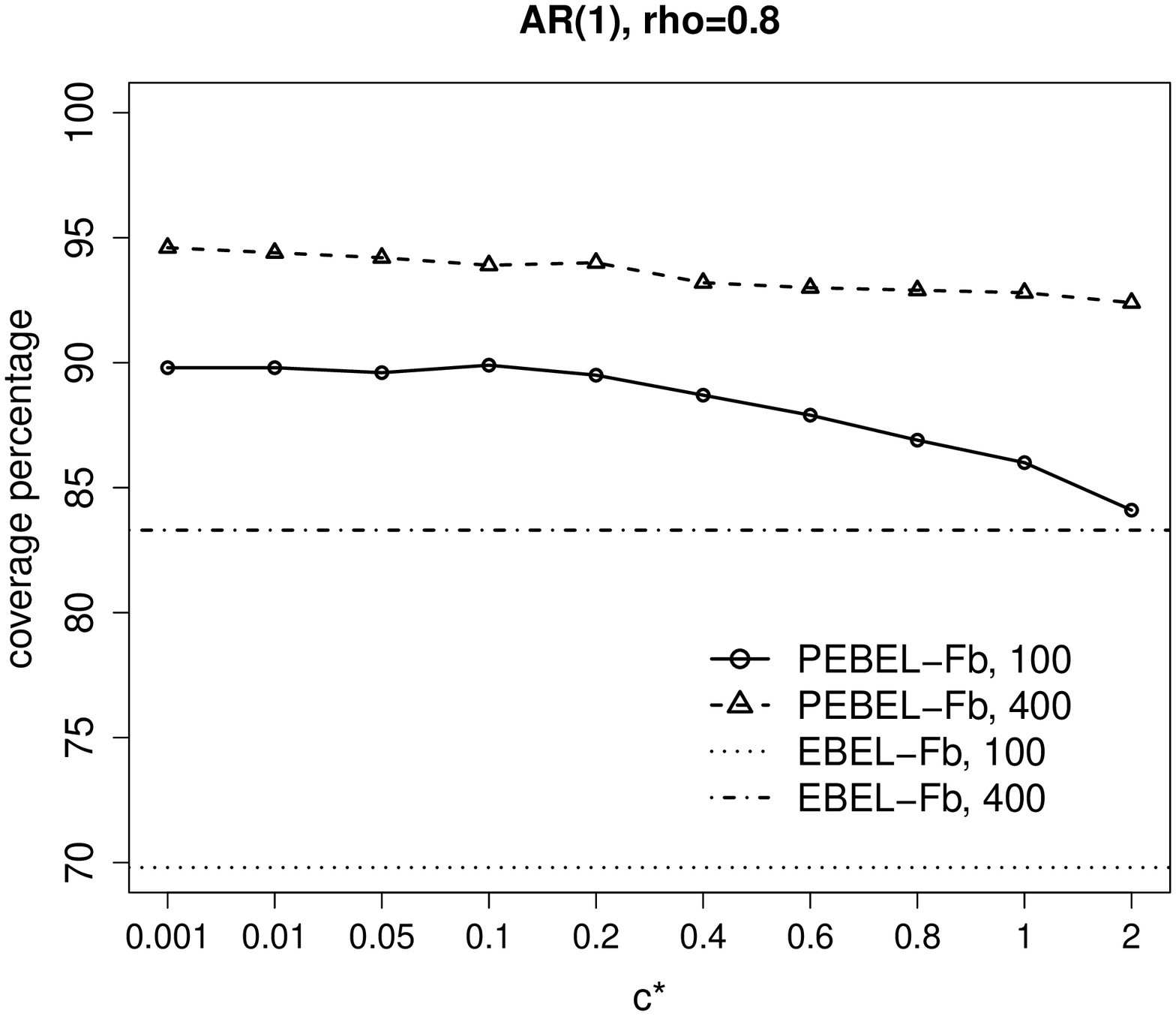}
\includegraphics[height=5.2cm,width=6cm]{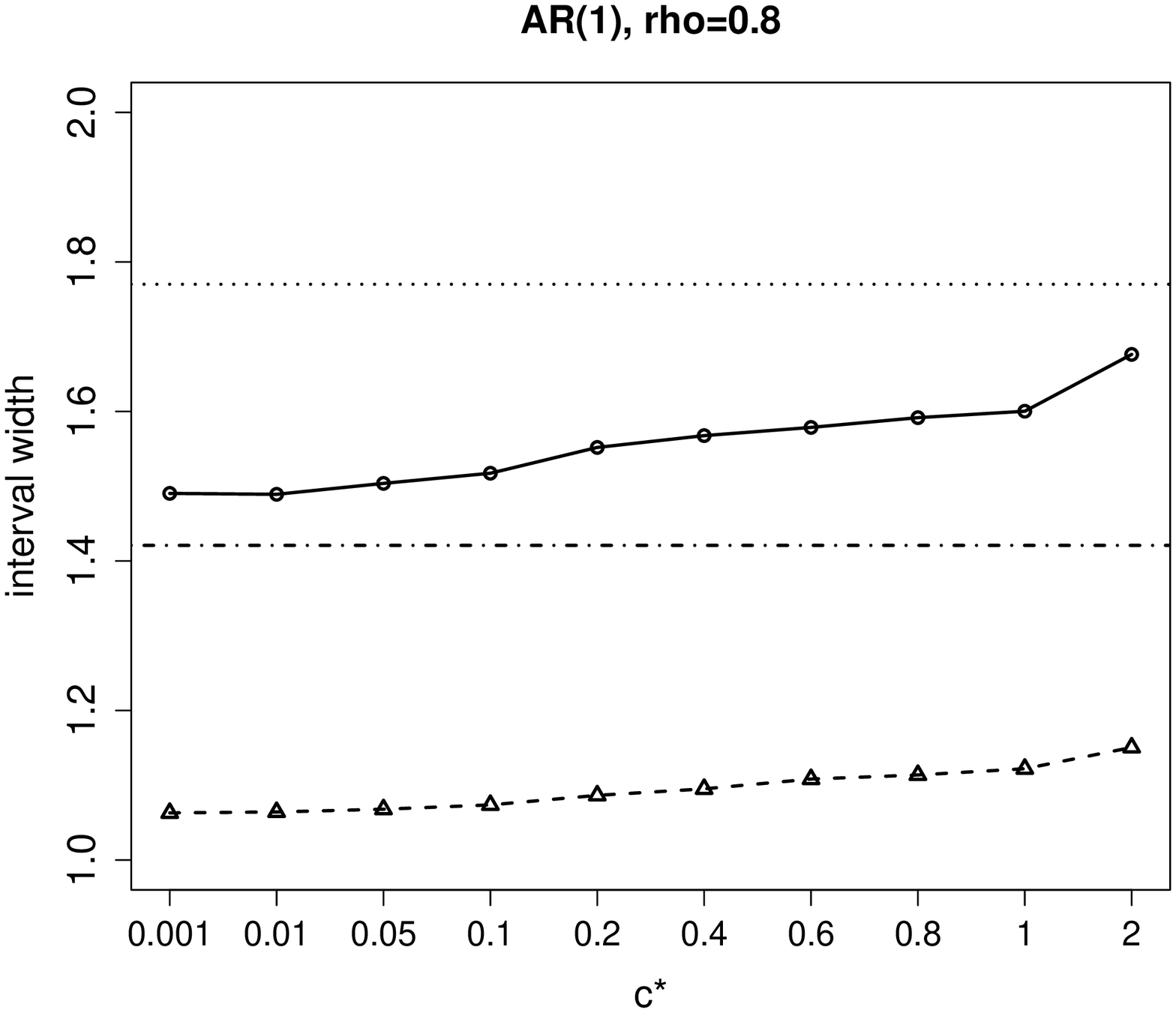}
\includegraphics[height=5.2cm,width=6cm]{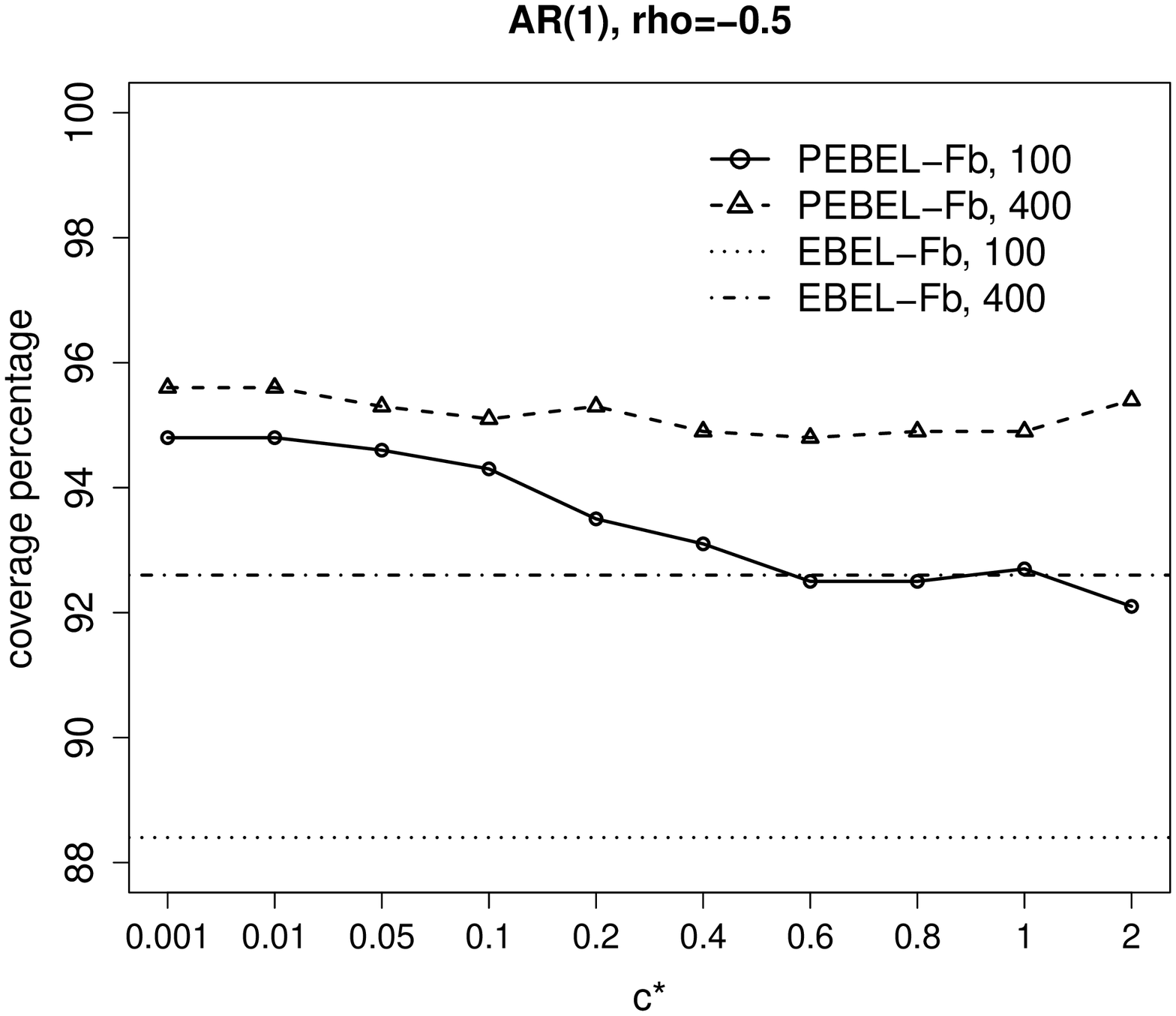}
\includegraphics[height=5.2cm,width=6cm]{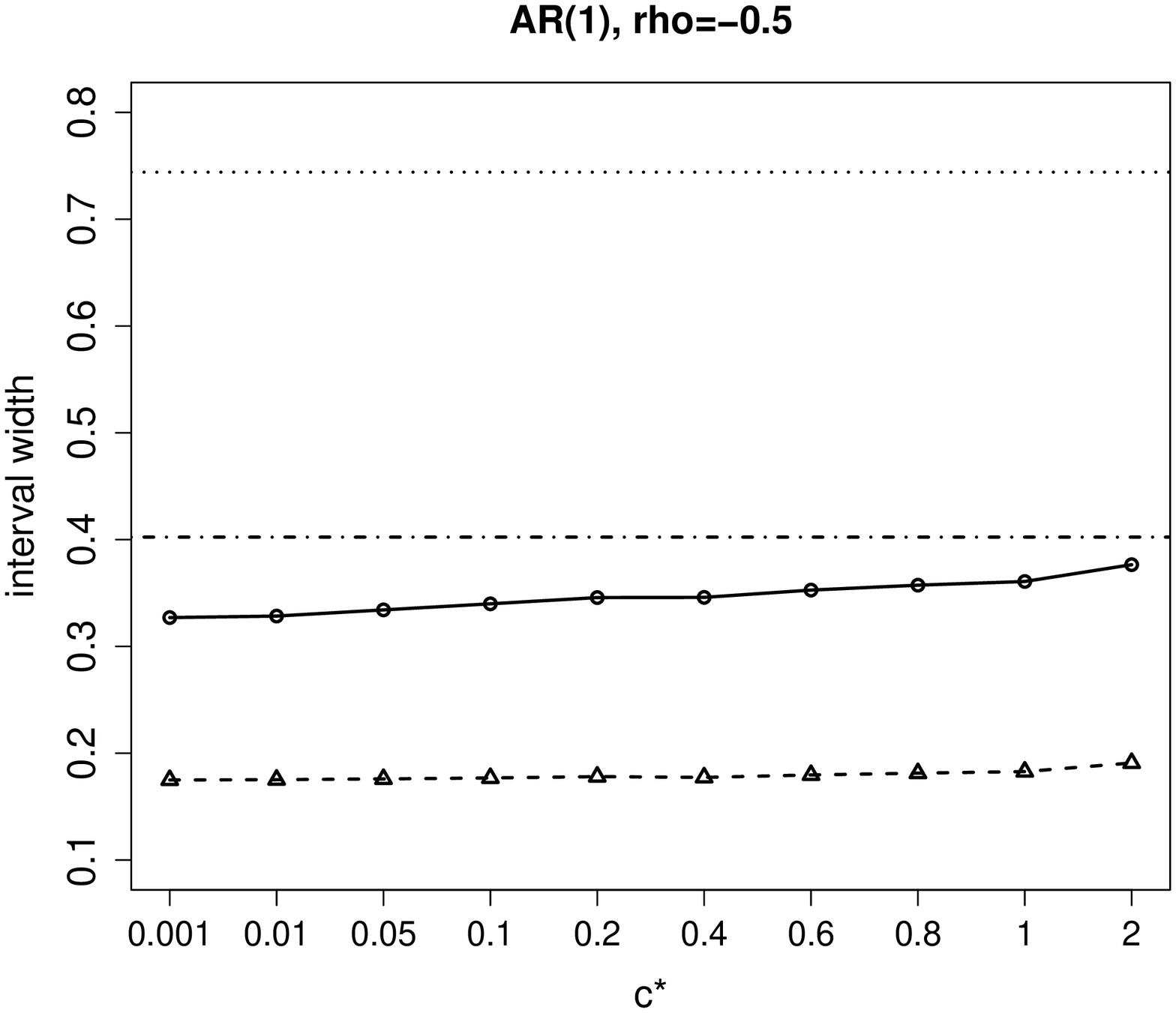}
\caption{Coverage probabilities (left panels) and interval widths
(right panels) for the mean delivered by the PEBEL with various
$c^*$ and $Q(r,s)=(1-|r-s|)\mathbf{I}\{|r-s|\leq 1\}$, and EBEL,
where $k=1$. The nominal level is 95\% and the number of Monte Carlo
replications is 1,000.}\label{fig:pebel-ar-k1}
\end{figure}

\newpage

\begin{figure}[H]
\centering
\includegraphics[height=5.2cm,width=6cm]{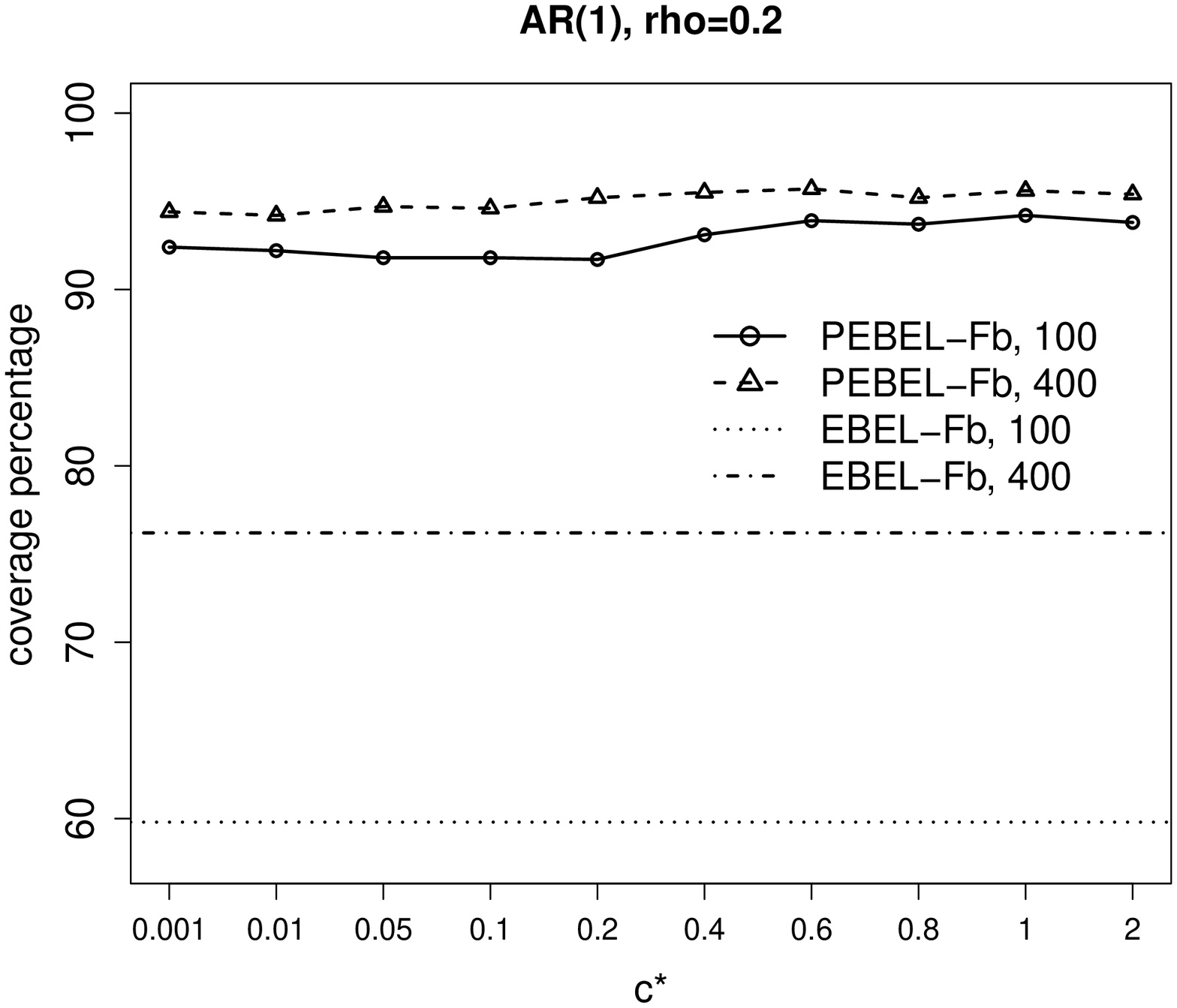}
\includegraphics[height=5.2cm,width=6cm]{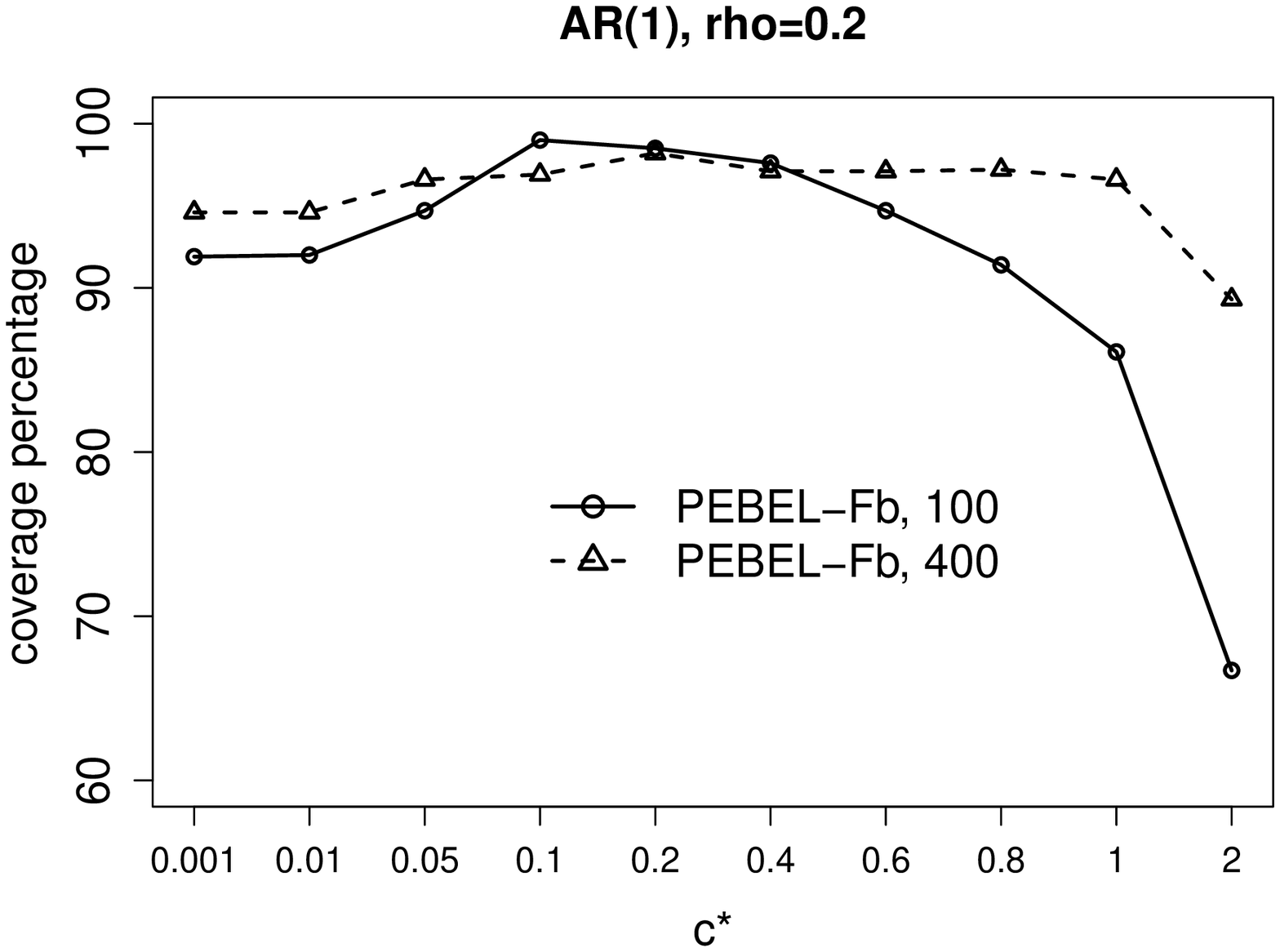}
\includegraphics[height=5.2cm,width=6cm]{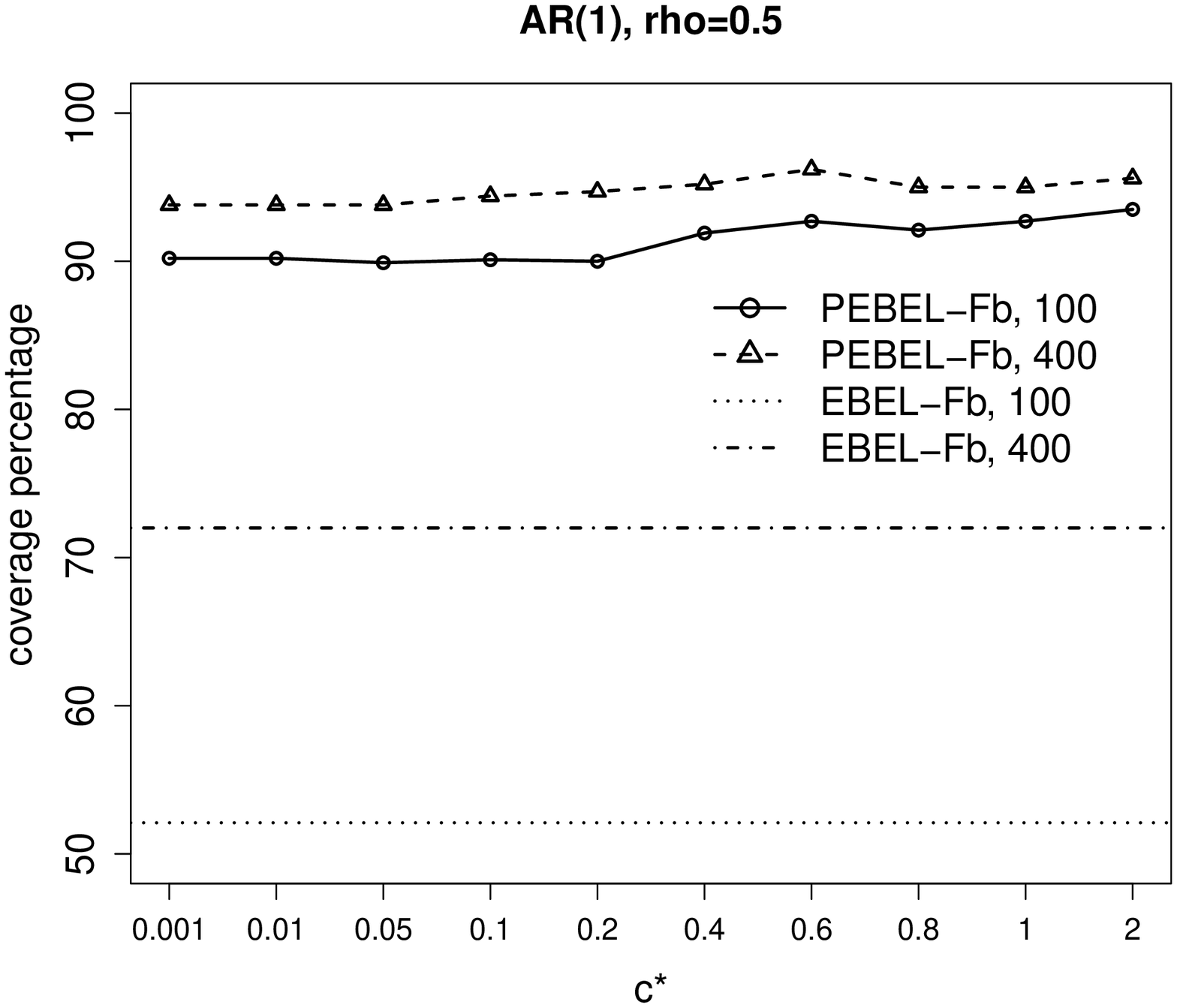}
\includegraphics[height=5.2cm,width=6cm]{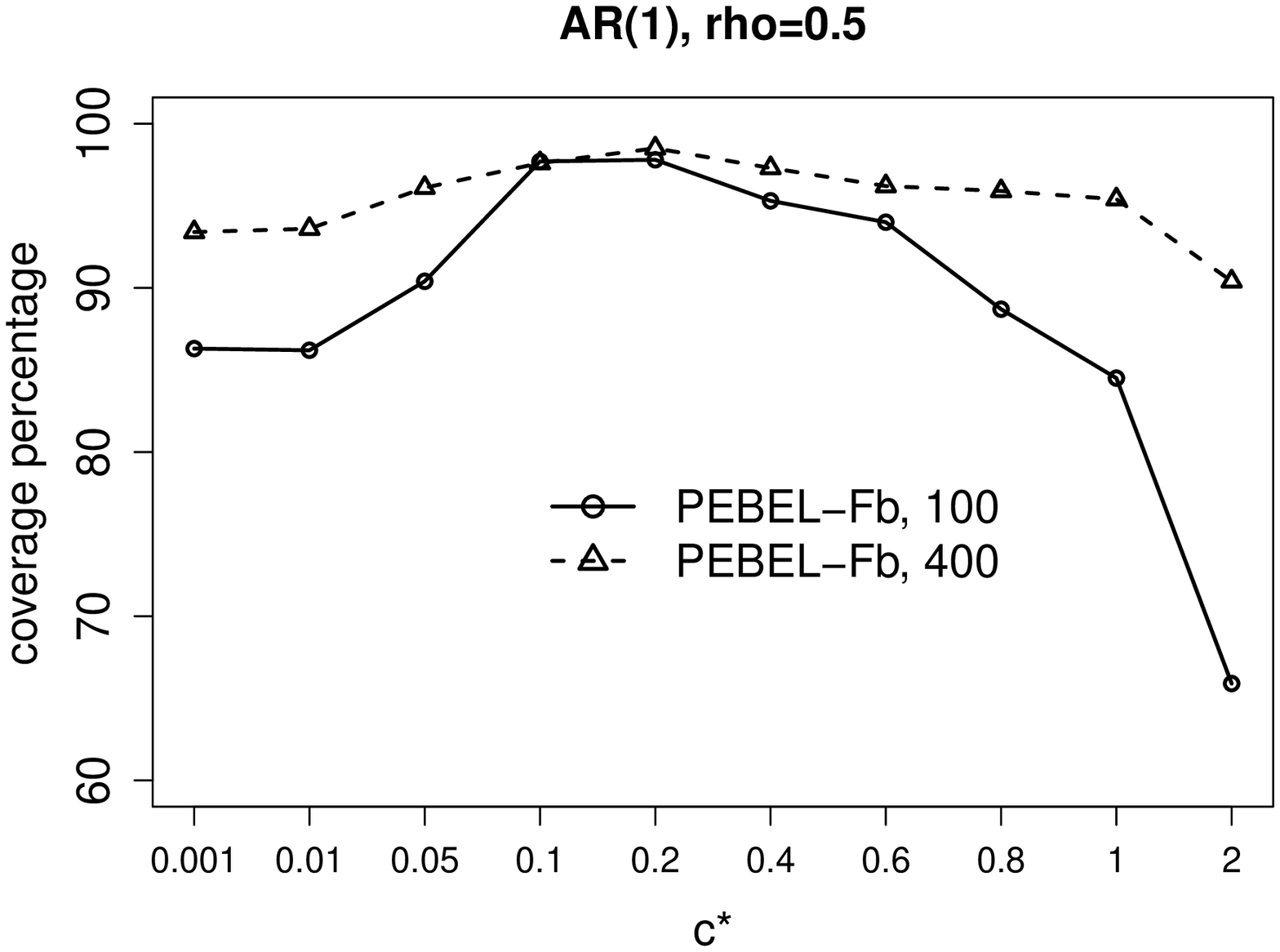}
\includegraphics[height=5.2cm,width=6cm]{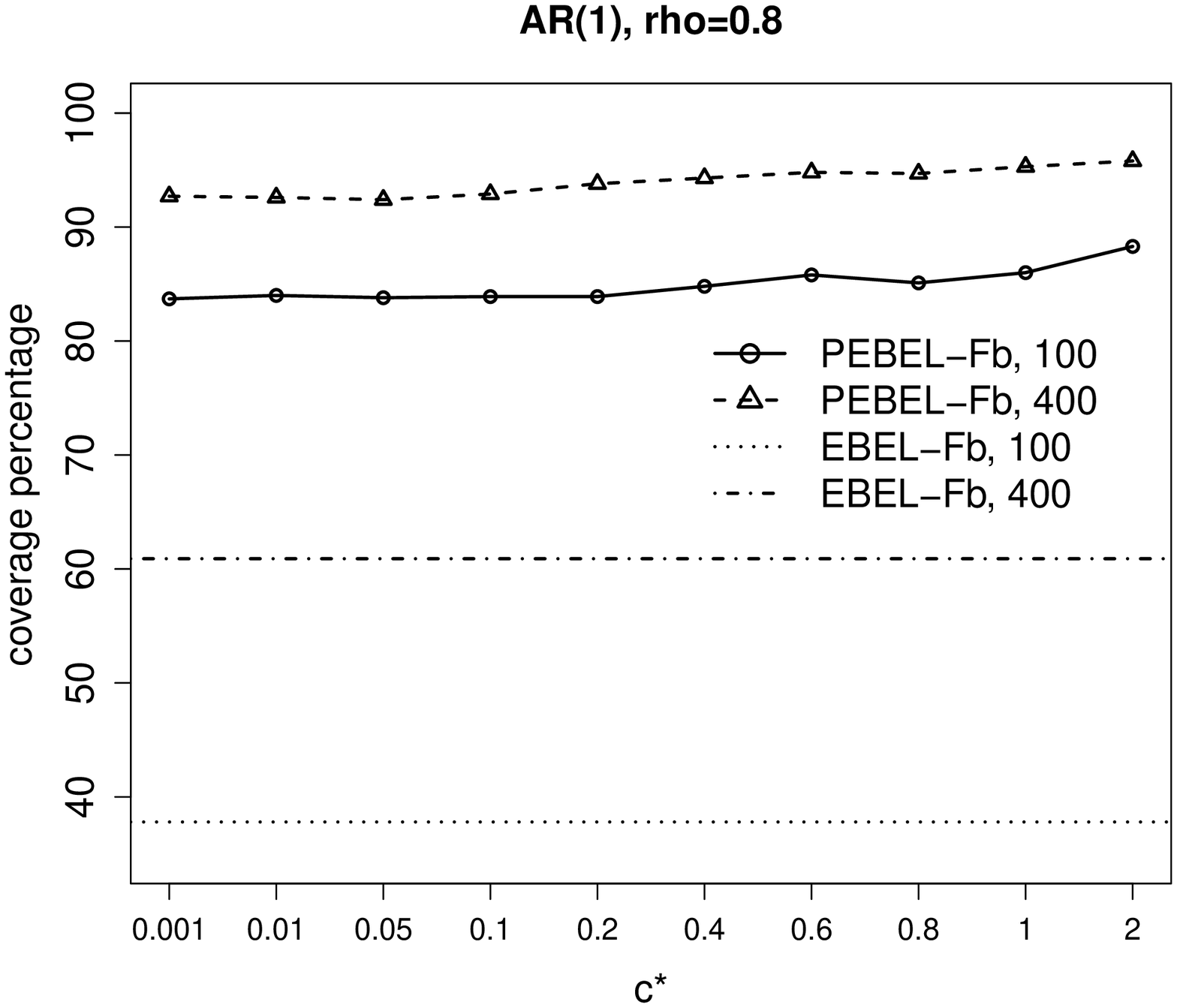}
\includegraphics[height=5.2cm,width=6cm]{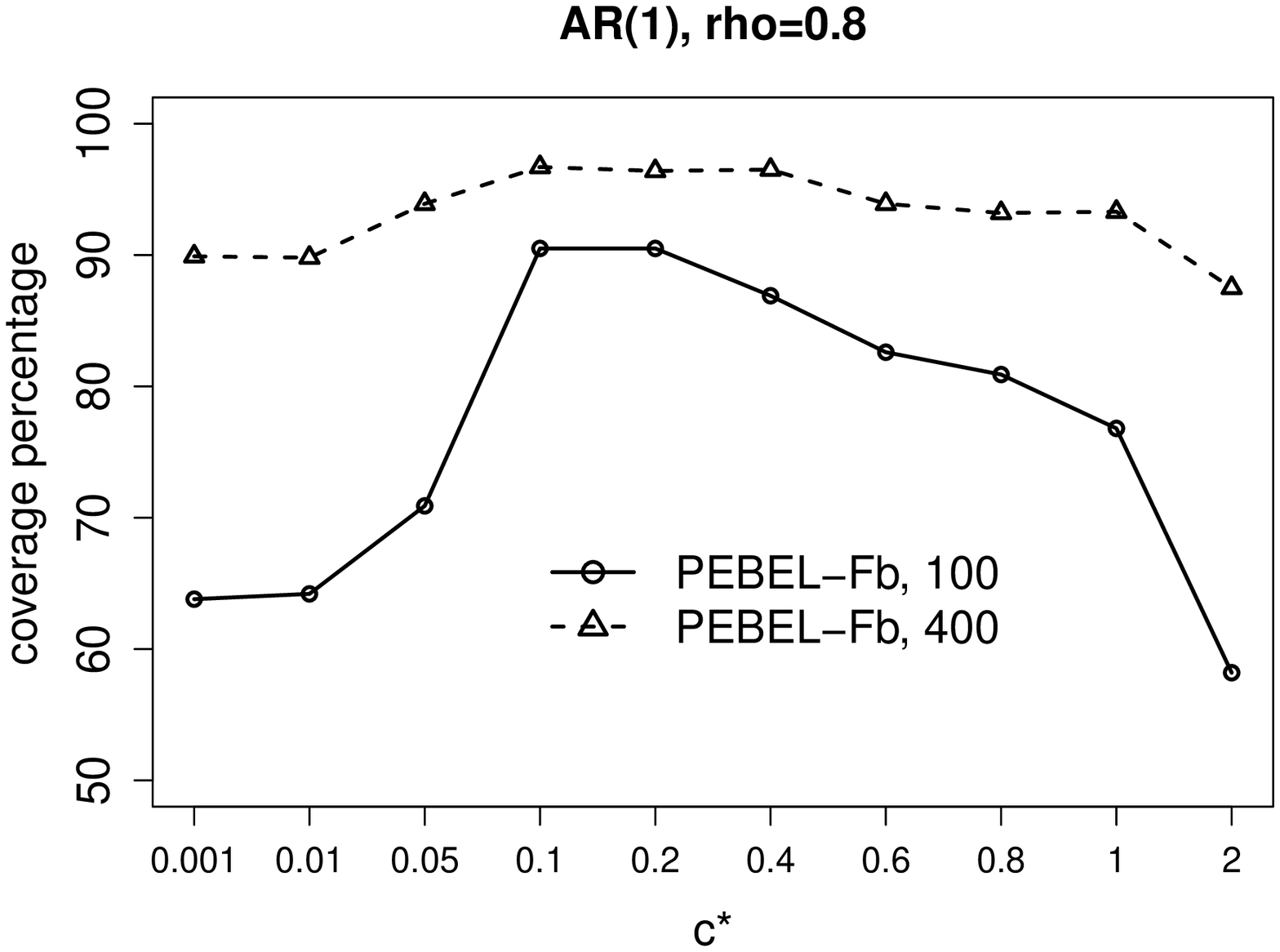}
\includegraphics[height=5.2cm,width=6cm]{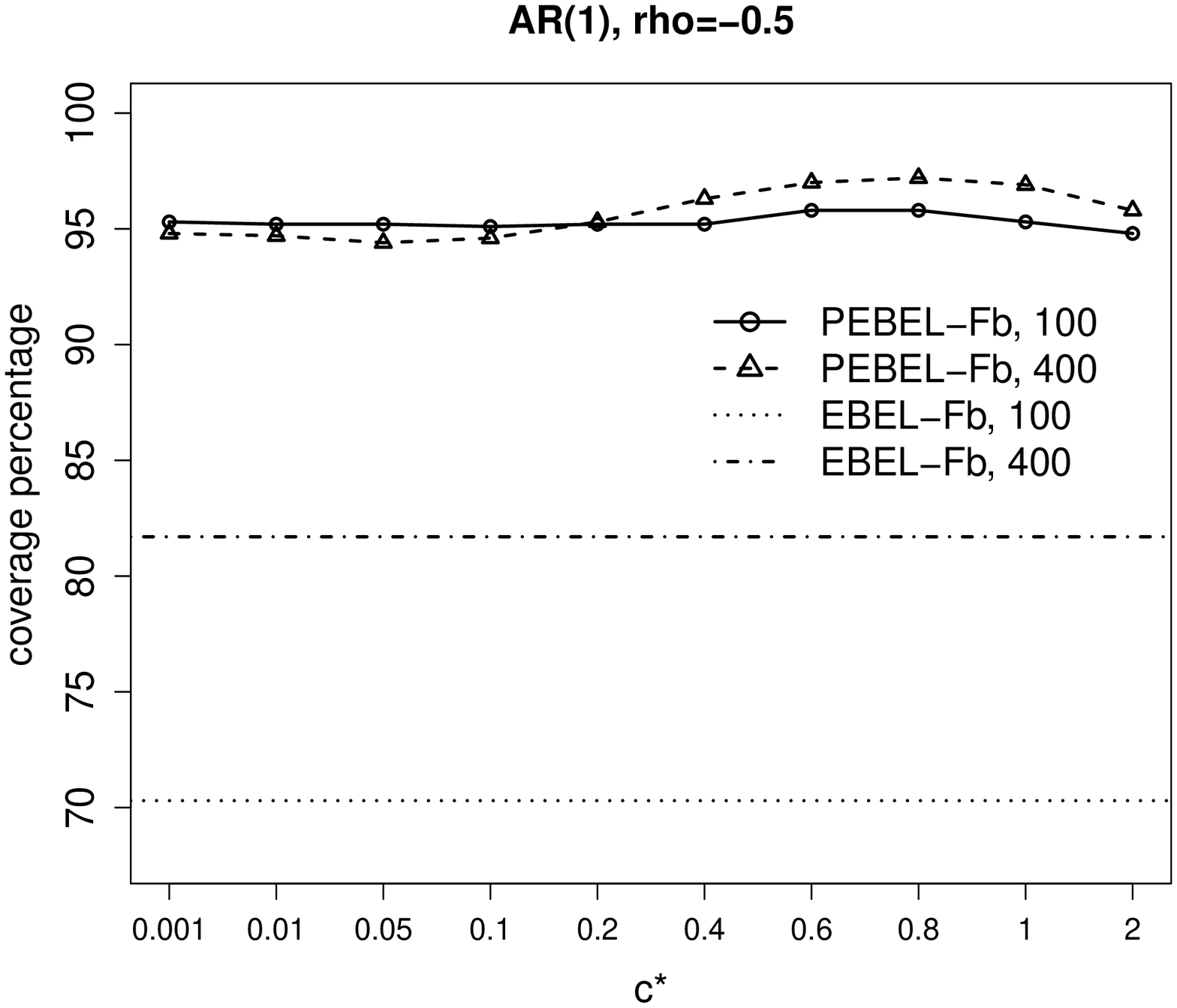}
\includegraphics[height=5.2cm,width=6cm]{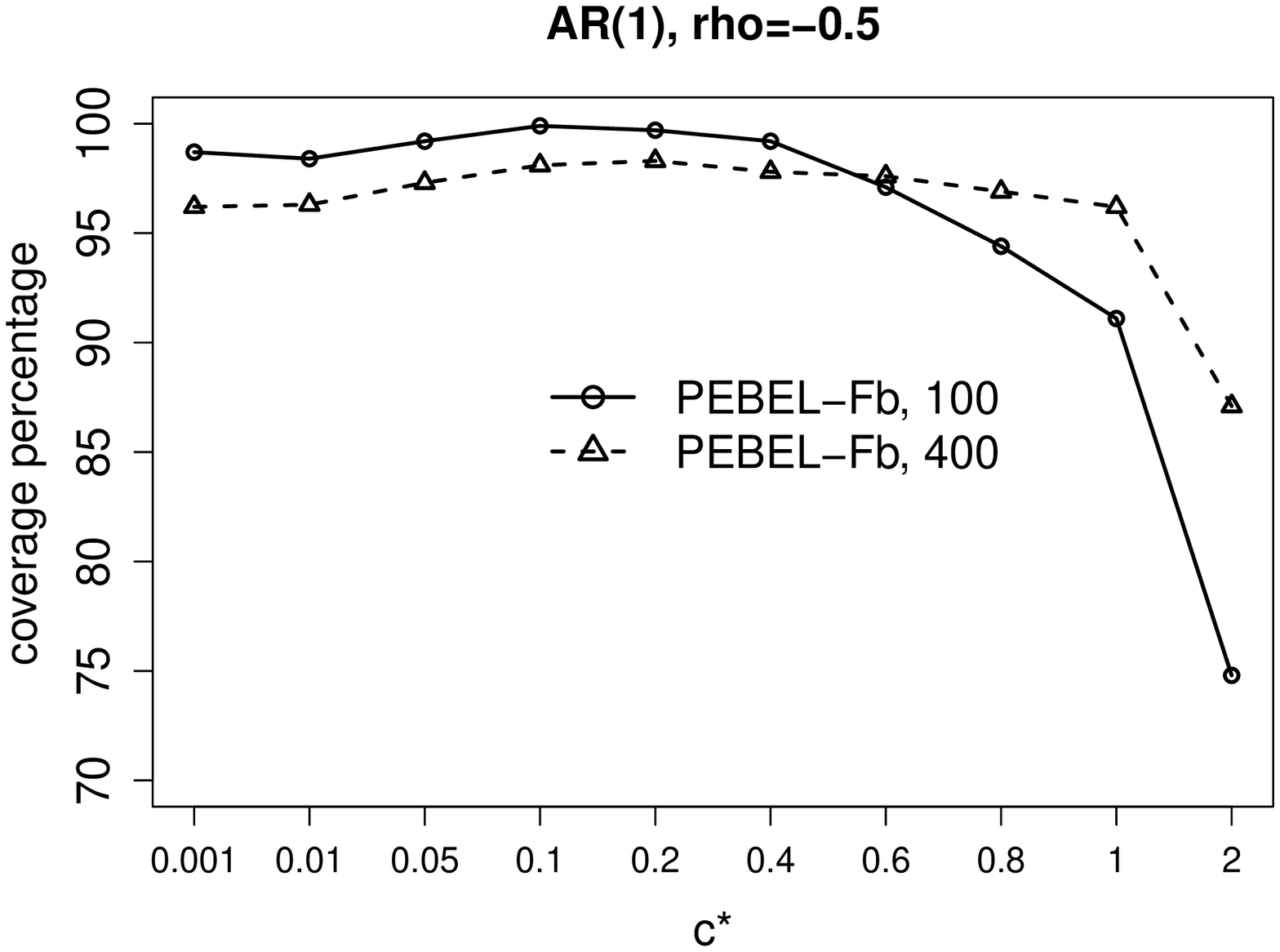}
\caption{Coverage probabilities for the mean delivered by the PEBEL
with various $c^*$ and $Q(r,s)=(1-|r-s|)\mathbf{I}\{|r-s|\leq 1\}$,
and EBEL, where $k=2$ for the left column and $k=5$ for the right
column. The nominal level is 95\% and the number of Monte Carlo
replications is 1,000.}\label{fig:pebel-ar}
\end{figure}

\newpage


\begin{figure}[H]
\centering
\includegraphics[height=10cm,width=10cm]{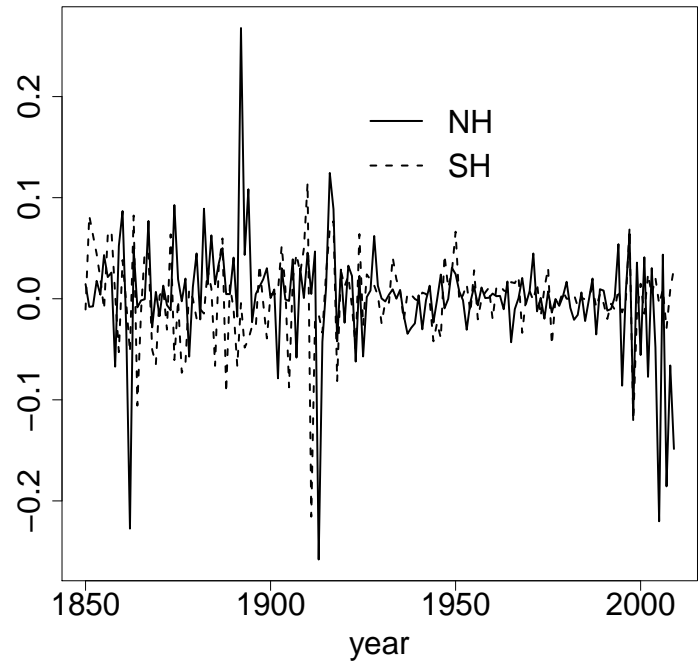}
\caption{Plot of $f(Z_t,\hat{\theta})=X_t(Y_t-X_t\hat{\theta})$ with
$Z_t=(X_t,Y_t)$. NH: northern hemisphere; SH: southern
hemisphere.}\label{fig:temp}
\end{figure}

\newpage

\vskip 1em \centerline{\LARGE \bf Supplementary material} \vskip 1em
\setcounter{section}{0}
\renewcommand{\thesection}{S\arabic{section}}
\setcounter{subsection}{0}
\renewcommand{\thesubsection}{S\arabic{subsection}}
\setcounter{equation}{0}
\renewcommand{\theequation}{S\arabic{equation}}
\setcounter{lemma}{0}
\renewcommand{\thelemma}{S\arabic{lemma}}
\setcounter{theorem}{0}
\renewcommand{\thetheorem}{S\arabic{theorem}}
\setcounter{assumption}{0}
\renewcommand{\theassumption}{S\arabic{assumption}}
\setcounter{proposition}{0}
\renewcommand{\theproposition}{S\arabic{proposition}}
\setcounter{example}{0}
\renewcommand{\theexample}{S\arabic{example}}
\setcounter{figure}{0}
\renewcommand{\thefigure}{S\arabic{figure}}
\setcounter{table}{0}
\renewcommand{\thetable}{S\arabic{table}}
\setcounter{remark}{0}
\renewcommand{\theremark}{S\arabic{remark}}

\section{Additional numerical results}
Consider the time series regression model which is given by
$$y_t=\beta_0+\beta_{1}x_{t1}+\cdots+\beta_{m_0}x_{tm_0}+\eta_t,\quad t=1,2,\dots,n,$$
where $x_t=(x_{t1},\dots,x_{tm_0})'$ is generated from VAR(1) process with
the coefficient matrix $\rho I_{m_0}$ and standard multivariate normal errors, $\{\eta_t\}$ is an AR(1) process with
coefficient $\rho$ and standard normal errors. We are interested in constructing confidence contour for the regression coefficients
$\beta=(\beta_0,\beta_1,\dots,\beta_{m_0})'$ whose true value is set to be $\beta^*=(0,0,\dots,0)'$. The moment condition is given
by $f(z_t,\beta)=\tilde{x}_t(y_t-\tilde{x}_t'\beta)$ with $\tilde{x}_t=(1,x_t')'$ and $\beta\in\mathbb{R}^{m_0+1}$. To implement the penalized methods, we consider the self-normalization matrix
\begin{equation*}
\Psi_n(\hat{\theta}_n)=\frac{1}{n}\sum^{n}_{i=1}\sum^{n}_{j=1}\left(1-\left|\frac{i-j}{n}\right|\right)f(z_i,\hat{\beta}_n)f(z_j,\hat{\beta}_n)',
\end{equation*}
with $\hat{\beta}_n$ being the least square estimate. We set $\rho=-0.5,0.2,0.5,0.8$, and $m_0=1,4.$
Figures \ref{fig:pel-ar-reg}-\ref{fig:pebel-ar-reg} present respectively the coverage probabilities for PBEL and PEBEL, and their unpenalized counterparts at the 95\% nominal level.
Note that for $m_0=4$, we only present the results for the penalized methods as the unpenalized counterparts severely suffer from the coverage upper bound problem.
Overall, the results are qualitatively similar to those for the mean case. When $m_0=1$ (i.e. $k=2$), the PBEL provides better coverage uniformly over $b$ as compared
to the BEL when $\rho=0.5$ and 0.8. The improvement becomes more significant as the block size grows. When $m_0=4$ (i.e. $k=5$), we note that the choice of $c^*$ that delivers the most accurate
coverage is delicate as it depends on $b$ and the strength of dependence. For PEBEL, the improvement on the coverage probabilities is again significant for
the range of $c^*$ being considered. When $m_0=4$, the coverage probability is sensitive to the choice of $c^*$ and we expect the block bootstrap method described in Section \ref{sub:sim-pebel}
to be useful in this case.

\section{Fixed-$b$ asymptotics versus small-$b$ asymptotics}
For the coverage upper bound problem, the fixed-$b$
method  is more appropriate than the small-$b$ method in terms of describing the finite sample
situation because one cause of the coverage upper bound problem in
the dependent case is the choice of the blocking strategy and block
size which is explicitly reflected in the fixed-$b$ limit. For instance, when the finite sample bound problem occurs,
the fixed-$b$ method correctly
reflects such a phenomenon in the asymptotics, while the original
BEL under the small-$b$ asymptotics is somewhat ``over-optimistic'' as the corresponding upper bound in the limit
is always one regardless of what the finite sample bound is.


The basic philosophy behind the fixed-$b$ method is to embed the
finite sample situation in a different limiting thought experiment,
where $b$ (the fraction of block size relative to sample size) is
held fixed as the sample size grows. In the small-$b$ asymptotics,
$b$ goes to zero as sample size $n\rightarrow +\infty$, which is a
convenient assumption for deriving a limiting distribution. However,
as pointed out in Kiefer and Vogelsang (2005) [also see Neave
(1970)], this assumption can be unrealistic when the deduced results
are used as approximations to the finite sample case where the value of $b$
can never be zero. Thus, the finite sample coverage upper bound
problem in general does not go away since $b$ cannot be zero.
The usefulness of the fixed-$b$ approach has been demonstrated in
many other contexts; see Sun (2013); Kiefer and Vogelsang (2005); Shao
(2010), Shao and Politis (2013), among others. Using the higher-order Edgeworth expansions,
Jansson (2004), Sun et al. (2008), Sun (2014) and Zhang and Shao
(2013) rigorously proved that the fixed-$b$ type asymptotics
provides a high order refinement over the traditional small-$b$ type
asymptotics in the Gaussian location model. It is also worth
pointing out that the fixed-$b$ method and small-$b$ method are
consistent for relatively small $b$ in the BEL context [see e.g. Remark 2 of Zhang and
Shao (2014)]. Given the connection between BEL and generalized method of moments
framework, for which the fixed-$b$ approach has been successfully
extended to [see Kiefer and Vogelsang (2005), Vogelsang (2003)], we
believe that the nice properties of fixed-$b$ approach as found in other
contexts carry over to the BEL case, which has been partially confirmed in Zhang and Shao (2014).

The fixed-$b$ asymptotics not only provides better approximation for
the original BEL but it also tends to provide better approximation
for the penalized counterpart. To further illustrate the superiority
of the fixed-$b$ approach over the small-$b$ approach in the PBEL
context, we shall present some simulation results. Following the
setup in Section \ref{sec:numerical}, we focus on the confidence
region for the mean of univariate/multivariate time series. In the
univariate case, we consider the AR(1) process $z_t=\rho
z_{t-1}+\varepsilon_t$ with $\rho=-0.5, 0.2, 0.5, 0.8$, where
$\{\varepsilon_t\}$ is a sequence of i.i.d standard normal errors.
In the multidimensional case (i.e. $k>1$), we generate multivariate
time series with each component being independent AR(1) process. The
sample size $n=100$ and the nominal level is 95\%. To construct
confidence interval for the mean of the time series, we consider the
BEL and PBEL under both fixed-$b$ and small-$b$ asymptotics. Recall
the definition of the PBEL ratio test statistic which is given by,
\begin{equation*}
elr_{pbel}(\theta)=\min_{\mu\in\mathbb{R}^k}\left\{\frac{2}{nb}\max_{\lambda\in\mathbb{R}^k}\sum^{N}_{t=1}\log
\left\{1+\lambda'(f_{tn}(\theta)-\mu)\right\}+\frac{\tau}{b}\delta_n\left(\mu\right)\right\}.
\end{equation*}
For PBEL under the small-$b$ asymptotics, we shall use the critical
value from $\chi^2$ distribution to conduct inference. The $\chi^2$
approximation is only valid when $b$ is small and $c^*$ is large
because for large $c^*$, the PBEL ratio statistic behaves like the
BEL ratio statistic which has a $\chi^2$ limit when $b$ is small.
Note that when $c^*$ is small, the penalty term dominates and the
limiting distribution is no longer $\chi^2$. Thus for relatively
small $c^*$ (e.g. $c^*=0.01$), we do not present the results for
PBEL under the small-$b$ asymptotics. Figure \ref{fig:pbel-add}
plots the coverage probabilities for the four methods: small-$b$
based BEL, fixed-$b$ based BEL, small-$b$ based PBEL and fixed-$b$
based PBEL. Again, the fixed-$b$ based methods significantly outperform the
small-$b$ counterparts for relatively large $b$. We also observe
that the penalized methods in general provide improvement over the
unpenalized counterparts. The improvement is quite significant
especially when $c^*$ is small and $b$ is relatively large. Overall,
this small simulation illustrates the advantage of the fixed-$b$
method and also demonstrates the usefulness of the penalized method.

\newpage

\begin{figure}[H]
\centering
\includegraphics[height=5.2cm,width=6cm]{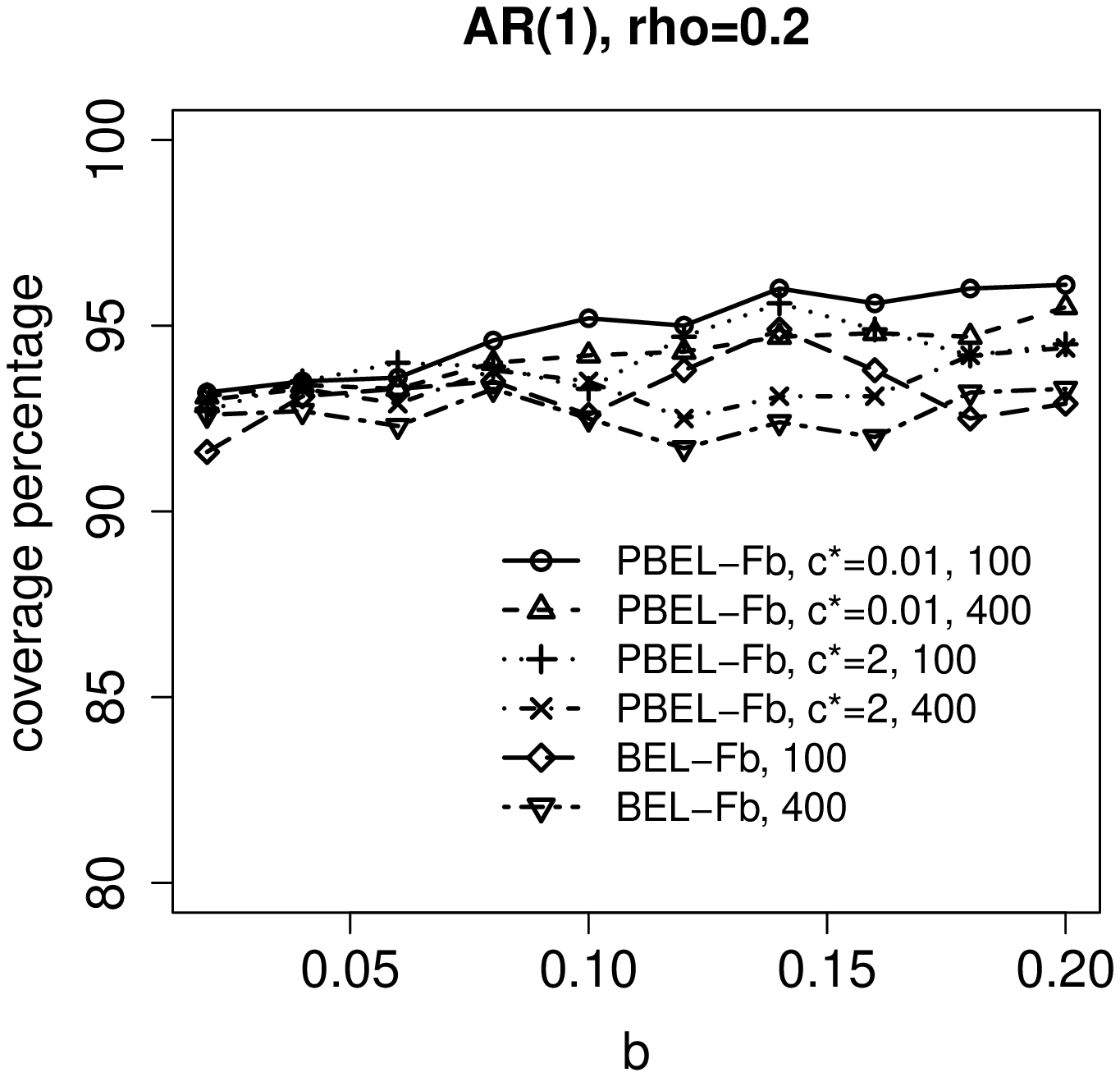}
\includegraphics[height=5.2cm,width=6cm]{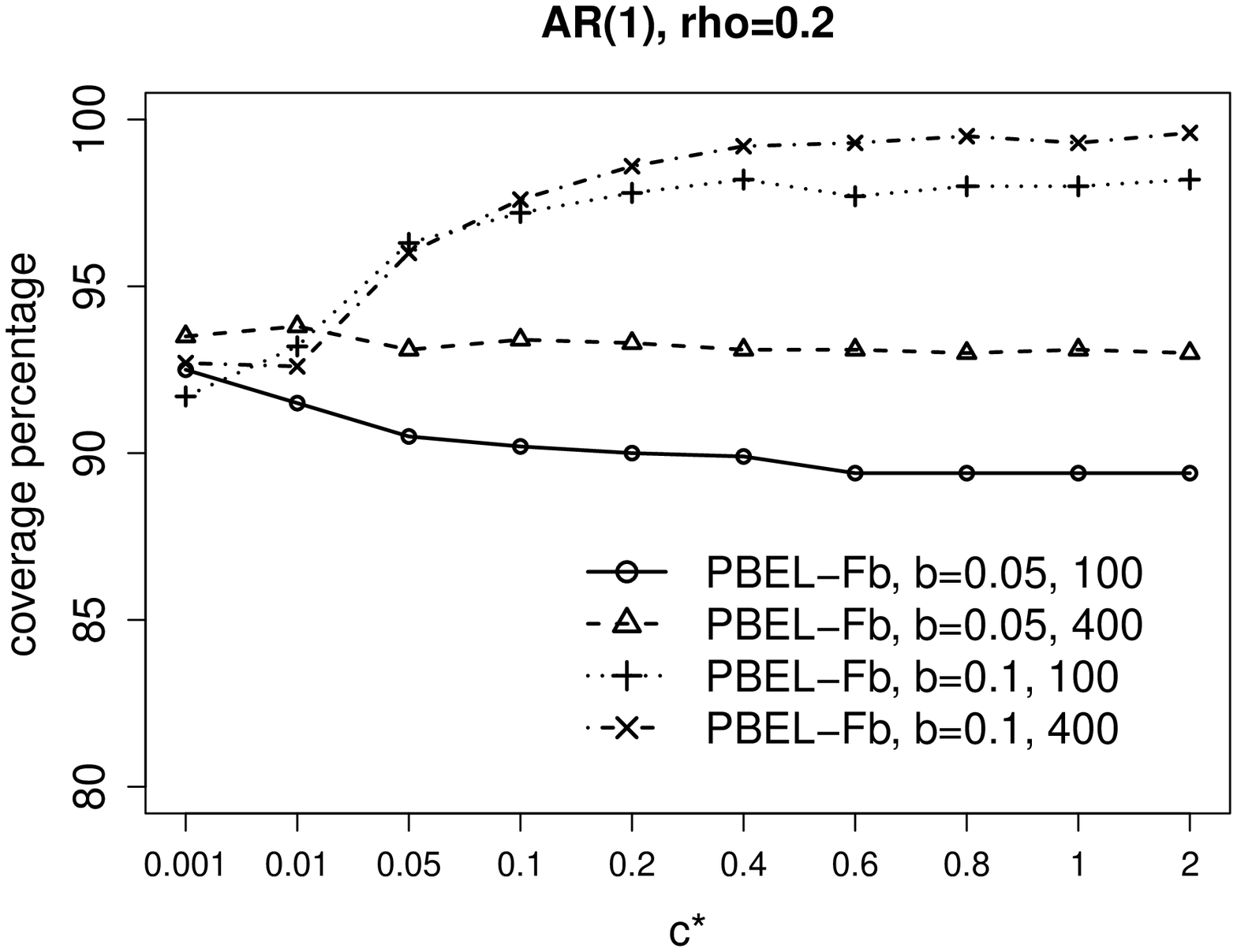}
\includegraphics[height=5.2cm,width=6cm]{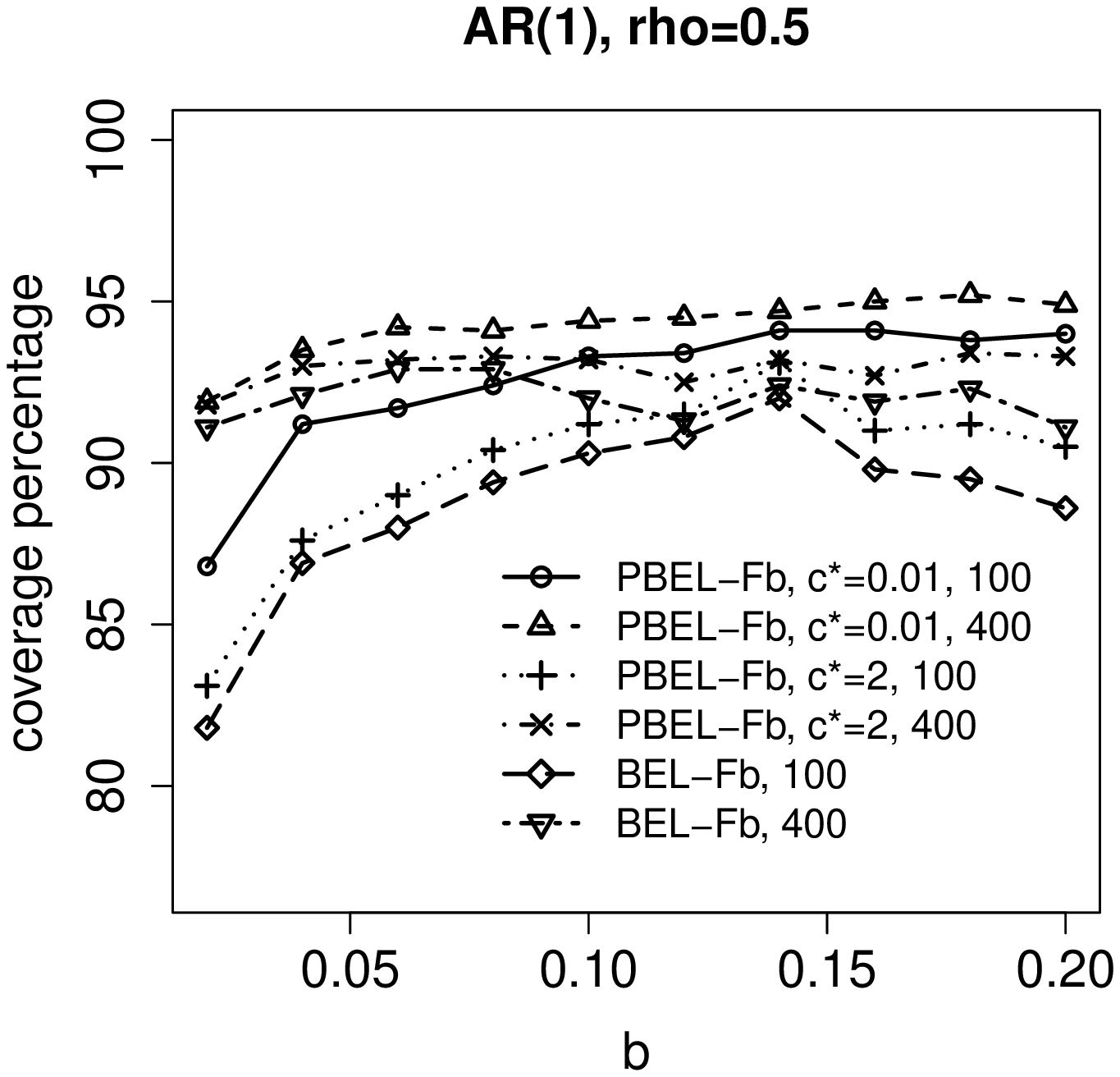}
\includegraphics[height=5.2cm,width=6cm]{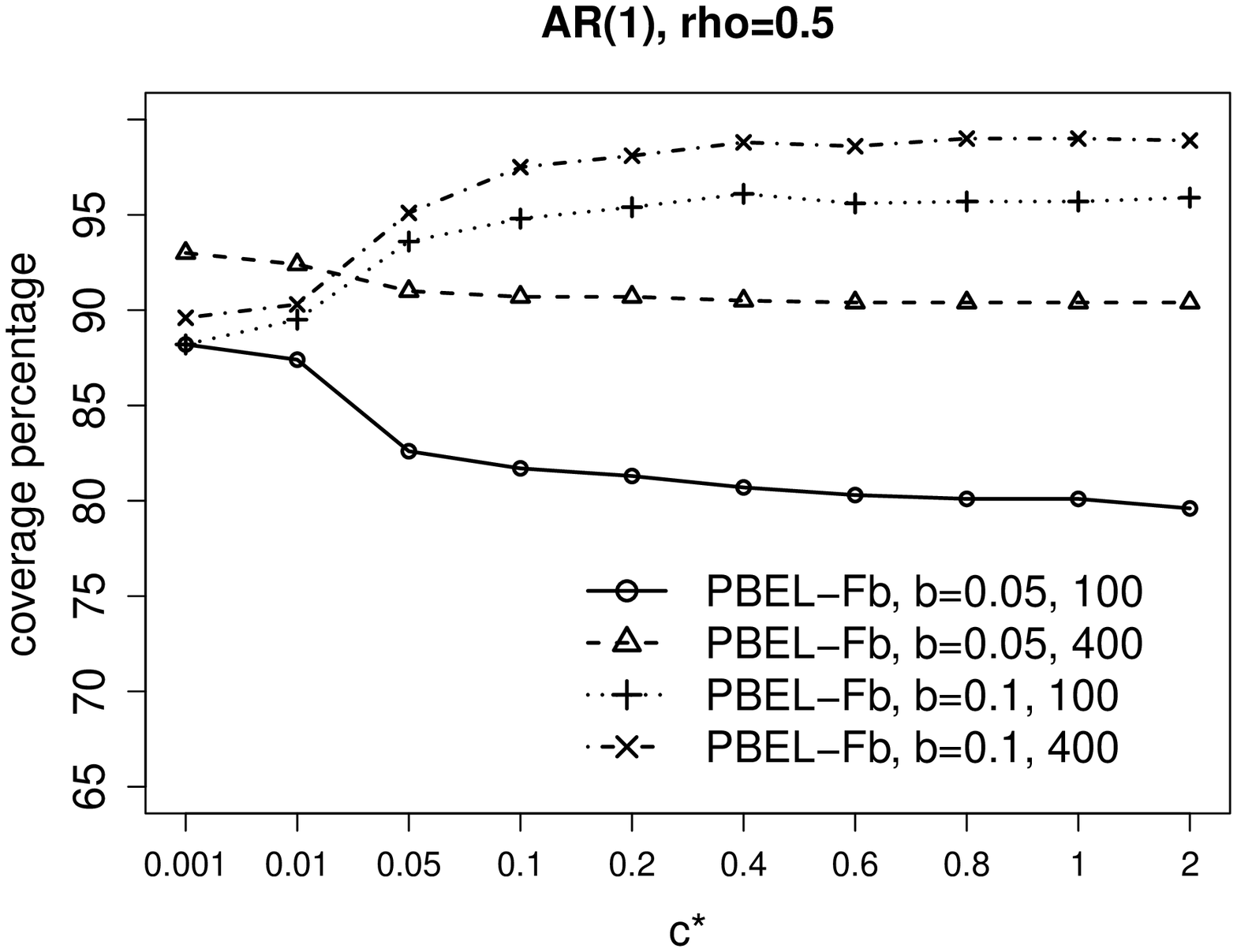}
\includegraphics[height=5.2cm,width=6cm]{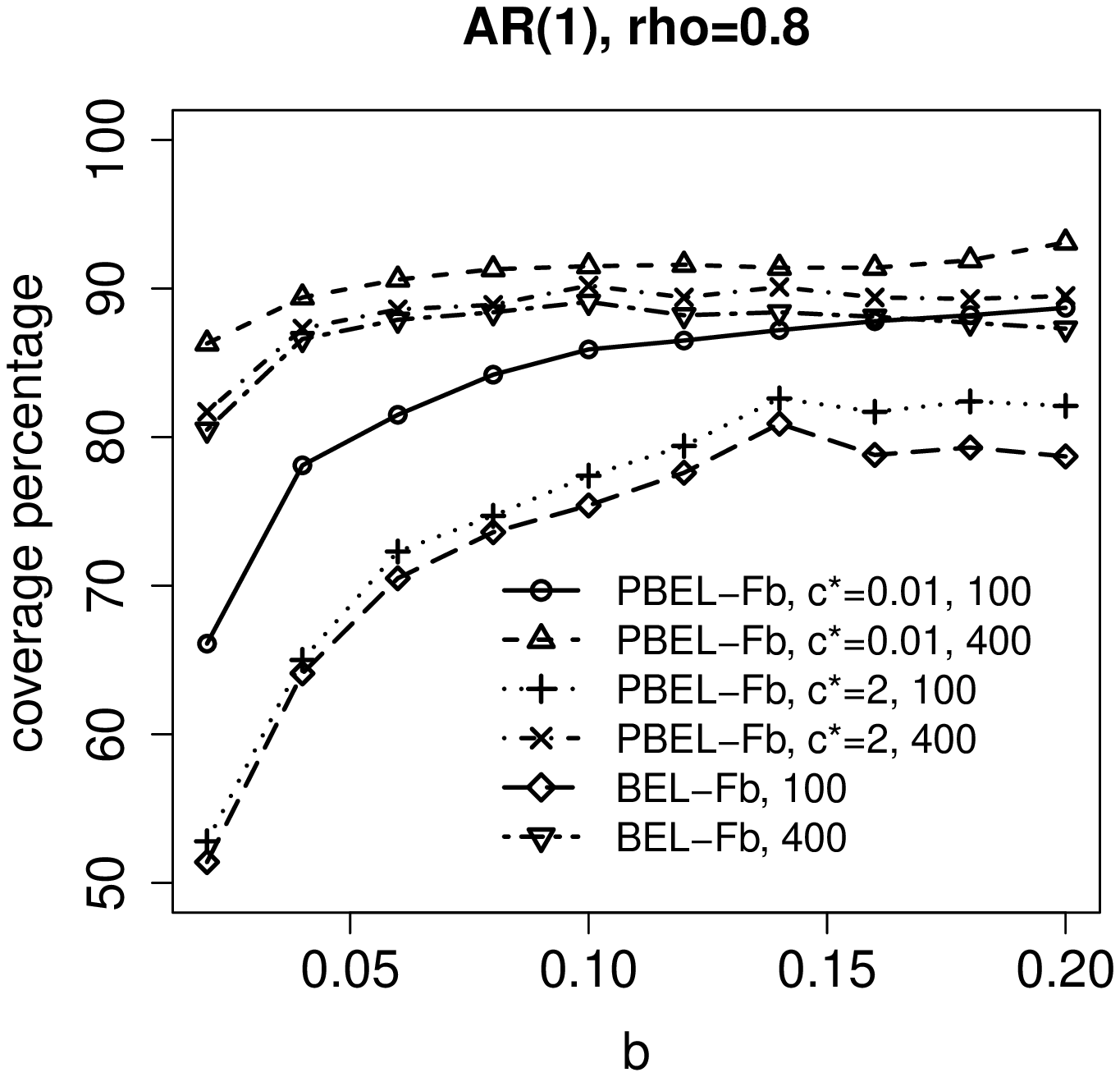}
\includegraphics[height=5.2cm,width=6cm]{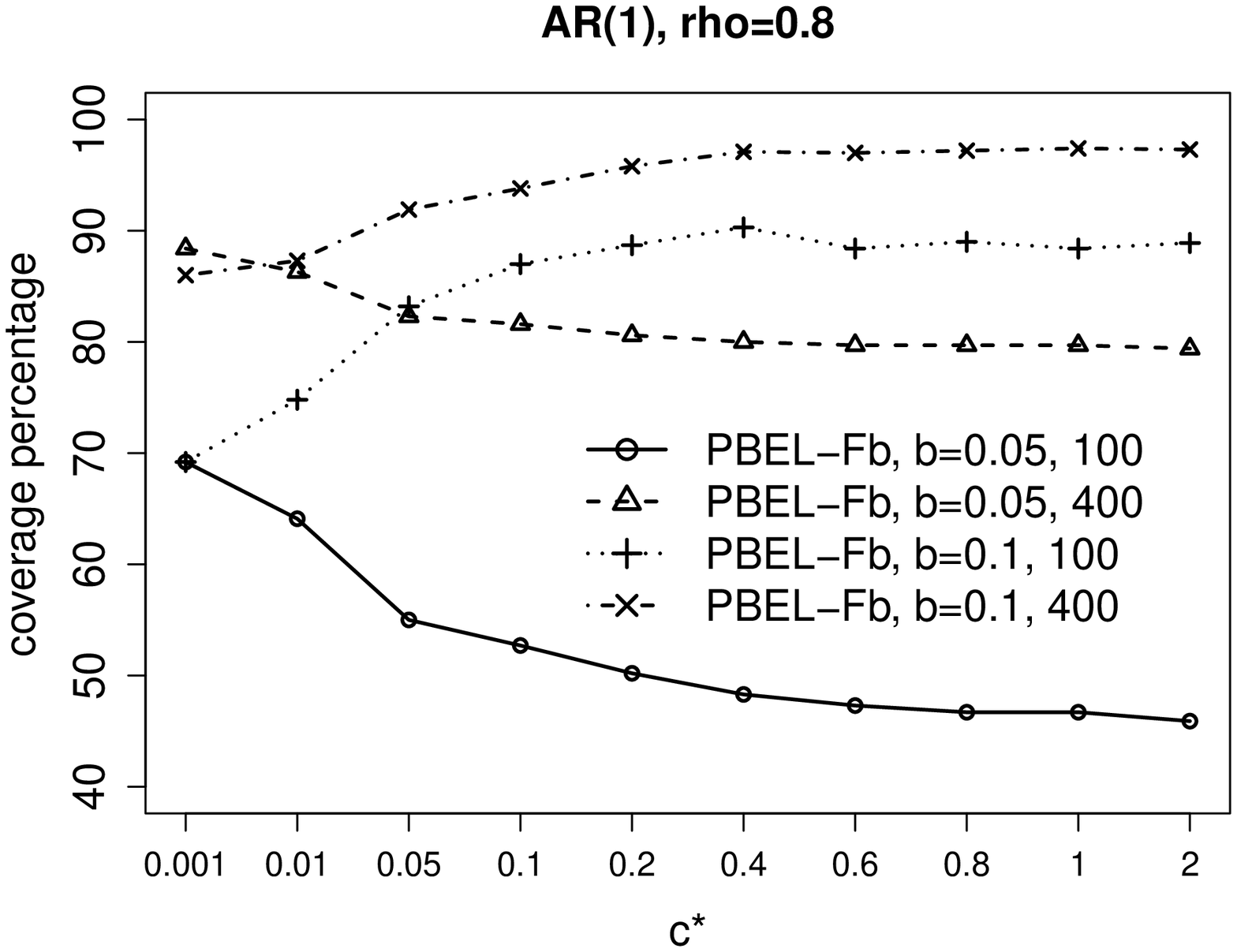}
\includegraphics[height=5.2cm,width=6cm]{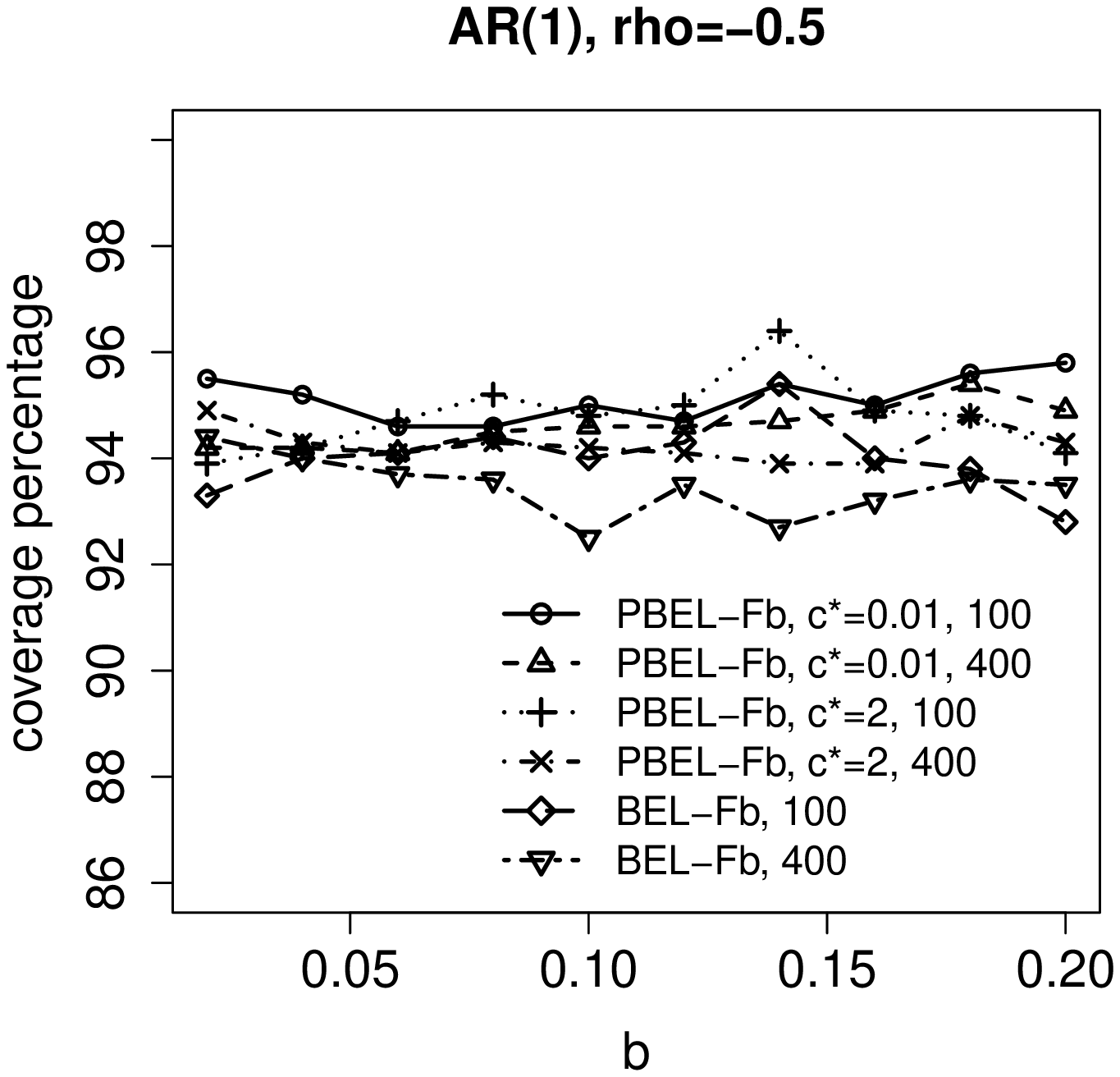}
\includegraphics[height=5.2cm,width=6cm]{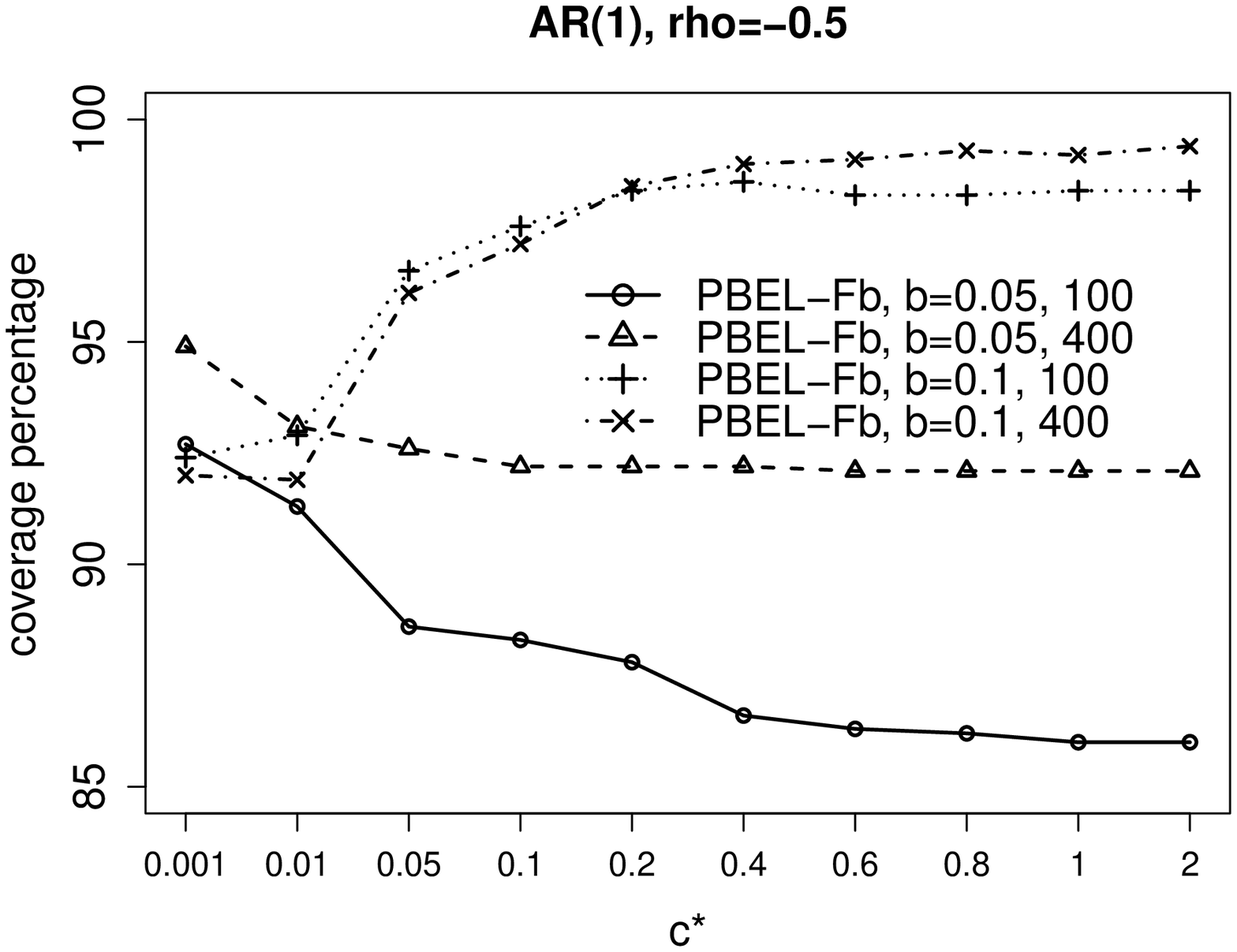}
\caption{Coverage probabilities for the regression coefficients delivered by the PBEL
with $Q(r,s)=(1-|r-s|)\mathbf{I}\{|r-s|\leq 1\}$,
and BEL, where $m_0=1$ for the left column and $m_0=4$ for the right
column. The nominal level is 95\% and the number of Monte Carlo
replications is 1,000.}\label{fig:pel-ar-reg}
\end{figure}

\newpage

\begin{figure}[H]
\centering
\includegraphics[height=5.2cm,width=6cm]{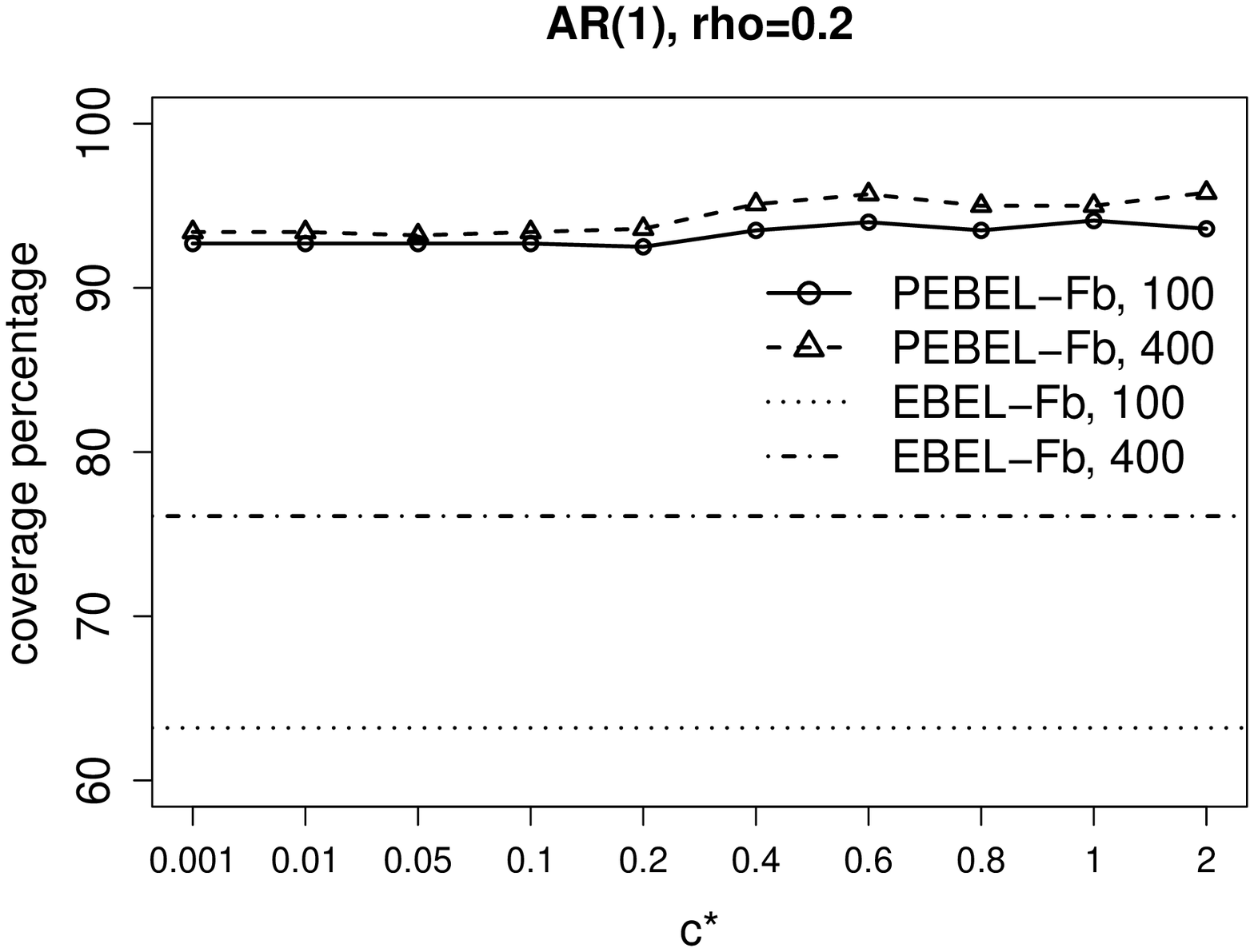}
\includegraphics[height=5.2cm,width=6cm]{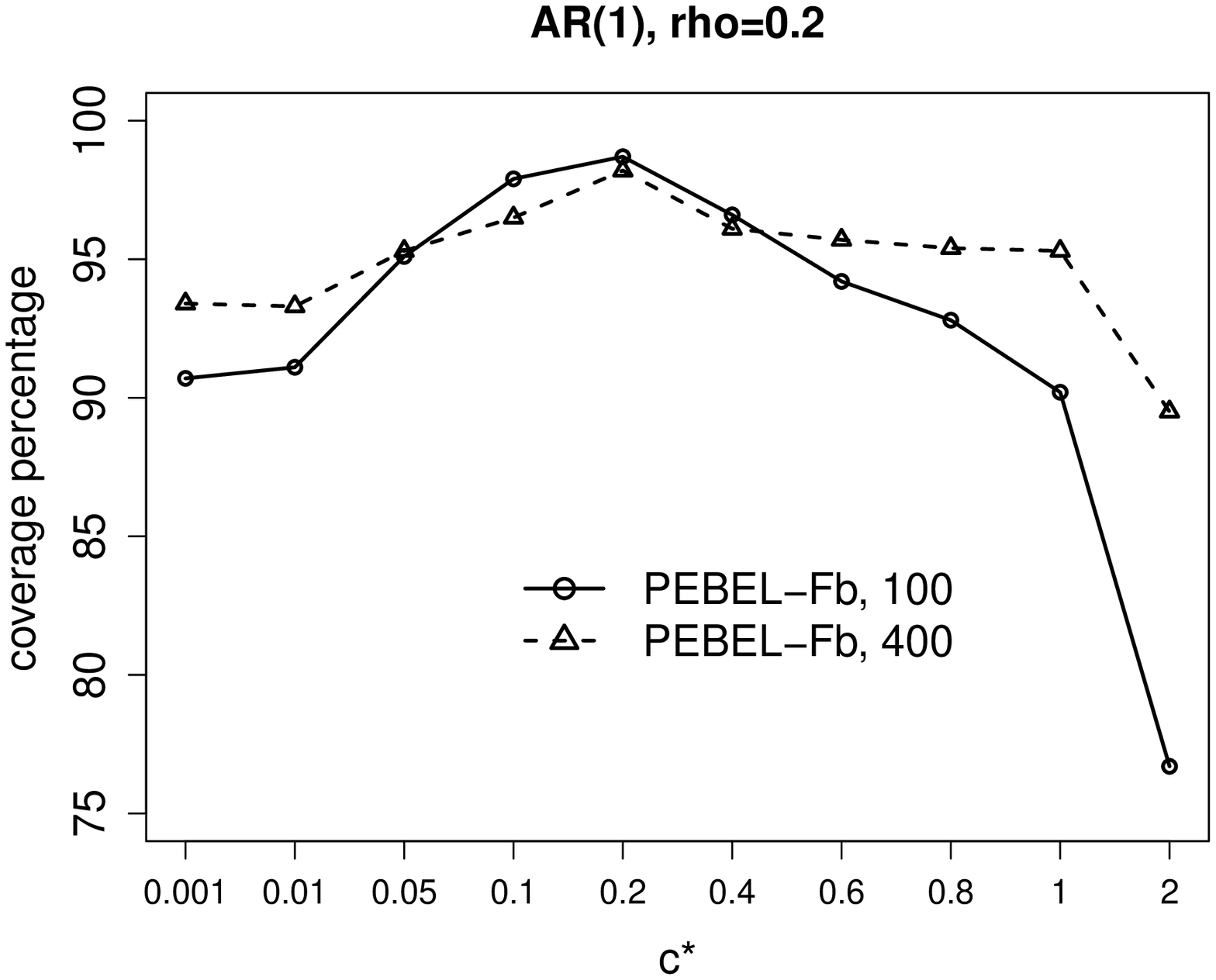}
\includegraphics[height=5.2cm,width=6cm]{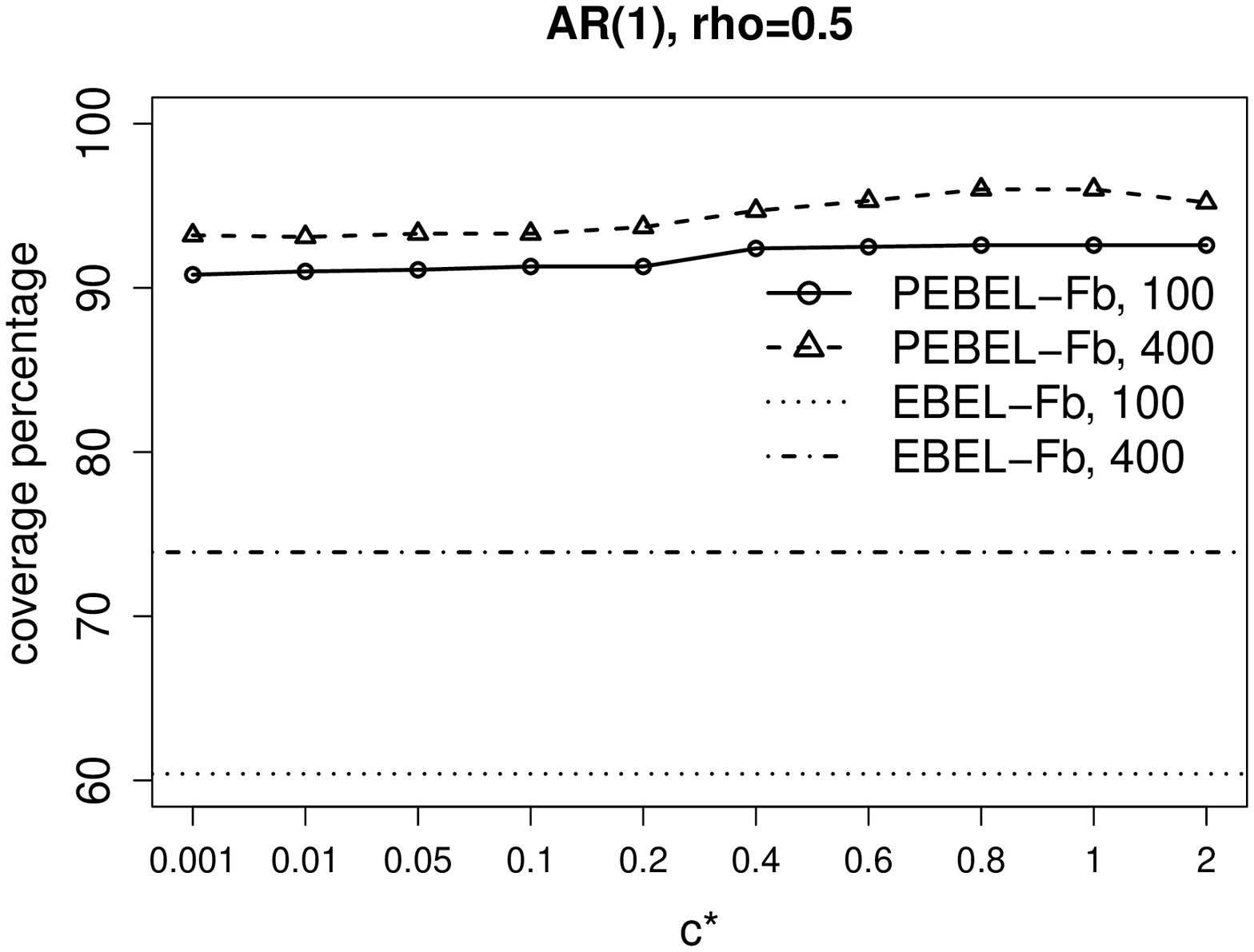}
\includegraphics[height=5.2cm,width=6cm]{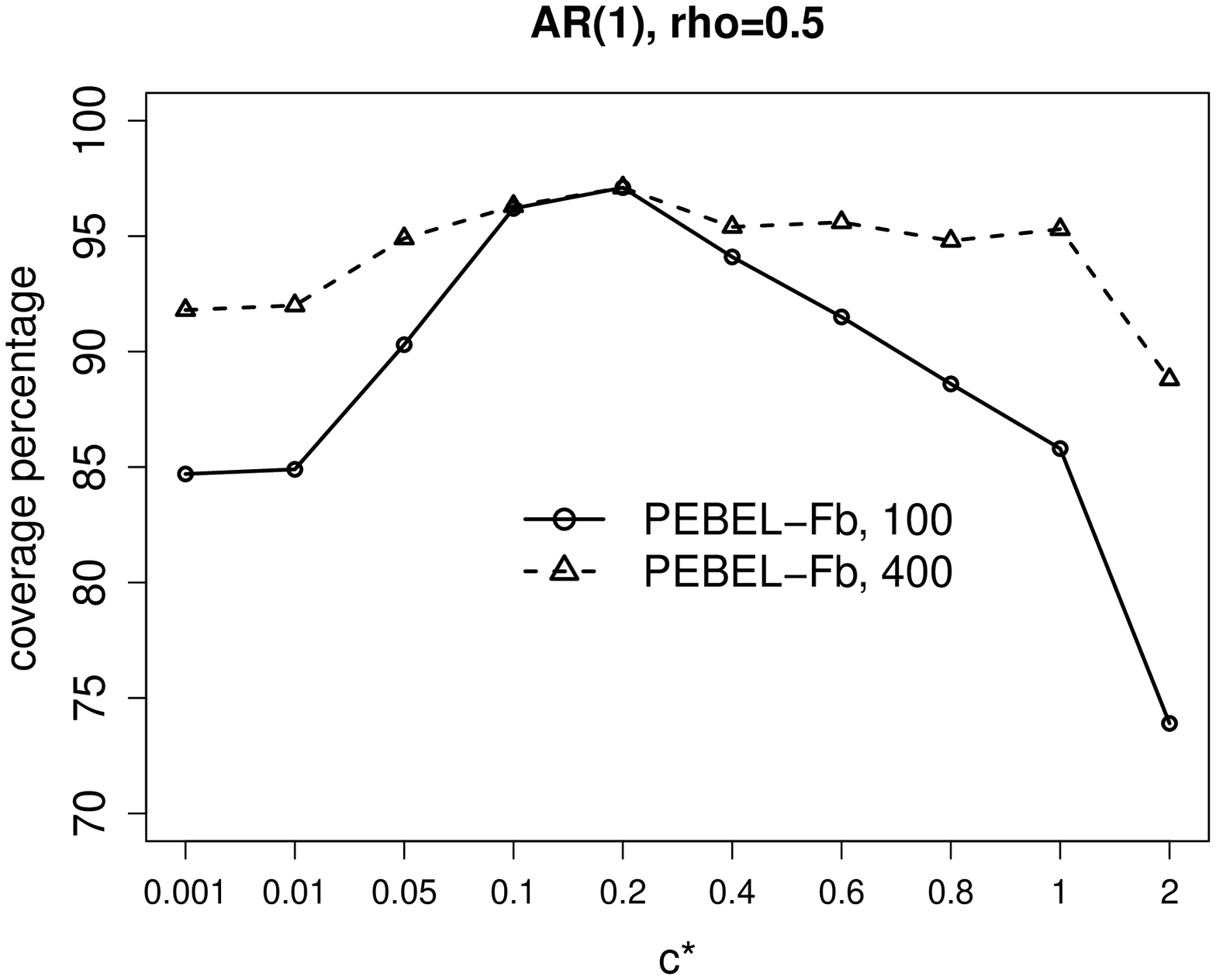}
\includegraphics[height=5.2cm,width=6cm]{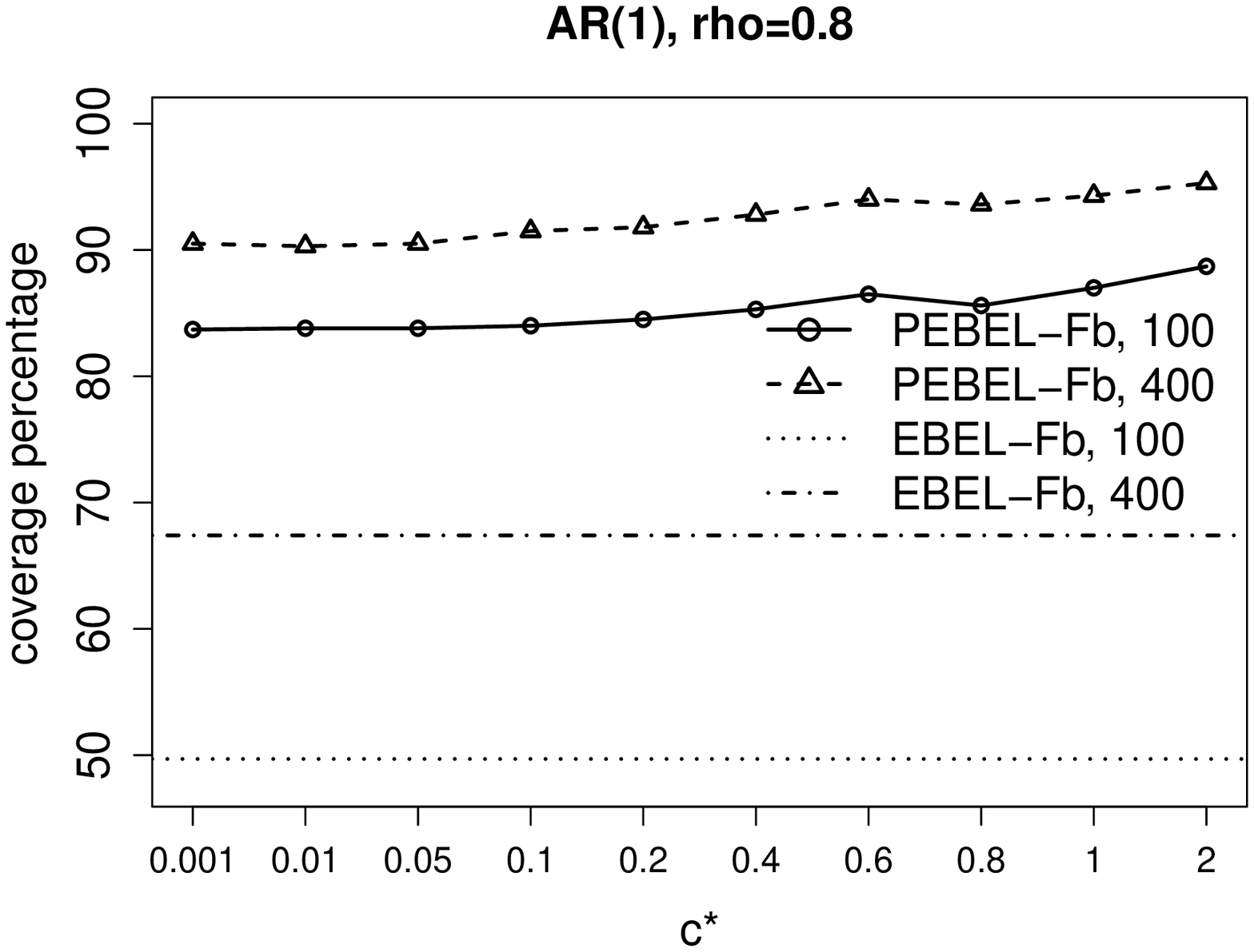}
\includegraphics[height=5.2cm,width=6cm]{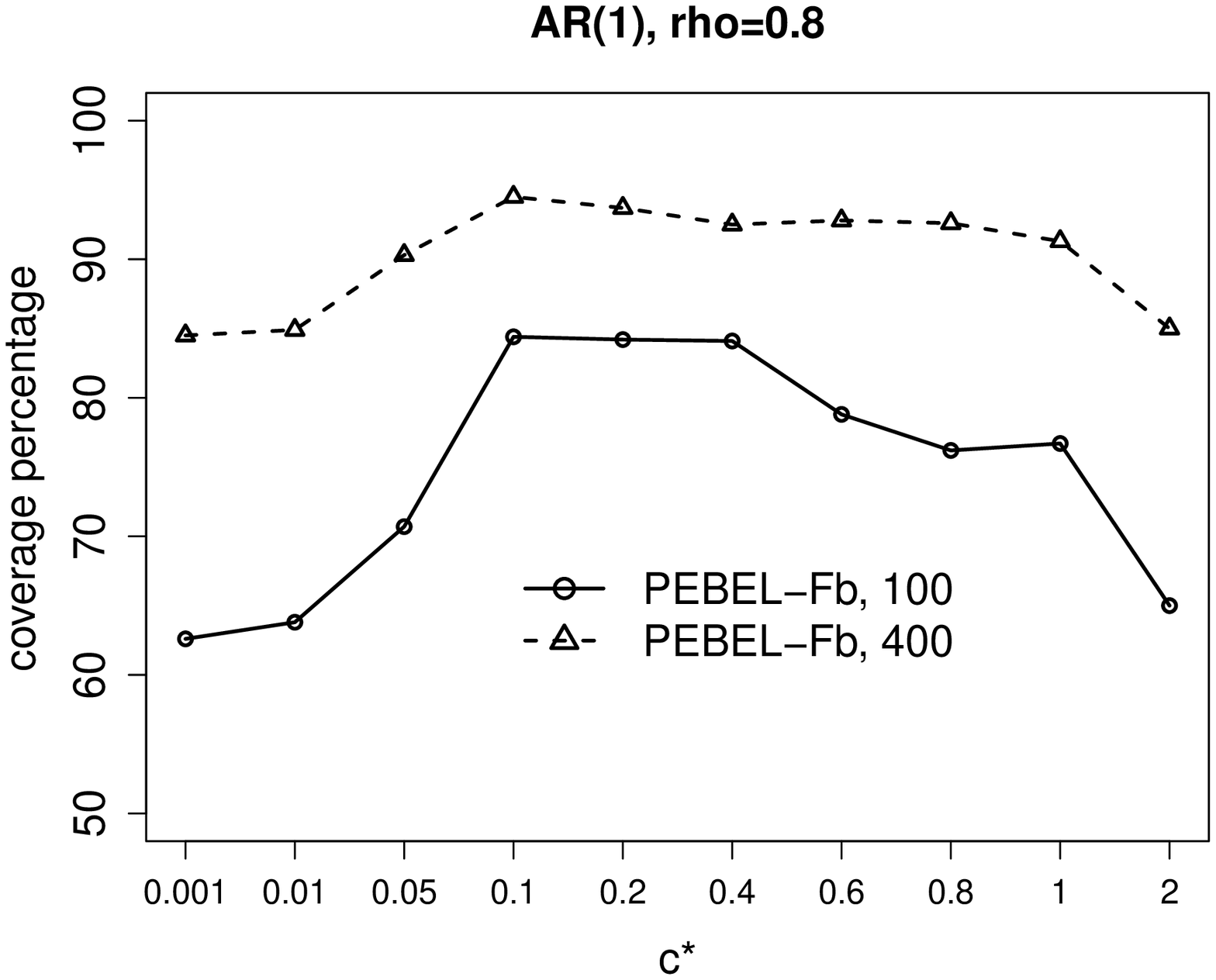}
\includegraphics[height=5.2cm,width=6cm]{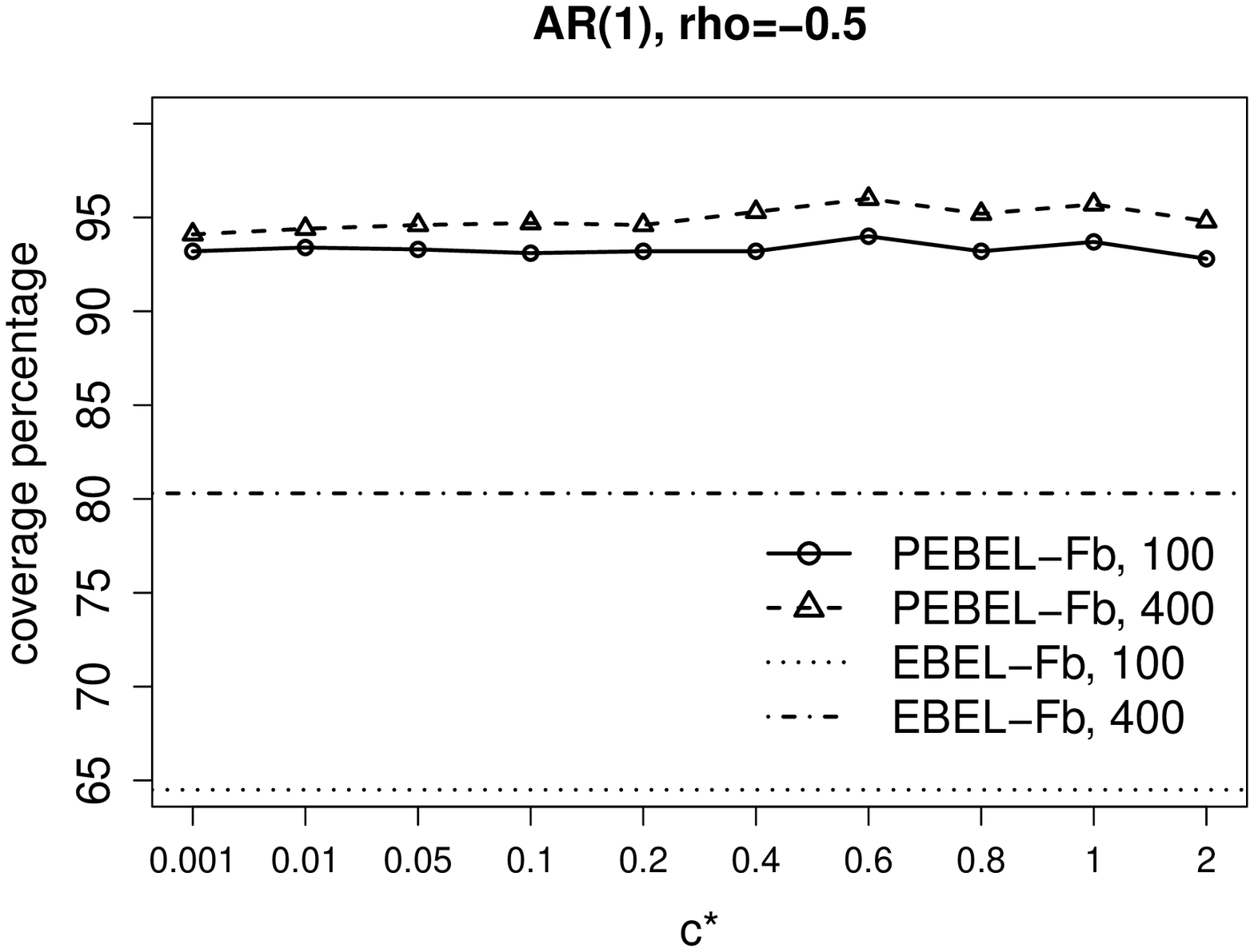}
\includegraphics[height=5.2cm,width=6cm]{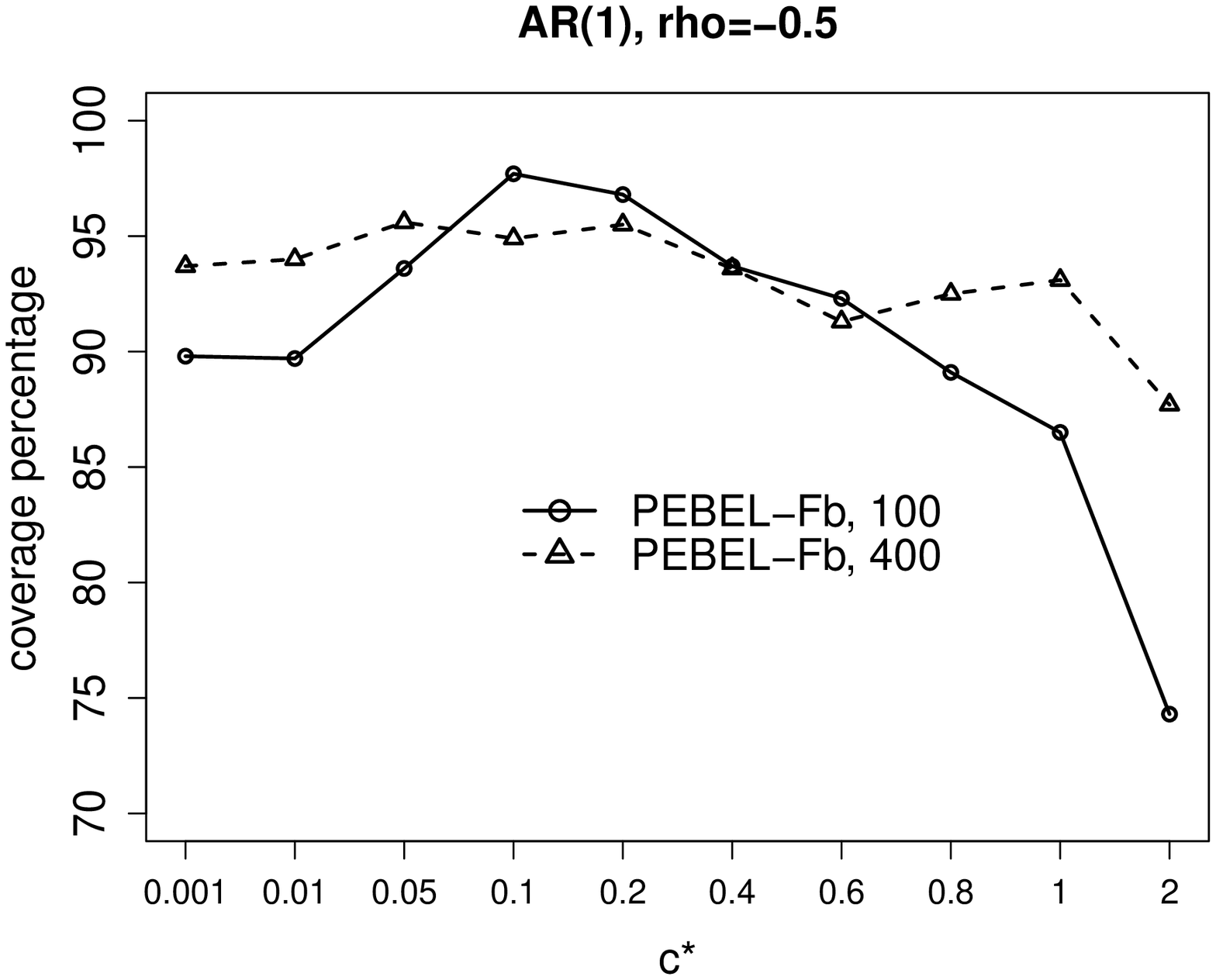}
\caption{Coverage probabilities for the regression coefficients delivered by the PEBEL
with various $c^*$ and $Q(r,s)=(1-|r-s|)\mathbf{I}\{|r-s|\leq 1\}$,
and EBEL, where $m_0=1$ for the left column and $m_0=4$ for the right
column. The nominal level is 95\% and the number of Monte Carlo
replications is 1,000.}\label{fig:pebel-ar-reg}
\end{figure}

\newpage

\begin{figure}[H]
\centering
\includegraphics[height=5.2cm,width=6cm]{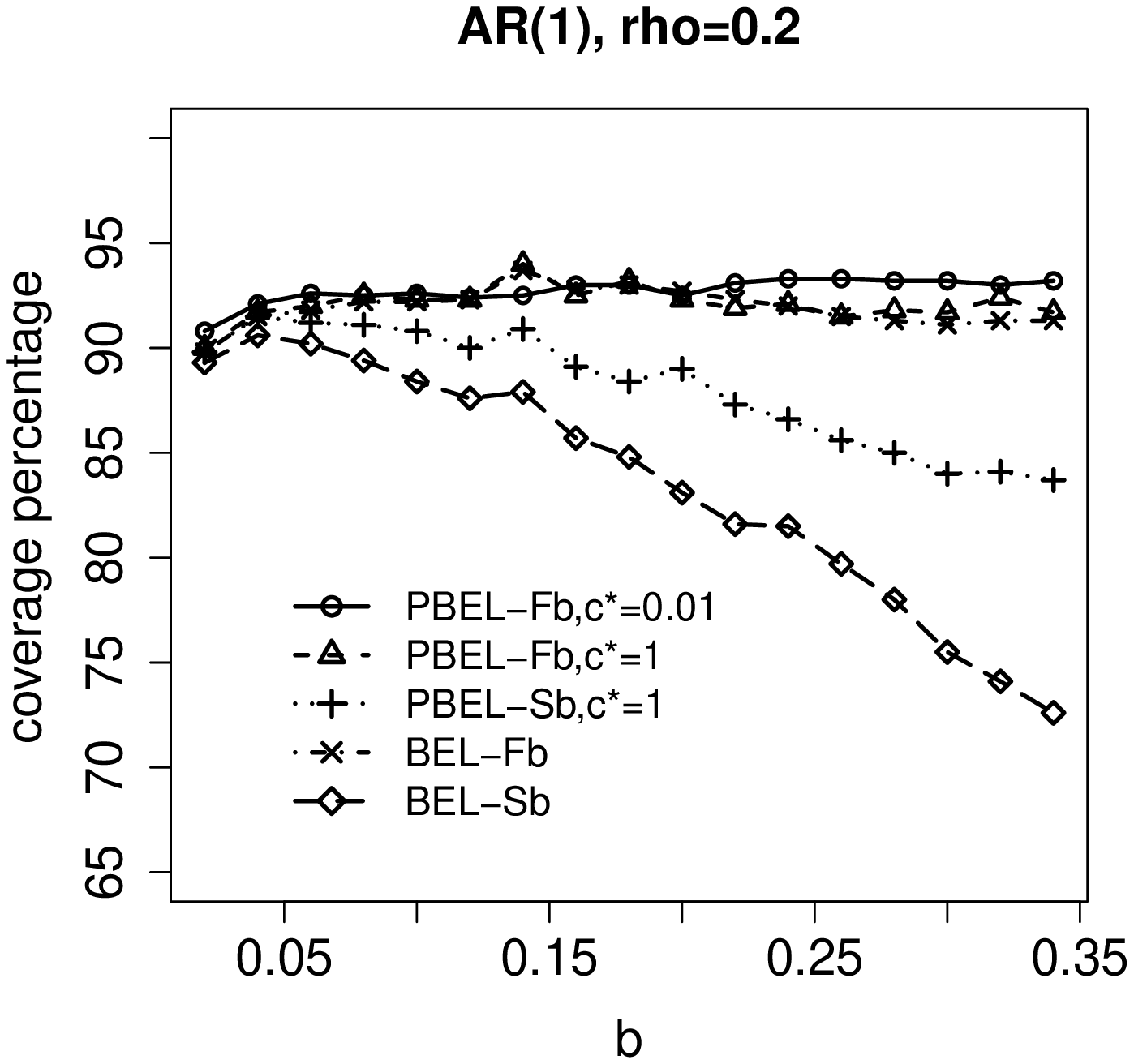}
\includegraphics[height=5.2cm,width=6cm]{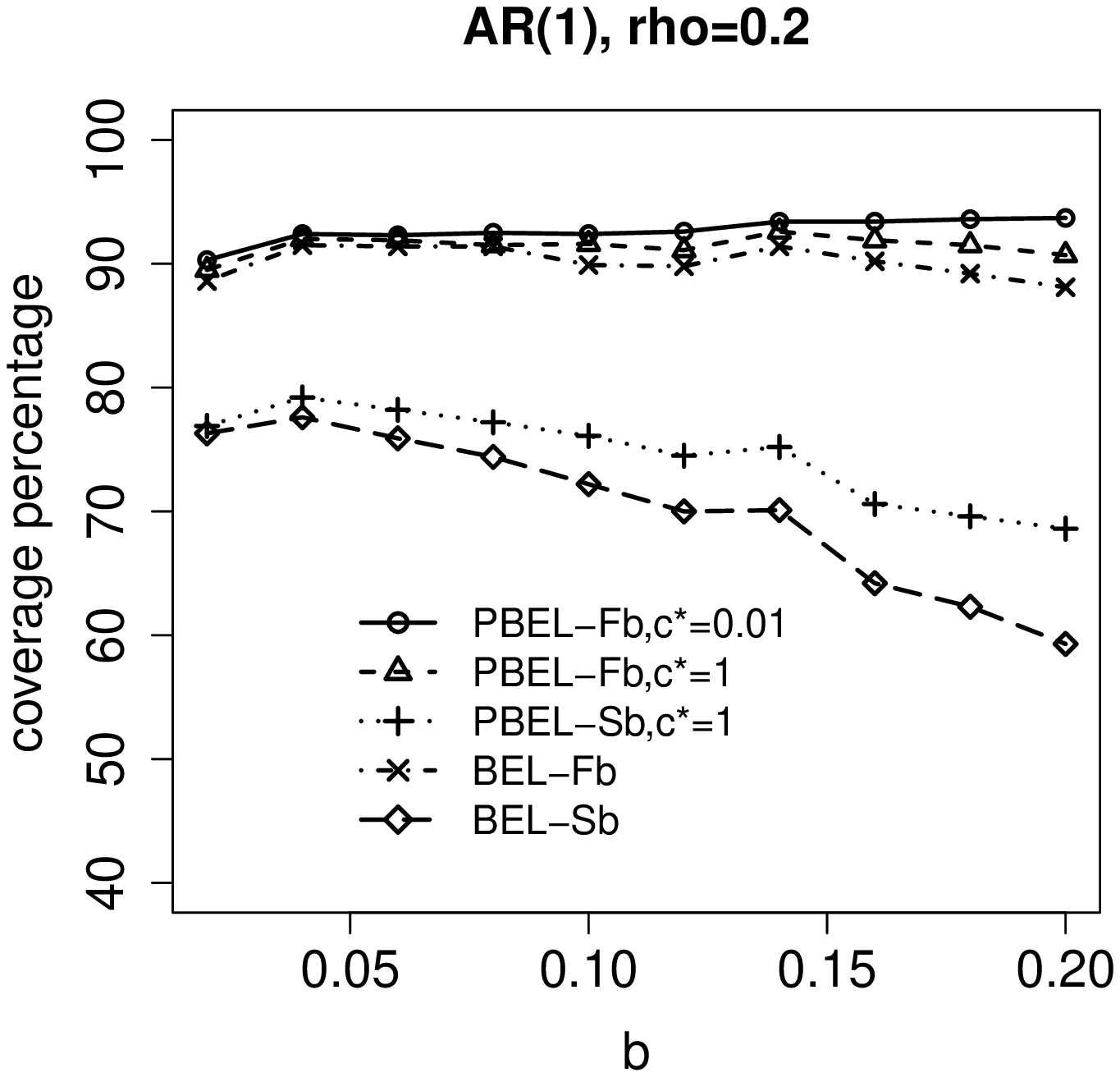}
\includegraphics[height=5.2cm,width=6cm]{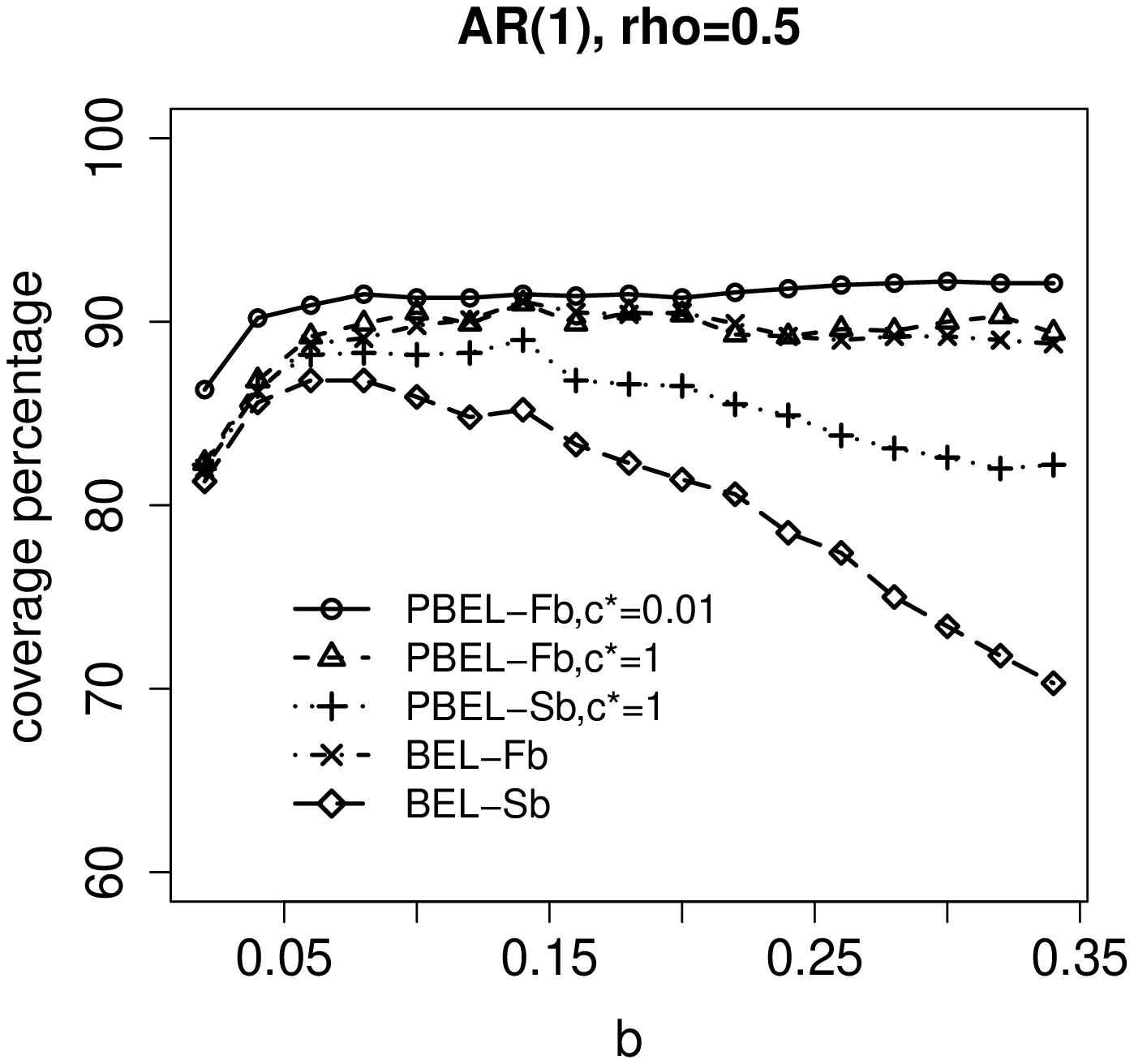}
\includegraphics[height=5.2cm,width=6cm]{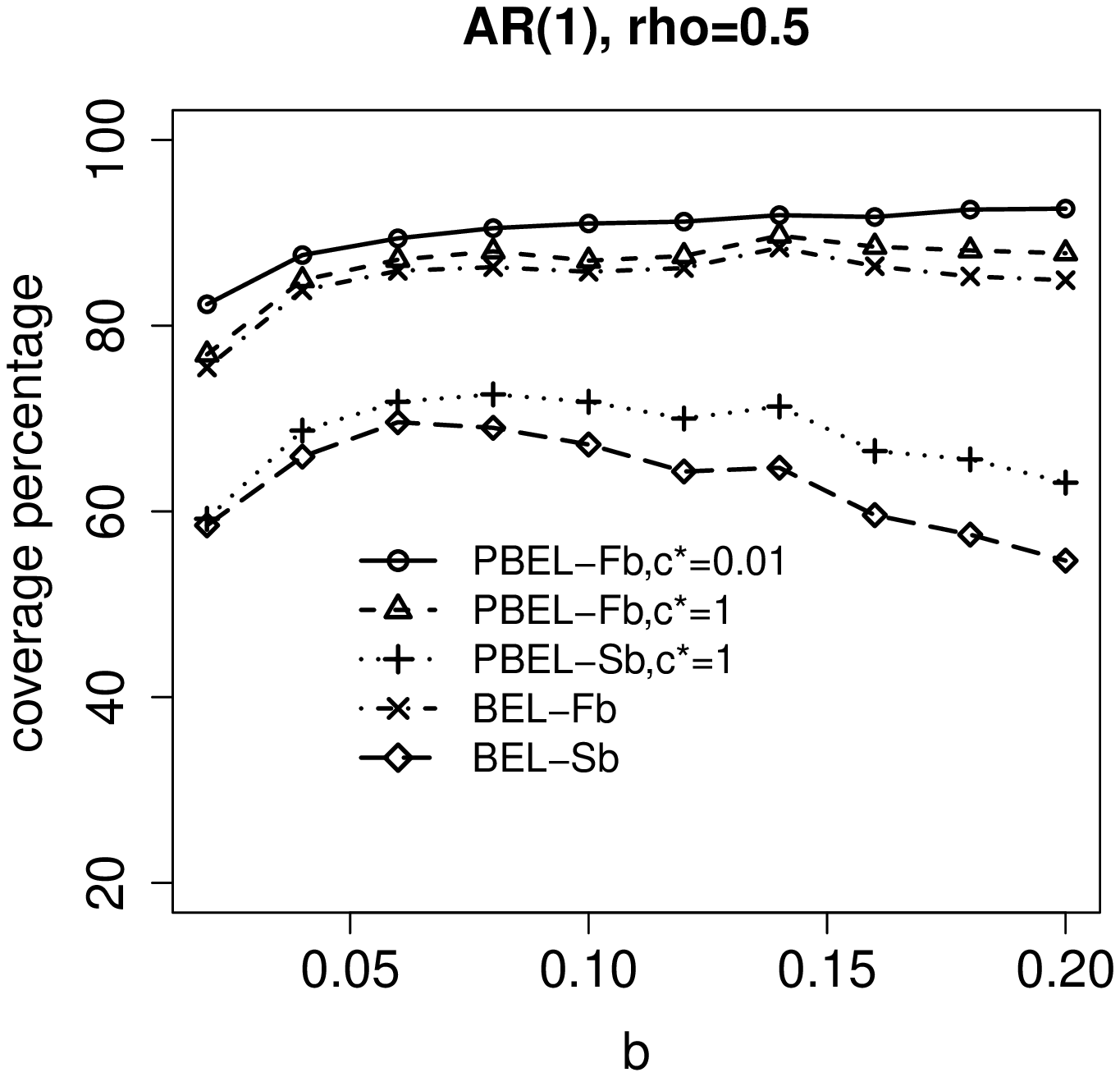}
\includegraphics[height=5.2cm,width=6cm]{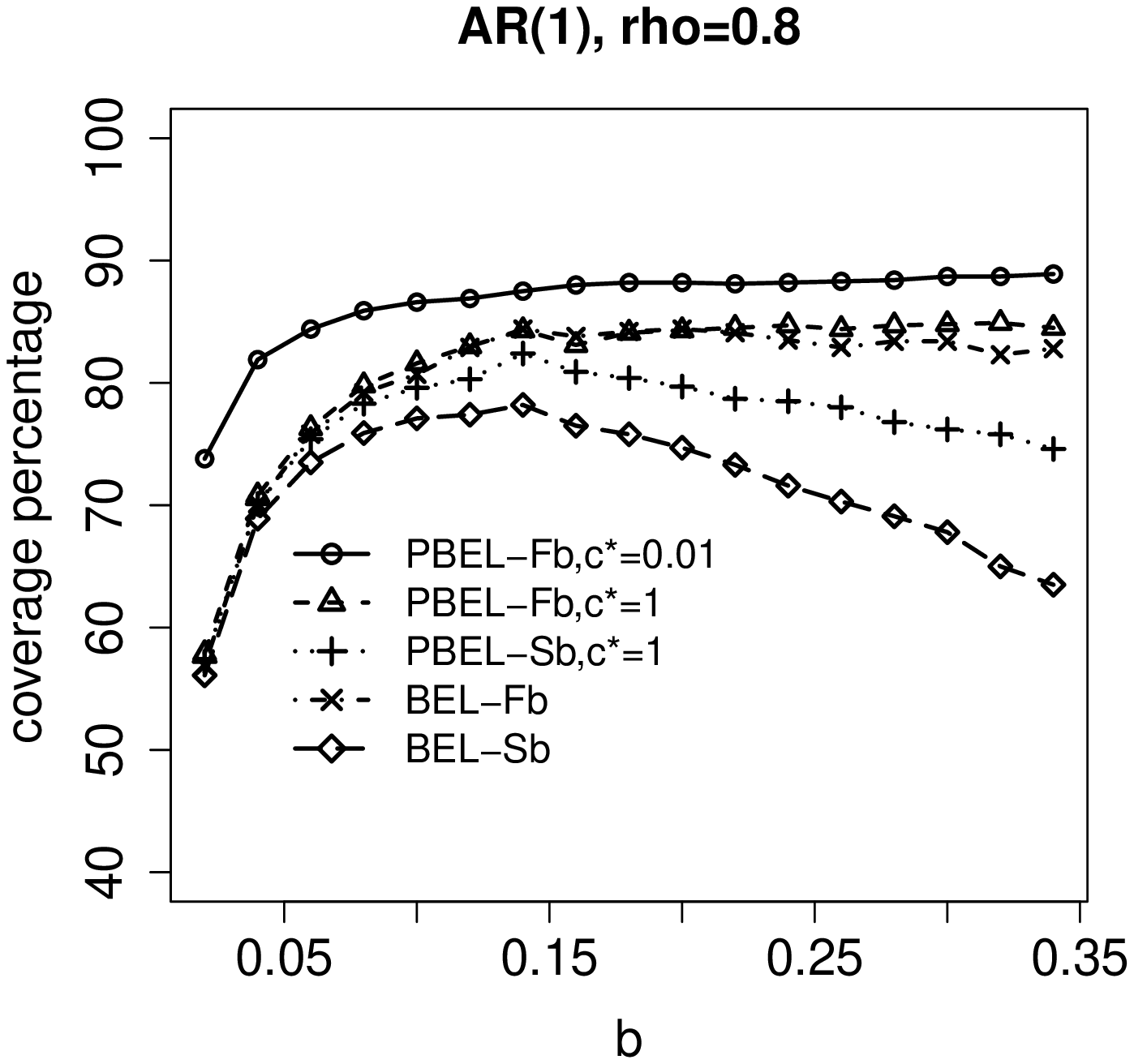}
\includegraphics[height=5.2cm,width=6cm]{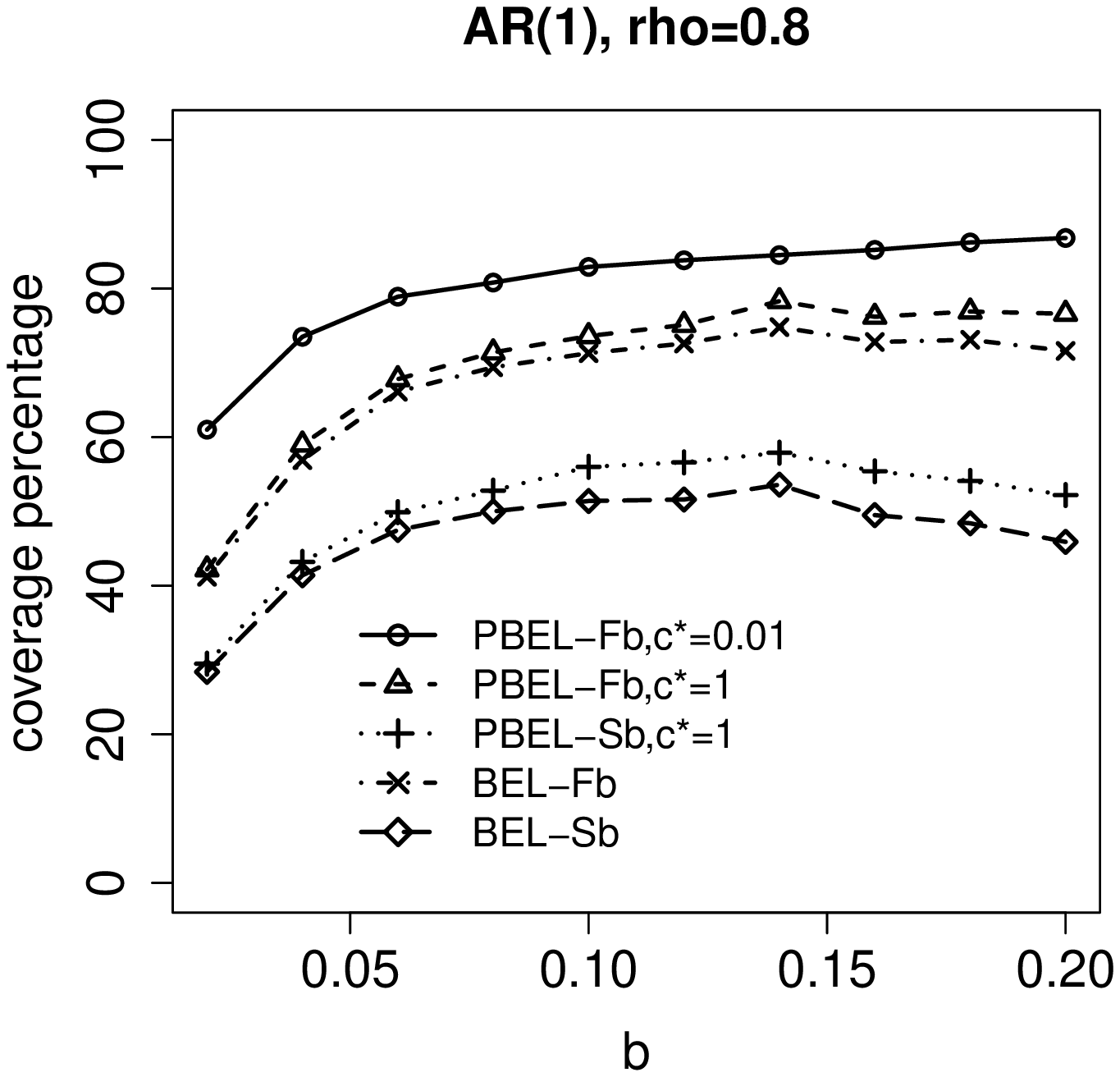}
\includegraphics[height=5.2cm,width=6cm]{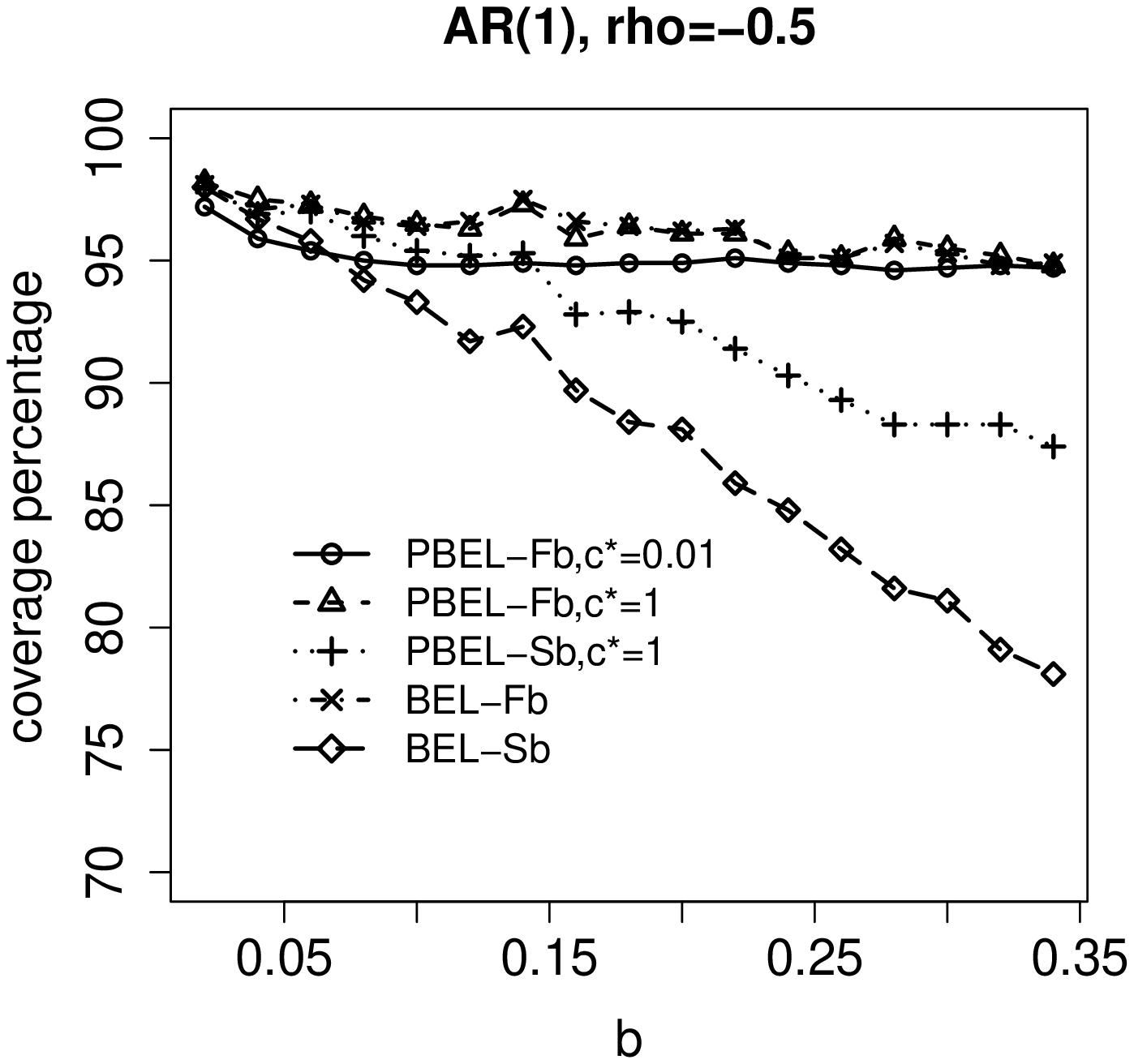}
\includegraphics[height=5.2cm,width=6cm]{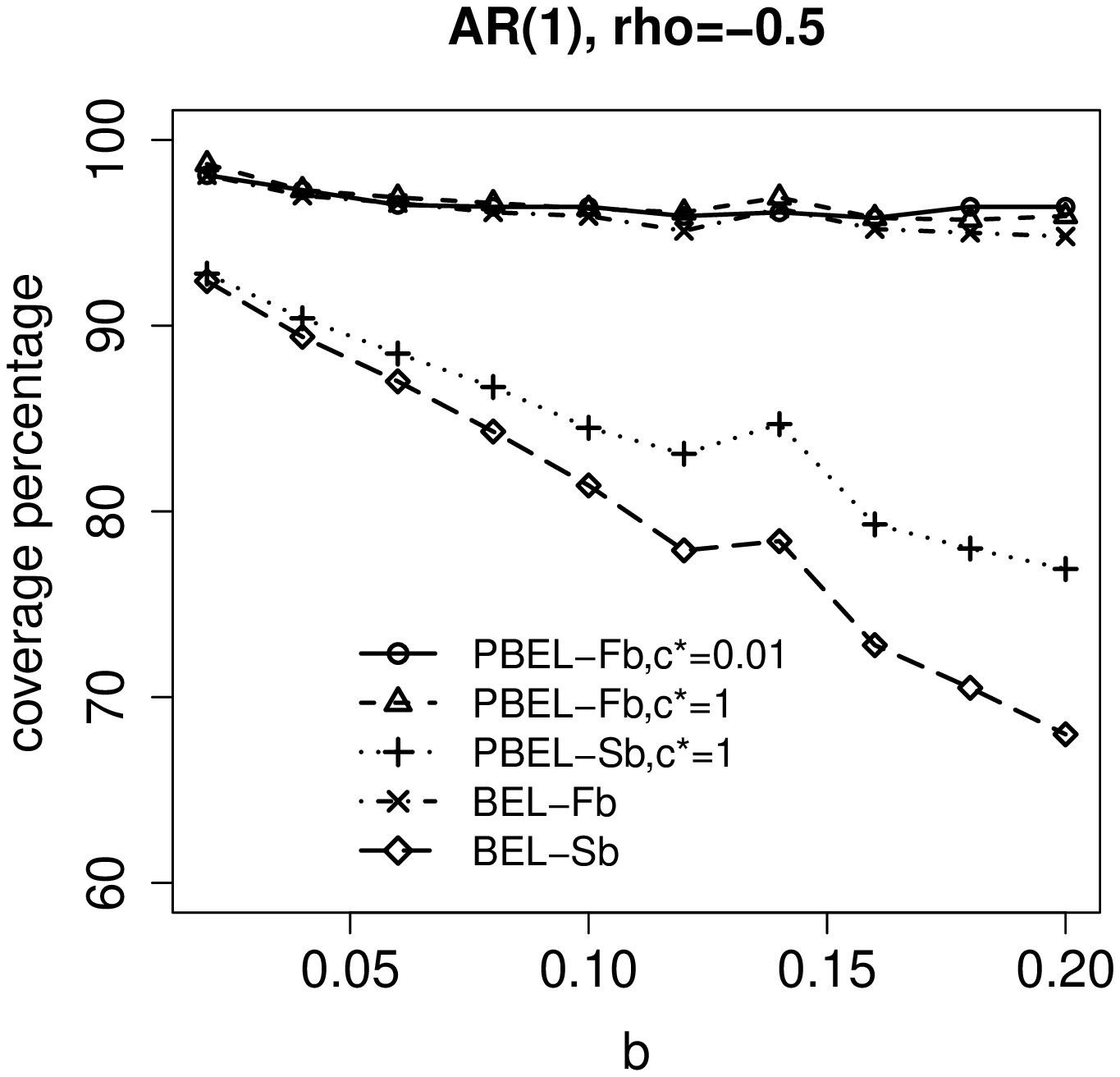}
\caption{Coverage probabilities for the mean delivered by the PBEL
with $Q(r,s)=(1-|r-s|)\mathbf{I}\{|r-s|\leq 1\}$, and BEL under both
small-$b$ and fixed-$b$ asymptotics, where $k=1$ for the left column
and $k=2$ for the right column. The nominal level is 95\% and the
number of Monte Carlo replications is 1,000.}\label{fig:pbel-add}
\end{figure}

\newpage

\begin{table}[H]
\caption{95\% quantiles of the limiting distributions for BEL
($u_{el,k}(b; 1-\alpha)$) and PBEL (with the Bartlett kernel) under
the fixed-$b$ asymptotics}\label{tab:crit-PBEL}
\begin{center}\footnotesize
\begin{tabular}{c rrrrrrrr}\toprule
& \multicolumn{2}{c}{BEL} & \multicolumn{2}{c}{PBEL, $c^*=0.01$} &
\multicolumn{2}{c}{PBEL, $c^*=0.2$} & \multicolumn{2}{c}{PBEL,
$c^*=2$}
\\ \cmidrule(r){2-3} \cmidrule(r){4-5} \cmidrule(r){6-7} \cmidrule(r){8-9 }
$b$ & $k=1$ & $k=2$ & $k=1$ & $k=2$ & $k=1$ & $k=2$  & $k=1$ & $k=2$
\\ \midrule
0.02&3.96&6.10&2.58&4.44&3.83&6.03&3.94&6.16\\
0.04&4.07&6.45&2.11&3.89&3.76&6.06&3.99&6.34\\
0.06&4.29&6.84&1.83&3.46&3.77&6.26&4.09&6.60\\
0.08&4.46&7.37&1.61&3.19&3.77&6.50&4.16&6.96\\
0.10&4.76&7.89&1.43&3.01&3.86&6.77&4.34&7.40\\
0.12&5.18&8.63&1.29&2.79&3.94&7.14&4.43&7.89\\
0.14&5.44&9.49&1.18&2.62&3.98&7.49&4.52&8.46\\
0.16&5.91&10.49&1.10&2.54&4.03&7.85&4.64&9.28\\
0.18&6.29&11.91&1.03&2.50&4.23&8.38&4.95&10.29\\
0.20&6.72&13.42&0.95&2.43&4.22&8.90&5.09&11.55\\
0.22&7.27&16.53&0.90&2.36&4.23&9.36&5.29&13.38\\
0.24&7.76&24.09&0.85&2.62&4.34&10.82&5.61&18.43\\
0.26&8.69&Inf&0.81&2.70&4.41&10.99&6.07&24.08\\
0.28&9.83&Inf&0.77&2.75&4.50&11.52&6.51&33.17\\
0.30&11.40&Inf&0.73&2.56&4.61&11.66&6.86&80.27\\
0.32&13.54&Inf&0.70&2.72&4.68&14.63&7.74&73.31\\
0.34&21.11&Inf&0.68&2.81&4.68&23.04&8.23&67.06\\
\midrule
\end{tabular}
\end{center}
\end{table}

\begin{table}[H]
\caption{95\% quantiles of the limiting distributions for PEBEL
(with the Bartlett kernel)}\label{tab:crit-EBEL}
\begin{center}\footnotesize
\begin{tabular}{c rr}\toprule
$c^*$ & $k=1$ & $k=2$ \\ \midrule
0.001&0.007&0.013\\
0.010&0.063&0.120\\
0.050&0.264&0.453\\
0.100&0.445&0.709\\
0.200&0.684&1.011\\
0.400&0.943&1.371\\
0.600&1.105&1.622\\
0.800&1.217&1.758\\
1.000&1.303&1.902\\
2.000&1.577&2.413\\

\midrule
\end{tabular}
\end{center}
\end{table}

\end{document}